\documentclass{article}
\usepackage[top=1in, bottom=1in, left=1in, right=1in]{geometry}

\usepackage{amsthm,amsmath,amssymb,amsfonts,mathrsfs}
\usepackage{url}

\usepackage{hyperref, cleveref}
\usepackage{multicol}

\usepackage{nicefrac,bbm}

\usepackage{enumitem,mathtools,pifont}
\usepackage{color,fancybox,graphicx,fullpage,subfigure}
\usepackage{scalerel}

\usepackage[T1]{fontenc}
\usepackage[boxruled]{algorithm2e}
\usepackage[framemethod=TikZ]{mdframed}

\newcommand{\nfrac}{\nicefrac}
\newcommand{\sfrac}[2]{#1/#2}

\newcommand{\poly}{\mathop{\rm poly}}

\newcommand{\supp}{\mathrm{supp}}

\newcommand{\argmax}{\operatornamewithlimits{argmax}}

\newcommand{\ltxlabel}{\ltx@label}

\newcommand{\cD}{\mathcal{D}}
\newcommand{\cE}{\mathcal{E}}
\newcommand{\cF}{\mathcal{F}}

\newcommand{\cP}{\mathcal{P}}

\newcommand{\wt}{\widetilde}

\newcommand{\eps}{\varepsilon}
\renewcommand{\epsilon}{\varepsilon}

\newcommand{\Z}{\mathbb Z}
\newcommand{\N}{\mathbb N}
\newcommand{\R}{\mathbb R}

\newcommand{\E}{\mathbb E}

\newcommand{\D}{\mathcal{D}}

\newcommand{\distN}{\mathcal{N}}

\newcommand{\Beta}{\text{\rm Beta}}

\newcommand{\inparen}[1]{\left(#1\right)}
\newcommand{\inbrace}[1]{\left\{#1\right\}}
\newcommand{\insquare}[1]{\left[#1\right]}

\newcommand{\sinparen}[1]{(#1)}
\newcommand{\sinbrace}[1]{\{#1\}}
\newcommand{\sinsquare}[1]{[#1]}

\newcommand{\expect}{\mathbb{E}}
\newcommand\av{\operatorname{\E}}

\newcommand{\given}{\;\middle|\;}

\newcommand{\ex}[1]{\av\insquare{#1}}
\newcommand{\exgiv}[2]{\av\insquare{#1 \given #2}}
\newcommand{\prob}[1]{\Pr \insquare{#1}}
\newcommand{\probgiv}[2]{\Pr \insquare{#1 \;\middle|\; #2}}
\newcommand{\yesnum}{\addtocounter{equation}{1}\tag{\theequation}}
\newcommand{\tagnum}[1]{\addtocounter{equation}{1}{\tag{#1; \theequation}}}

\makeatletter
\newcommand{\customlabel}[2]{%
\protected@write \@auxout {}{\string \newlabel {#1}{{#2}{\thepage}{#2}{#1}{}} }%
\hypertarget{#1}{}
}
\makeatother

\newtheorem{theorem}{Theorem}[section]
\newtheorem{lemma}[theorem]{Lemma}
\newtheorem{definition}[theorem]{Definition}

\newtheorem{proposition}[theorem]{Proposition}
\newtheorem{fact}[theorem]{Fact}
\newtheorem*{theorem*}{Theorem}
\newtheorem{remark}[theorem]{Remark}

\crefname{section}{Section}{Sections}
\crefname{theorem}{Theorem}{Theorems}
\crefname{observation}{Observation}{Observations}
\crefname{proposition}{Proposition}{Propositions}
\crefname{claim}{Claim}{Claims}
\crefname{condition}{Condition}{Conditions}
\crefname{example}{Example}{Examples}
\crefname{fact}{Fact}{Facts}
\crefname{lemma}{Lemma}{Lemmas}
\crefname{corollary}{Corollary}{Corollaries}
\crefname{definition}{Definition}{Definitions}
\crefname{remark}{Remark}{Remarks}

\SetAlgorithmName{Model}{dynamicalsystem}{List of Dynamics}
\SetKwInput{kwInit}{Initialize}
\SetKwFor{For}{For}{do}{end}
\SetKwFor{ForEach}{For each}{do}{end}

\newcommand{\sexp}[1]{{\hbox{\tiny$($}}#1{\hbox{\tiny$)$}}} %
\newcommand{\et}{{\sexp{t}}} %
\newcommand{\es}{{\sexp{s}}} %
\newcommand{\at}{a^\et}

\newcommand{\betat}{\beta^\et}
\newcommand{\deltat}{\delta^\et}
\newcommand{\ux}{U_{X}^\et}
\newcommand{\uy}{U_{Y}^\et}

\newcommand{\xit}{X_i^\et}
\newcommand{\yjt}{Y_j^\et}
\newcommand{\omt}{\omega}

\newcommand{\cH}{\mathcal{H}}

\newcommand{\util}{\textsc{Util}}
\newcommand{\negsp}{\hspace{-0.5mm}}
\newcommand{\tmpsum}{\text{$\textstyle\sum_{s=1}^t Z^\es$}}
\newcommand{\uxs}{U_X^\es}
\newcommand{\uys}{U_Y^\es}

\newcommand{\unconsgrp}{\textsc{ControlGrp}}
\newcommand{\rrgrp}{\textsc{RRGrp}}
\newcommand{\lb}{\ \underline{a}\ }

\newcommand{\zo}{\{0,1\}}

\newcommand{\commentalg}[1]{{\small\em\quad \textcolor{gray}{//{#1}}}\hspace{-3mm}}

\newcommand{\factor}{\ensuremath{\text{\footnotesize$\frac{b}{a^{\sexp{1}}-1}$}}}
\newcommand{\uniffactor}{\frac{16\ln{n}}{n(1-\rho)}\cdot\frac{a^{\sexp{1}}+b}{a^{\sexp{1}}-1}}
\newcommand{\uniffactorb}{\frac{16\ln{n}}{n(1-\rho)}\cdot\frac{a^{\sexp{1}}(a^{\sexp{1}}+b)}{a^{\sexp{1}}-1}}

\newcommand{\white}[1]{\textcolor{white}{#1}}

\mdfdefinestyle{MyFrame}{%
linecolor=black,
outerlinewidth=0.5pt,
roundcorner=5pt,
innertopmargin=5pt,
innerbottommargin=5pt,
innerrightmargin=10pt,
innerleftmargin=10pt,
backgroundcolor=white}

\mdfdefinestyle{FrameQuestion}{linecolor=white, outerlinewidth=0pt,
roundcorner=5pt,
innertopmargin=0pt, innerbottommargin=0pt,
innerrightmargin=20pt, innerleftmargin=20pt, backgroundcolor=white}

\mdfdefinestyle{FrameBox}{linecolor=black, outerlinewidth=0pt,
roundcorner=5pt,
innertopmargin=5pt, innerbottommargin=5pt,
innerrightmargin=5pt, innerleftmargin=5pt, backgroundcolor=white}

\mdfdefinestyle{FrameBox2}{linecolor=black, outerlinewidth=0pt,
roundcorner=5pt,
innertopmargin=0pt, innerbottommargin=5pt,
innerrightmargin=5pt, innerleftmargin=5pt, backgroundcolor=white}

\begin{document}

\title{The Effect of the Rooney Rule on Implicit Bias in the Long Term}

\author{L. Elisa Celis \\ Yale University \and Chris Hays \\ Yale University \and Anay Mehrotra \\ Yale University \and Nisheeth K. Vishnoi \\ Yale University}

\maketitle
\begin{abstract}
  A robust body of evidence demonstrates the adverse effects of implicit bias in various contexts, including hiring,  admissions, criminal justice, and healthcare.
  The Rooney Rule is a simple intervention developed to counter implicit bias in hiring, and has been  implemented in the private and public sector in various settings~\cite{duru2018rooney, collins2007tackling}.
  The Rooney Rule requires that a selection panel include at least one candidate from an underrepresented group in their shortlist of candidates.
  Recently, Kleinberg and Raghavan~\cite{Selection_Problems_in_the_Presence_of_Implicit_Bias} proposed a mathematical model of implicit bias and studied the effectiveness of the Rooney Rule when applied to a single selection decision.
  However, selection decisions often occur repeatedly over time; e.g., a software firm is continuously hiring employees, or a university makes admissions decisions every year.
  Further, it has been observed  that, given consistent counterstereotypical feedback, implicit biases against underrepresented candidates can change (e.g., \cite{dasgupta2008social}).

  In this paper, building on Kleinberg and Raghavan's model and work on opinion dynamics, we consider a model of how a selection panel's implicit bias changes over time given their hiring decisions either with or without the Rooney Rule in place.
  Our main result is that, for this model, when the selection panel is constrained by the Rooney Rule, their implicit bias roughly reduces at a rate that is inverse of the size of the shortlist---independent of the total number of candidates, whereas without the Rooney Rule, the rate is inversely proportional to the number of  candidates.
  Thus, our model predicts that when the  number of candidates is much larger than the size of the shortlist, the Rooney Rule enables a significantly faster reduction in implicit bias, providing additional reason in favor of instating it as a strategy to mitigate implicit bias.
  Towards empirically evaluating the long-term effect of the Rooney Rule in repeated selection decisions, we conduct an iterative candidate selection experiment on Amazon Mechanical Turk. We observe that, indeed, decision-makers subject to the Rooney Rule select more minority candidates \emph{in addition to} those required by the rule itself than they would if no rule is in effect, and in fact are able to do so without considerably \mbox{decreasing the utility of candidates selected.}
\end{abstract}

\newpage
\setcounter{tocdepth}{2}
\tableofcontents

\newpage

\section{Introduction}
Implicit bias is the unconscious association of certain qualities (or lack thereof) to individuals of socially salient groups, like those defined by race, gender, or sexuality.
In recent decades, a large body of experimental research has demonstrated the adverse effects of implicit bias in a wide range of contexts, including hiring
\cite{rooth2010automatic, ziegert2005employment, corinne2012science},
university admissions~\cite{capers2017implicit, posselt2016inside},
criminal justice~\cite{kaang2012implicit, hall2016black, bennet2010unraveling},
and healthcare~\cite{chapman2007sterotyping, green2007implicit, jimenez2010perioperative}.
In fact, even when decisions are based on quantifiable characteristics of the applicants, selection panels can systematically undervalue underrepresented candidates. %
For instance, it was found that women in managerial positions had to show roughly twice as much evidence to be seen as equally competent to men~\cite{williams2014double,lyness2006fit}, and evaluators, across jobs, unknowingly customized their evaluation criteria to favor the stereotypical gender~\cite{uhlmann2005constructed}. %

The Rooney Rule is a simple and widely adopted policy to counteract the adverse effects of implicit bias~\cite{collins2007tackling, waldstein2015success, reid2016rethinking}.
It requires that at least one among a shortlist of candidates (for further interviews or evaluation) picked by a selection panel must come from an underrepresented group.
It was originally instituted in the interview processes for hiring head coaches in the NFL in 2003, and since then,
has been adopted by various corporations such as, Amazon, Facebook, and Microsoft~\cite{passariello2016tech}, and in several public sector contexts in the US \cite{duru2018rooney, bland2017schumer}.
In fact, in 2020, the NFL broadened the Rooney Rule to require at that least two minority candidates in the shortlist of interviewees for head coaching positions be from an underrepresented group~\cite{young2020nfl}.
This motivates a generalization of the rule to the $\ell$-th order Rooney Rule (also proposed in \cite{Selection_Problems_in_the_Presence_of_Implicit_Bias}), which requires at least $\ell$ of the shortlisted candidates to be in the underrepresented group.

Although there is evidence that the Rooney Rule has had a positive impact in various contexts \cite{collins2007tackling, dickey2017lyft}, it is a subject of much  debate~\cite{collins2007tackling, waldstein2015success}.
Proponents of the policy argue that it counteracts the effects of implicit bias, while critics warn that it can lead to poorer selections.

Towards demonstrating the effectiveness of the Rooney Rule, Kleinberg and Raghavan ~\cite{Selection_Problems_in_the_Presence_of_Implicit_Bias} recently proposed a mathematical model of implicit bias and showed that the Rooney rule can improve the ``true'' utility of the selection panel in a single hiring decision.
More precisely, they consider $n$ candidates partitioned into two disjoint groups $G_X,G_Y\subseteq [n]$, where $G_X$ is the group of underrepresented candidates.
Each candidate has a true, {\em latent utility}, which is the value they would contribute if selected, and an {\em observed utility} which is the selection panel's (potentially biased) estimate of their latent utility.
They model the panel's implicit bias as a multiplicative factor $\beta\in [0,1]$, such that
the observed utility of underrepresented candidates (those in $G_X$) is $\beta$ times their latent utility, while the observed utility of all other candidates (those in $G_Y$) is the same as their latent utility.%
\footnote{To be precise, \cite{Selection_Problems_in_the_Presence_of_Implicit_Bias} consider $\beta\in (1,\infty)$, and assume that the observed utility of an underrepresented candidate is $\nfrac{1}{\beta}$ times their latent utility. Considering $\beta\in [0,1]$ is more convenient in our setting.}
Thus, if $\beta=1$, the panel evaluates underrepresented candidates without bias, and its bias against them becomes more severe as $\beta$ approaches 0.
The panel shortlists $k$ candidates (out of $n$) with the highest observed utility.
In the setting where $n$ is much larger than $k$, \cite{Selection_Problems_in_the_Presence_of_Implicit_Bias} characterize conditions on $\beta$, the proportion of underrepresented candidates ($\rho \coloneqq \frac{|G_X|}{n}$), and the distribution of latent utilities, such that under these conditions, applying the Rooney Rule (for $\ell=1$) increases the total latent utility of the shortlisted candidates.

An important benefit of the Rooney Rule is that the panel has the opportunity to closely evaluate qualified underrepresented candidates, see that their latent utility was greater than expected and learn to evaluate underrepresented candidates more accurately.
Indeed,	studies show that implicit biases can change over time~\cite{schuman1997racial, charlesworth2019patterns} and with changes in local-environments \cite{dasgupta2013implicit}.
In particular, it has been observed that exposure to other groups~\cite{dasgupta2008social, anderson2013imperative} and counterstereotypical evidence opposing the implicit beliefs~\cite{blair2001imagining, dasgupta2001malleability} can help reduce implicit bias.
Thus, one would hope that as the panel observes the latent utilities of more underrepresented candidates over multiple iterations of selection, its implicit biases would change.
This is in line with work on belief and opinion formation, which model how individuals update their beliefs and opinions based on the information they observe~\cite{acemoglu2011opinion, josang2001logic}. %
At a high-level, these works model the beliefs of individuals using  probability distributions, where, each time an individual receives new information, the distribution is updated to incorporate the new information and reflect the corresponding new beliefs~\cite{chazelle2019iterated, jadbabaie2012non}.\\

\subsection{Our contributions}
We consider a mathematical model for implicit bias and how it updates each selection decision. Under the assumptions of the model, the Rooney Rule provably enables a significantly faster reduction in implicit bias of the panel over multiple selection decisions when compared to the unconstrained condition; this gives a mathematical explanation for the aforementioned empirical observations. %

\medskip
Our model maintains a probability distribution over the implicit bias of the panel and updates this distribution after each iteration depending on the ratio of the latent utility of the shortlisted candidates and the observed utility of the shortlisted candidates; see \cref{sec:ourmodel}.
Technically, we show that, when the panel uses the $\ell$-th order Rooney Rule for $\ell \geq 1$, its implicit bias reduces, roughly, at the rate of $\frac{1}{(k-\ell+1)}$ --  independent of $n$ (\cref{thm:rooney_rule}), whereas, when the panel is not constrained by the Rooney Rule, then the rate at which its implicit bias reduces is, roughly,  $\frac{1}{n}$ (\cref{thm:no_rooney_rule}).
Thus, when the number of applicants $n$ is much larger than the size of the shortlist $k$, our model predicts that using the Rooney Rule leads to a significantly faster reduction in the panel's implicit bias.
Expanding on these results, we  characterize the effect of other parameters (such as the proportion of underrepresented candidates $\rho$) on the change in the panel's implicit bias over time (\cref{sec:predictions_of_model}).
We also discuss how our results generalize to other models, where the panel's implicit bias is drawn from distributions not in the beta family and updated using other rules (\cref{sec:technical_remarks}).
Thus, our theoretical results complement the work of Kleinberg and Raghavan \cite{Selection_Problems_in_the_Presence_of_Implicit_Bias} and provide an additional reason  to instate the Rooney Rule as a strategy to mitigate implicit bias.

\medskip

Towards empirically evaluating the effect of the Rooney Rule in repeated selections, we enlist participants on Amazon's Mechanical Turk to participate in an iterative selection experiment (Section~\ref{sec:empirical_results}).
We represented candidates from two different groups with different colored tiles and applied bias to the observed utilities of one of the groups.
In each iteration, participants were incentivized to maximize the latent utility of their selection, and the latent utilities of their selections were revealed after each round.
We observe that the participants subject to the Rooney Rule selected significantly more underrepresented candidates \emph{in addition to} those required by the rule itself than participants not subject to the rule, without substantially decreasing the utility of candidates selected.

\subsection{Related work}\label{sec:related_work}
\paragraph{Implicit bias.}
Studying implicit bias is a rich field in psychology~\cite{greenwald2006implicit, greenwald1995implicit} and
several works study the origins of implicit bias~\cite{payne2019historical, rudman2004sources}, its adverse effects~\cite{lyness2006fit, sadler2012world, williams2014double}, and its long-term trends~\cite{charlesworth2019patterns}.
We point the reader to the excellent treatise~\cite{kite2016psychology} for an overview of the field.

\medskip

\cite{Selection_Problems_in_the_Presence_of_Implicit_Bias} introduce a model for implicit bias, and under this model, characterized conditions where the Rooney Rule improves the latent utility of the selection.
Under the same model, \cite{celis2020interventions} study the ranking problem (a generalization of selection) under implicit bias, and propose simple constraints on rankings which improve the latent utility of the output ranking.
Both \cite{Selection_Problems_in_the_Presence_of_Implicit_Bias} and \cite{celis2020interventions} consider the latent utility in a single instance of the problem,
whereas, we are interested in how the implicit bias of the panel changes over multiple iterations. %

\medskip

\cite{EmelianovGGL20} study selection under a different model of bias: where the panel's observed utility has higher than average noise for underrepresented candidates. %
They consider a family of constraints, and show that, in their model, these constraints always increase the latent utility.
Unlike them, our goal is to understand the effects of constraints on the implicit bias of the panel.
In \cref{sec:technical_remarks}, we discuss how our results generalize when there is noise in $\betat$.
Accounting for other forms of noise in the observed utilities can be an interesting extension to this work.

\paragraph{Belief update models.}
Works on opinion dynamics and social learning study mathematical models of how people's beliefs change when they gain new information~\cite{chazelle2019iterated, jadbabaie2012non, acemoglu2011bayesian}.
Several works in this field represent beliefs by probability distributions and study simple rules, similar to the one we consider, to update these distributions~\cite{chazelle2019iterated, degroot1974reaching}. %
We refer the reader to \cite{acemoglu2011opinion} for a comprehensive overview of the field.

In a similar vein, the theory of subjective logic~\cite{josang2016subjective} mathematically models beliefs under uncertainty. %
A seminal work \cite{josang2001logic} gives a mapping from beliefs to a beta distribution $\Beta(a,b)$; roughly, $a$ is the evidence favouring the belief and $b$ evidence against it.
A canonical example is the statement: ``A ball drawn at random (from an urn of red and black balls) will be red''~\cite{josang2001logic}.
A person observes multiple draws from the urn, and after each draw updates their belief.
Our model uses the beta distribution to model the panel's implicit bias parameter and, in relevant contexts, can be viewed as the panel's ``belief'' in the following statement:
{\em the latent utilities of candidates from group $G_X$ and group $G_Y$ are identically distributed.}

\paragraph{Long-term impact.}
Several prior works have studied the long-term impacts of affirmative action policies on society~\cite{liu2018delayed, MouzannarOS19, HuC18shorttermintervention}.
\cite{liu2018delayed} consider how common fairness constraints in classification settings affect the underlying population over time.
\cite{MouzannarOS19} assumes individuals have binary utilities, (qualified or unqualified), and they give asymptotic results for a broad-set of dynamics depending on how the fraction of qualified individuals in each group changes. %
\cite{HuC18shorttermintervention} studies a dynamics in the context of a labor market.

\smallskip
In contrast, we allow for non-binary utilities, study the effect on the panel's implicit bias, and give non-asymptotic results. %

\paragraph{Iterated learning experiments.} %
In a classic formulation of a \textit{function learning experiment}, %
in each iteration, participants are given a numeric input and asked to predict its numeric output from some examples or using knowledge they accumulated so far.
Several experiments in cognitive science~\cite{busemeyer1997learning,lucas2015rational,koh2015rational} and behavioral economics~\cite{coutts2019news,holt2009update} use function learning experiments to study human performance on prediction tasks with incomplete information. %
Related to our work, these techniques have also been used to (attempt to) measure implicit biases~\cite{lindstrom2014racial}.
We refer the reader to~\cite{busemeyer1997learning} for an overview of the experimental work in this topic.

\smallskip
Our empirical experiment builds on the classic iterated learning experimental design---there is a simple linear relationship between the observed (input) and latent utilities (output). %
However, we do not ask for participants to explicitly predict the output; instead, they implicitly do so by selecting the observed utilities which they they predict will have the highest latent utilities (output). %

\subsection*{Notation}

For a natural number $n\in \N$ by $[n]$ we denote the set $\{1,2,\dots,n\}$, and for a real number $x\hspace{-0.5mm}\in\hspace{-0.5mm} \R$ by $\exp(x)$ we denote $e^x$.
We use calligraphic letters such as $\D$ and $\cP$ to denote  {\em distributions}, and $X \sim \D$ denotes a sample $X$ drawn from $\D$.
For a distribution $\cD$ its support is the set $$\{x\colon \Pr_{X\sim\cD}[X=x]>0 \},$$ we denote the support of $\cD$ by $\supp(\cD)$.
If the support of a distribution is an interval over the reals, we say that the distribution is continuous.
We use $(t)$ in the superscript to indicate  the $t$-th iteration.
We use the subscript $i$ to index the underrepresented candidates and $j$ to index all other candidates.
We use $\Beta(a,b)$ to denote the beta distribution with parameters $a$ and $b$.
It holds that
$$\expect_{\beta\sim \Beta(a,b)}\sinsquare{\beta}=\frac{a}{(a+b)}.$$
Formally, given $a,b\geq 0$, define $\Beta(a,b)$ to be the distribution with the following cumulative density function: for all $x\in [0,1]$
$$\Pr_{\beta\sim \Beta(a,b)}[\beta\leq x] \coloneqq \frac{\int_0^x y^{a-1} (1-y)^{b-1}dy}{\int_0^1 y^{a-1} (1-y)^{b-1}dy}.$$

\medskip
\section{Model}\label{sec:ourmodel}
In each round of selection, there are $n$ candidates, and a selection panel shortlists $k$ of them.
The candidates are partitioned into two disjoint groups $G_X,G_Y\subseteq [n]$, where $G_X$ denotes the group of underrepresented candidates. %
The intersection $G_X\cap G_Y$ is empty and $G_X\cup G_Y = [n]$.
We call the candidates in $G_X$ the {\em $X$-candidates} and those in $G_Y$ the {\em $Y$-candidates}.
Each candidate has a true or {\em latent utility} which is the value which a candidate would contribute if selected.
Denote this utility by $X_i\geq 0$ for the $i$-th $X$-candidate, and by $Y_j\geq 0$ for the $j$-th $Y$-candidate.
We assume that the latent utilities of all candidates are independently and identically (i.i.d.) drawn in each iteration from some continuous distribution $\cP$. %
We assume that $\cP$ has non-negative and bounded support.
Let $\rho$ be the fraction of the underrepresented candidates: %
$$\rho\coloneqq \frac{|G_X|}{|G_X|+|G_Y|}.$$
{While the utilities of individual candidates change with time, we assume that $\rho$ itself does not change and, as a consequence, under our assumption $G_X$ and $G_Y$ also do not change.}

\subsection[Implicit bias model]{Implicit bias model of \cite{Selection_Problems_in_the_Presence_of_Implicit_Bias}}
In \cite{Selection_Problems_in_the_Presence_of_Implicit_Bias}, based on the empirical observations of \cite{wenneras2001nepotism},
the setting where the panel does not observe  latent utilities and instead sees an {\em observed utility}, which is its (possibly biased) estimate of the latent utilities is considered. %
They consider the following model of observed utilities parameterized
by an implicit bias parameter $\beta\in [0,1]$: define the observed utilities of an $X$-candidate $i\in G_X$  as
$$\wt{X}_i \coloneqq \beta\cdot X_i.$$
The observed utility of a $Y$-candidate is assumed to be the same as its latent utility.
Notice that in the above definition, if $\beta = 1$, then $\wt{X}_i=X_i$ for all $i\in G_X$ and the panel evaluates $X$-candidates without bias.
It is sometimes useful to define the following vectors: $X\coloneqq (\dots, X_i, \dots)$,\ $\wt{X}\coloneqq (\dots, \wt{X}_i,\dots )$ and $Y\coloneqq (\dots, Y_j,\dots )$, where $i$ varies \mbox{over $G_X$ and $j$ varies over $G_Y$.}

\subsection{Candidate selection problems and the Rooney Rule}
The utility of a subset of candidates $S$ is defined as the sum of the utilities of all candidates in $S$.
Given a subset $S\subseteq[n]$, define its {\em total observed utility} as
\begin{align}
  \util{}(S,\wt{X},Y) \coloneqq \sum_{i\in S\cap G_X}\hspace{-0mm} \wt{X}_i + \sum_{j\in S\cap G_Y}\hspace{-0mm} Y_i.\label{eq:total_observed_utility} %
\end{align}
Similarly, define the {\em total latent utility} of $S\subseteq [n]$ as the sum of the latent utilities of all candidates in $S$:  $\util{}(S,{X},Y)$ (where we replace $\wt{X}$ in Equation~\eqref{eq:total_observed_utility} by ${X}$).

As in	\cite{Selection_Problems_in_the_Presence_of_Implicit_Bias}, we assume that	the selection panel selects a subset of candidates $S$ of size $k$ which maximizes total observed utility:
\begin{align}
  S\coloneqq \argmax_{T\subseteq[n]\colon |T|=k} \util{}(T,\wt{X},Y).\label{eq:opt_subset_without_rooney_rule}
\end{align}
Note that when $\beta<1$, the set $S$ may have  few $X$-candidates, disadvantaging those candidates. %

The $\ell$-th order Rooney Rule tries to address this by requiring the panel to select at least $\ell$ $X$-candidates. (Note that if $\ell=1$, then the $\ell$-th order Rooney Rule is \mbox{the same as the usual Rooney Rule). Let}
\begin{align}
  \mathcal{R}(\ell) \coloneqq \inbrace{T\subseteq [n] \colon |T\cap G_X| \geq \ell\ \text{and}\ |T| = k}
\end{align} \label{eq:satisfying_subsets}

\noindent
be the set of all subsets of size $k$ satisfying the $\ell$-th order Rooney Rule.
The panel constrained by the $\ell$-th order Rooney Rule picks a subset $S_\ell \in \mathcal{R}(\ell)$ satisfying the rule which maximizes the total observed utility:
\begin{align}
  S_\ell \coloneqq \argmax_{T\in \mathcal{R}(\ell)}\ \util{}(T,\wt{X},Y).\label{eq:selection_with_const}
\end{align}
Notice that the set $S$ selected by the panel without the the Rooney Rule (in Equation~\eqref{eq:opt_subset_without_rooney_rule}) is the same as the set $S_0$ above.

As we discuss next, the panel updates its beliefs as a function of the total latent utility and the total observed utility of $S_\ell$.
To simplify the notation there, let the $U$ and $\wt{U}$ be the total latent utility and the total observed utility of $S_\ell$:
\begin{align}
  U\coloneqq \util(S_\ell, X, Y),\\
  \wt{U}\coloneqq \util(S_\ell, \wt{X}, Y).
\end{align}
Note that, before selecting the candidates, the panel knows $\wt{U}$, but does not know $U$.
Our implicit bias update model assumes that the panel sees $U$ {\em after} selecting the candidates.
{The rationale is that after the panel selects the candidates $S_\ell\subseteq[n]$ and observes their actual performance---at the job or in an interview---it can better estimate their latent utility $U$.}

\subsection{Implicit bias update model}
Inspired by the works \cite{josang2001logic,vskoric2016flow} that give  mappings from beliefs to beta distributions, in each iteration,		we model the implicit bias $\beta$ as a draw from the a beta distribution $\Beta(a,b)$, where $a,b> 1$ and, roughly, $a$ is the evidence favouring the belief and $b$ evidence against it.

In our setting, $a$ is roughly proportional to the belief that $X$-candidates have the same latent utility as $Y$-candidates, and $b$ is roughly proportional to bias against it. %
Notice that the larger $a$ is, the closer $\expect[\beta]$ is to 1 (no bias), and the larger $b$ is, the closer $\expect[\beta]$ is to 0 (largest bias).\footnote{This follows since the expected value of a $\Beta(a,b)$ random variable is $\nfrac{a}{(a+b)}$.}

If $U>\wt{U}$, then the panel has evidence that the $X$-candidates performed better than expected.
In this case, the panel's implicit bias reduces, i.e., $a$ would increase.
Since $\beta\in [0,1]$, we can argue that $U \geq \wt{U}$.
To see this, note that
$$(U-\wt{U})=(1-\beta)\cdot\sum_{i\in S_\ell\cap G_X}X_i,$$ and since $1-\beta\geq 0$ and $X_i\geq 0$ for all $i$, we have that
$$(U-\wt{U})\geq 0.$$
Thus, the panel does not receive evidence to support it's bias, and so, $b$ is a constant in this model.

To summarize, given parameter $a> 1$ which varies over iterations and some a fixed parameter $b>1$, \mbox{we consider the distribution} $$\cD(a)\coloneqq\Beta(a,b),$$ and draw the panel's implicit bias $\beta$ from $\cD(a)$.

Since we consider multiple iterations of the above model of candidate selection, we need to specify how $\beta$ evolves.
Let a superscript $(t)$ on a variable indicate the variable's value at the $t$-th iteration.
We start with $a^{(1)}$ to be some fixed number greater than $1$.
Suppose that in the $t$-th iteration, the panel selects a subset $S^{\sexp{t}}_\ell$ whose total latent utility is $U^{\sexp{t}}$ and total observed utility is $\wt{U}^{\sexp{t}}$.
We propose and study the following  update rule:
\begin{align*}
  a^{\sexp{t+1}} &\coloneqq \frac{U^{\sexp{t}}}{\wt{U}^{\sexp{t}}} \cdot \at,\tagnum{Update rule}\customlabel{eq:update_rule}{\theequation}\\
  \beta^{\sexp{t+1}} &\sim \cD(a^{\sexp{t+1}}).
\end{align*}

\noindent
We summarize the complete mathematical model in  Model~\ref{alg:iterative_selection_and_panel_learning}.

\setlength{\algomargin}{0.5em}
\begin{algorithm}[h!]
  \SetAlgoNoEnd
  \caption{Our implicit bias update model}
  \label{alg:iterative_selection_and_panel_learning}
  \kwInit{The Rooney Rule parameter $\ell\in \Z_{\geq 0}$, a parameter $a^{\sexp{1}}>1$, and a constant $b>1$.}\vspace{2mm} %
  \For{$t=1,2,\dots$}{
  \vspace{1mm}
  {\bf Sample} $\sinbrace{Y_j^{\sexp{t}}}_{j\in G_Y}$ and $\sinbrace{X_i^{\sexp{t}}}_{i\in G_X}$ i.i.d. from $\cP$.\vspace{1mm}

  {\bf Sample} $\betat$ from $\D(a^{\sexp{t}})\coloneqq\Beta\sinparen{a^{\sexp{t}}, b}$.\hspace{12em} \commentalg{Implicit bias} \vspace{1mm}

  {\bf Define} $\wt{X}_i^{\sexp{t}}  \coloneqq \betat\cdot X_i^{\sexp{t}}$ for all $i\in G_X$.\hspace{13em}\commentalg{Observed utilities}\vspace{1mm}

  {\bf Select} $S^{\sexp{t}}_\ell\coloneqq \argmax_{S\in \mathcal{R}(\ell)}\  \util{}(S, \wt{X}^{\sexp{t}}, {Y}^{\sexp{t}})$.\hspace{9.4em}\commentalg{Shortlist}\vspace{1mm}\\
  {\bf Let} $U^{\sexp{t}} \coloneqq \util{}(S^{\sexp{t}}_\ell,X,Y)$ and $\wt{U}^{\sexp{t}}\coloneqq \util{}(S^{\sexp{t}}_\ell,\wt{X},{Y})$.\vspace{1mm}

  {\bf Update}  $a^{\sexp{t+1}} \coloneqq \inparen{\frac{U^{\sexp{t}}}{\wt{U}^{\sexp{t}}}} \cdot  a^{\sexp{t}}$.\hspace{15.9em}\commentalg{Update implicit bias}\vspace{1mm}
  }
\end{algorithm}

\subsection{Discussion  of the model}

While	we  consider a simple multiplicative update rule,
one could also study other  update rules.
For instance, some prior works on opinion dynamics study additive updates and updates where the new value is a convex combination of the current value and observation~\cite{degroot1974reaching, krause2000discrete}.
We discuss how our results generalize to other update rules in \cref{sec:generalizing_for_other_update_rules}.

Note that our model assumes that the panel evaluates the latent utility of all selected candidates together, and so observes $U^{\sexp{t}}$ and $\wt{U}^{\sexp{t}}$ and not the latent utilities of individual candidates; we discuss other scenarios in \cref{sec:limitations}.

The update rule~\eqref{eq:update_rule} in our model is only a function of the current belief (i.e., $\cD^{\sexp{t}}$) and observation (i.e., $\frac{U^{\sexp{t}}}{\wt{U}^{\sexp{t}}}$).
This is equivalent to the assumption made by the Bayesian-Without-Recall model~\cite{rahimian2016learning, chazelle2019iterated}.
(Note, however, that the update in Equation~\eqref{eq:update_rule} is not Bayesian).
Our update rule also bears resemblance to the multiplicative weights update method \cite{AHK} that has been successfully used to explain learning in social groups \cite{CelisKV17} and sexual evolution \cite{PNAS2:Chastain16062014}.

While our update rule does encode past observations, considering more complicated rules which explicitly include all past observations can lead to interesting extensions of this work.

Our model incorporates uncertainty in the panel's implicit bias at each iteration by drawing from a distribution, but assumes that the utilities do not have any noise.
It would be interesting to consider models which incorporate noise in the observations of the committee as well. %

\section{Theoretical results}\label{sec:theoretical_results}
Our first result considers the case when the panel uses the Rooney Rule.
It shows that, under the implicit bias update rule~\eqref{alg:iterative_selection_and_panel_learning}, when the panel uses the Rooney Rule, the reduction in its implicit bias  is independent of $n$.
\begin{theorem}[\textbf{Fast learning with the Rooney Rule}]\label{thm:rooney_rule}
  Under the implicit bias update rule~\eqref{alg:iterative_selection_and_panel_learning},
  given parameters $a^{\sexp{1}},b > 1$ there exists a constant $C_1>0$, such that,
  for all iterations $t\in \N$, population size $n\in \N$, number of candidates selected $k\in [n]$, Rooney Rule parameter $\ell\in [k]$, ratio of the underrepresented candidates $\rho\in (0,1)$,
  and continuous and bounded distribution of latent utility $\cP$,
  when the panel is constrained by the $\ell$-th order Rooney Rule, then
  \begin{align}
    \label{eq:lowerbound_on_bias}\expect\big[\beta^\et\big]
    &\geq  \inparen{1 - \frac{1-\expect[\beta^{\sexp{1}}]}{1+C_1 \cdot \frac{t\rho}{(k-\ell+1)}}}\cdot\inparen{1-e^{-t\rho/16}}
  \end{align}
  where the expectation is over the draws of implicit bias $\beta^{\sexp{s}}$ and latent utilities of candidates in all previous iterations $s\in [t-1]$.
\end{theorem}

\noindent
Note that the lower bound in Equation~\eqref{eq:lowerbound_on_bias} is independent of $n$ and, very roughly, scales as $\frac{\rho}{k-\ell+1}.$
At a high-level, this is because the Rooney Rule ensures that the panel selects at least $\frac{\ell}{k}$ fraction of the shortlist from $G_X$, and this fraction does not vary with $n$.
As we show in our next result, this independence doesn't hold if the panel does not use the Rooney Rule ($\ell=0$).

\cref{thm:rooney_rule} also shows that if the panel uses the $\ell$-th order Rooney Rule, in the long term (i.e., for a large $t$),
its expected implicit bias is almost 1 (no bias) independent of the number of candidates $n\in \N$.
To see this, let $t$  be a large value compared to $\frac{k-\ell+1}{\rho}$, then notice that both terms containing $t$ in Equation~\eqref{eq:lowerbound_on_bias} are small and do not depend on $n$.
Our next result considers the case when the panel does not use the Rooney Rule.
It assumes that $\cP$ is the uniform distribution on $[0,1]$ (denoted by ${\rm Unif}(0,1)$); we discuss how the result extends to other distributions in \cref{sec:thm:no_rooney_rule_gen}.
\begin{theorem}[\textbf{Slow learning without the Rooney Rule}]\label{thm:no_rooney_rule}
  Under the implicit bias update rule~\eqref{alg:iterative_selection_and_panel_learning},
  given parameters $a^{\sexp{1}},b > 1$,
  for all iterations $t\in \N$, number of candidates selected $k\in[n]$, and ratio of the underrepresented candidates $\rho\in (0,1)$,
  when the panel is not constrained by the Rooney Rule (i.e., $\ell=0$) and $\cP$ is ${\rm Unif}(0,1)$,
  there exists a $n_0\in \N$ and $C_2>0$, such that for all $n\geq n_0$
  \begin{align}
    \expect\big[\beta^\et\big] &\leq \expect\big[\beta^{\sexp{1}}\big] + C_2\cdot\frac{t\ln{n}}{n(1-\rho)},\label{eq:upperbound_on_bias}
  \end{align}
  where the expectation is over draws of implicit bias $\beta^{\sexp{s}}$ and latent utilities of candidates in all previous iterations $s\in [t-1]$.
\end{theorem}

\noindent
Thus, when the panel is not constrained by the Rooney Rule, the rate at which its implicit bias reduces (approaches 1) is no more than, roughly, $\frac{\ln n}{n(1-\rho)}.$
Further, from \cref{thm:no_rooney_rule} we can infer that for any fixed  $t\in \N$ there is a large enough $n$, such that panel's expected implicit bias after $t$ iterations is almost the same as its initial implicit bias: $\expect\sinsquare{\beta^{\sexp{1}}}$.

Thus, together with Theorem \ref{thm:rooney_rule}, Theorem \ref{thm:no_rooney_rule} shows a significant difference in the panel's implicit bias at a given iteration with and without the Rooney Rule when $n$ is much larger than $k$, providing evidence to support implementing the Rooney Rule.
\noindent We note that both \cref{thm:rooney_rule} and \cref{thm:no_rooney_rule} hold for any value of $t=1,2,\dots;$ and not just in the limit when $t$ is large.
The upper bound in \cref{thm:no_rooney_rule} increases (becomes weaker) as $t$ increases.
This is necessary as we can show that for a fixed $n$ as $t\to\infty$, the panel  becomes unbiased (see Lemma~\ref{thm:convergestoone} for a formal statement).

Finally,    we remark that in the statements of our theorems, we have tried to capture the dependence on each parameter as cleanly as possible and not tried to optimize the constants.
\subsection{{Qualitative predictions}}\label{sec:predictions_of_model}
We now summarize the qualitative effects of the parameters of our model on the panel's expected implicit bias.
\subsubsection{Effect of $n$}
As  discussed above, when the panel does not use the Rooney Rule,  the rate at which its implicit bias decreases slows down with $n$ (\cref{thm:no_rooney_rule}):
\begin{align*}
  \hspace{-13mm}\ell=0\colon\qquad \text{bias reduction rate} \propto \frac1n.
\end{align*}
Intuitively, this is because, as the number of available $Y$-candidates increases, a biased panel selects fewer $X$-candidates.
Thus, the panel becomes less likely to observe that $X$-candidates perform better than expected.
In contrast, when the panel uses the Rooney Rule, its expected implicit bias is independent of $n$ (\cref{thm:rooney_rule}):
\begin{align*}
  \hspace{8.5mm}\ell\geq 1\colon\qquad \text{bias reduction rate is independent of $n$}.
\end{align*}

\subsubsection{Effect of $\ell$}
When the panel uses the Rooney Rule,  increasing $\ell$ (fixing other parameters), increases its learning (\cref{thm:rooney_rule}):
\begin{align*}
  \hspace{-12mm}\ell\geq 1&\colon\qquad \text{bias reduction rate} \propto \ell.
\end{align*}
This is intuitive, as increasing $\ell$ leads the panel to select a larger number of $X$-candidates, and hence, is more likely to observe that they perform better than they expected.
Mathematically, this is because the update, $\frac{\wt{U}^\et}{U^\et}$, is an increasing function of the total utility of the $X$-candidates selected.
And the total utility of the $X$-candidates increases on selecting more candidates.
One might think that this means setting $\ell=k$ is the optimal choice.
However, in the real world, reducing the panel's bias is not the only objective---and the choice of $\ell$ should also consider other factors.

\subsubsection{Effect of $k$}
Increasing $k$, holding everything else fixed, decreases the panel's expected implicit bias at any given iteration (\cref{thm:rooney_rule}):
\begin{align*}
  \hspace{-13mm}\ell\geq 1\colon\qquad \text{bias reduction rate} \propto \frac{1}{k}.
\end{align*}

\subsubsection{Effect of $\rho$}
Both the lower bound~\eqref{eq:lowerbound_on_bias} and upper bound~\eqref{eq:upperbound_on_bias} are increasing functions of $\rho$.
We can infer that in both cases increasing $\rho$ (fixing other parameters) increases $\expect[\betat]$:
\begin{align*}
  \hspace{-11mm}\ell\geq 0&\colon\qquad \text{bias reduction rate} \propto \rho.
\end{align*}
Intuitively, increasing $\rho$ increases the fraction of $X$-candidates, and makes it more likely for the panel to select an $X$-candidate as the number of positive outliers (whose observed utility appears to be good \emph{despite} the bias) increases.

\subsection{Generalizations of our theoretical results}\label{sec:technical_remarks}
\subsubsection{Generalizing to other distributions of implicit bias}\label{sec:generalization_for_other_distributions}
In this section, we show how our results generalize when $\cD(a)$ is not a beta distribution.
This can be interesting, for instance, if the panel's implicit biases do not follow the beta family, or when the draws of $\betat$ are noisy.

\paragraph{Notation.} Let $\cD(a)$ be any continuous distribution supported on $[0,1]$ and parameterized by $a>1$. %
Define $\Phi(x)$ as the expected value of $\beta$ drawn from $\cD(x)$:
$$\Phi(x)\coloneqq \expect_{\beta\sim \cD(x)}\sinsquare{\beta},$$
and ${\rm median}(a)$ as the median of $\beta\sim\cD(a)$.
For example, we can consider $\cD(a)$ to be the truncated normal distribution with mean $a$ and fixed variance, in which case, ${\rm median}(a)=a$.\\[-1mm]

\noindent Assume that
\begin{enumerate}[leftmargin=*]
  \item $\Phi(\cdot)$ is an increasing function,
  \item $\Phi(\cdot)$ is concave,
  \item $\expect_{\beta\sim\cD(a)}\insquare{\frac{1}{\beta}}$ is decreasing in $a$ and finite for every $a>1$, and
  \item there is a constant $C_3>0$, such that, for all $a\negsp >\negsp 1,\ k\negsp \in \negsp [n], \ell\negsp \in\negsp  [k]$
  \begin{align*}
    \frac{1-{\rm median}(a)}{{\rm median}(a)+(k-\ell)}>\frac{C_3}{a \cdot (k-\ell+1)}.%
  \end{align*}
\end{enumerate}

\noindent We prove the following versions of \cref{thm:rooney_rule} and \cref{thm:no_rooney_rule}:
\begin{itemize}[leftmargin=*]
  \item {\em (Informal).} For all $a^{\sexp{1}}> 1$ there is a constant $C_4>0$, such that
  for all $t\in \N,\ n\in \N,\ k\in [n],\ \ell\in [k]$, and continuous and bounded $\cP$,
  when the panel applies the Rooney Rule, the following holds
  \begin{align}
    \expect\sinsquare{\betat}\geq \Phi\inparen{a^{\sexp{1}}+\frac{C_3}{C_4}\frac{t\rho}{(k-\ell+1)}}\cdot \inparen{1-e^{\frac{t\rho}{16}}}.\label{eq:gen_lowerbound}
  \end{align}
  \item {\em (Informal).} For all $a^{\sexp{1}}> 1$ there is a constant $C_5>0$, such that
  for all $t\in \N,\ k\in \N,\ \ell\in [k]$,
  when the panel applies the Rooney Rule and $\cP\coloneqq {\rm Unif}(0,1)$,
  there is a constant $n_0\in \N$, such that for all $n\geq n_0$
  the following holds
  \begin{align}
    \expect\sinsquare{\betat}\leq \Phi\inparen{a^{\sexp{1}}+C_5{\frac{t\ln{n}}{n(1-\rho)}}}+    {\frac{t\ln{n}}{n(1-\rho)}}.\label{eq:gen_upperbound_main}
  \end{align}
\end{itemize}
(We present their formal statements and proofs sketches in Section~\ref{sec:proof:generalization_for_other_distributions}.)

\paragraph{Discussion of assumptions.}
In Model~\eqref{alg:iterative_selection_and_panel_learning}, $a$, roughly, represents how unbiased the panel is.
Thus, it is natural to expect Assumption (1) and first part of Assumption (3).
We can avoid Assumption (2) (we explain this with the proof); although then the upper bound~\eqref{eq:gen_upperbound_main} becomes slightly weaker.
The second part of Assumption (3), intuitively says that the panel's implicit bias does not have a large probability mass near 0 (extreme bias).
Finally, Assumption (4) upper bounds the ${\rm median}(a)$ in terms of $a$: it says that the median is not too close to 1 (no bias) for a given $a$.

\paragraph{Discussion.}
Allowing for different distributions of implicit bias, potentially due to noise, is related to \cite{EmelianovGGL20}, who consider noise in the panel's implicit bias in one iteration.
However, they draw a different value of implicit bias for each candidate and consider different levels of noise for both groups of candidates.
Expanding Model~\ref{alg:iterative_selection_and_panel_learning} to incorporate these would be an interesting direction for future work.

\begin{remark}
  Note that if we substitute the definition of $\Phi(\cdot)$ for the beta distribution in Equation~\eqref{eq:gen_lowerbound} we recover \cref{thm:rooney_rule}, and if substitute the definition in Equation~\eqref{eq:gen_upperbound_main} we recover \cref{thm:no_rooney_rule}.
\end{remark}

\subsubsection{Generalizing to other update rules}\label{sec:generalizing_for_other_update_rules}
In this section, we consider the update rule
\begin{align*}
  a^{\sexp{t+1}}\coloneqq \at \cdot F\inparen{\frac{U^\et}{\wt{U}^\et}},\yesnum\label{eq:update_rule_gen:main}
\end{align*}
where $F\colon [1,\infty)\to[1,\infty)$ is a continuous function satisfying some mild assumptions.
In particular, we assume that
\begin{enumerate}
  \item $F(1)=1$,
  \item $F$ is strictly increasing, and
  \item $F$ is concave.
\end{enumerate}

\noindent \mbox{We prove the following versions of \cref{thm:rooney_rule} and \cref{thm:no_rooney_rule}:}
\begin{itemize}[leftmargin=*]
  \item {\em (Informal).} Under update rule~\eqref{eq:update_rule_gen:main}, if $F$ satisfies the above assumptions and the panel uses the Rooney Rule,
  then for all $\eps\in (0,1)$, there exists an iteration $t\in \N$, such that, $$\expect[\betat]\geq 1-\eps.$$
  \item {\em (Informal).} Under update rule~\eqref{eq:update_rule_gen:main}, if $F$ satisfies the above assumptions and the panel does not use the Rooney Rule,
  then for all $\eps\in (0,1)$ and $t\in \N$, there exists an $n_0\in \N$, such that, for all $n\geq n_0$ $$\expect[\betat]\leq \expect\sinsquare{\beta^{\sexp{1}}}+\eps.$$
\end{itemize}
(We present their formal statements and proofs sketches in Section~\ref{sec:proof:generalization_for_other_update_rules}.)

\subsubsection{Generalizing \cref{thm:no_rooney_rule} to other distributions.}\label{sec:thm:no_rooney_rule_gen}
We can extend \cref{thm:no_rooney_rule} to other continuous and bounded distributions of latent utility $\cP$.
In this case, we need to make a mild assumption on $\cP$, which is satisfied by several distributions.

\paragraph{Notation.} Let $M$ be the supremum of the support of $\cP$: $M\coloneqq \sup(\supp(\cP))$.
Define $T_\cP(\eps)$ as the probability that an $X$ drawn from $\cP$ is at least $(1-\eps) M$.
Formally,
$T_{\text{\footnotesize $\cP$}}(\eps)\coloneqq \Pr_{\text{\footnotesize $X\negsp\sim\negsp\cP$}}\insquare{X\geq (1-\eps) M}$.\\

\noindent Assume that: there exists a function $\eps\colon \N \to \R_{> 0}$ satisfying
\begin{align}
  \eps(n)+e^{-n\cdot T_\cP(\eps(n))} \leq \frac{1}{\poly(n)}.\tagnum{Assumption}\customlabel{eq:assumption_1}{\theequation}
\end{align}

\noindent Then we get the following upper bound in \cref{thm:no_rooney_rule}:
$$\expect\insquare{\betat} \leq  \expect\insquare{\beta^{\sexp{1}}} + \frac{t}{\poly(n)}.$$
Here, $\poly(\cdot)$ hides the dependence on $(1-\rho)$.
We do not state the dependence on $(1-\rho)$ as it changes with the choice of $\cP$.

\noindent (We present the formal statement in \cref{thm:no_rooney_rule_generalized}.)

\paragraph{Discussion of assumption.}
This assumption holds for several common distributions, including any continuous distribution with a compact interval support; see Section~\ref{sec:facts_related_to_assump}.
This also includes common truncated distributions (such as, the truncated powerlaw and truncated normal distributions)
For instance, when $\cP={\rm Unif}{(0,1)}$, we have $T_\cP(\eps(n))=\eps(n)$, and we can choose $\eps(n)=(\ln n)/n$.

\section{Overview of proof techniques}\label{sec:proof_overviews}
In this section, we present an overview of proofs of Theorems \ref{thm:rooney_rule} and \ref{thm:no_rooney_rule}.
The complete proofs appear in Section~\ref{sec:proofs}.
We first need some additional notation.
Let $\Phi(a)$ denote the expectation of $\Beta(a,b)$:
$$\Phi(a)\coloneqq \frac{a}{(a+b)}.$$
Note that $\Phi$ is an increasing and concave function.
Let $X_{\max}\negsp\coloneqq\negsp \max_{i}\xit{}$ and $Y_{\max}\negsp\coloneqq\negsp \max_{j}\yjt{}$.
Define $\ux$ and $\uy$ be the utility of all $X$-candidates selected and all $Y$-candidates selected in iteration $t$ respectively:
$$\ux\coloneqq \sum\nolimits_{i\in G_X \cap S_{\ell}^\et}\xit{}\ \text{\ and\ }\ \uy\coloneqq \sum\nolimits_{j\in G_Y\cap S_{\ell}^\et}\yjt{} .$$
It will be convenient to consider $\deltat\coloneqq \frac{U^\et}{\wt{U}^\et}-1$.
Expressing the update rule~\eqref{eq:update_rule} with $\deltat$  we get  $$a^{\sexp{t+1}} = \at \cdot (1 + \deltat).$$
We can equivalently write $\deltat$ as:
\begin{align*}
  \deltat = \frac{(1-\betat)\cdot \ux }{\betat\ux+\uy}.\yesnum\label{eq:def_of_delta:main_body}
\end{align*}
Observe that $\deltat$ is a decreasing function of $\uy$ and $\betat$, and an increasing function of $\ux$.
Further, as discussed in the \cref{sec:ourmodel}, $U^\et \geq \wt{U}^\et$, and so, $\deltat\geq 0$ and $a^{\sexp{t+1}}\geq \at$.

\subsection{Proof sketch of Theorem~\ref{thm:rooney_rule}}

Fix an iteration $t\in \N$.
In the proof, we show a lower bound on $\at$ which holds with high probability.
Then, using the fact that $\Phi(\cdot)$ is increasing, we can condition the event that ``$\at$ is large,'' to derive the desired lower bound on $\expect[\betat]=\expect_{\at}[\Phi(\at)]$.\\%

Towards this, consider the following events:
\begin{align*}
  \text{$\cE^\et \coloneqq \inbrace{ X_{\max}^\et> Y_{\max}^\et} \text{ and }\
  \cF^\et \coloneqq \inparen{\betat\leq {\rm median}(\betat)}.$}
\end{align*}
\mbox{Conditioning on them we can show a one-step lower bound on $\at$.}
\begin{lemma}[\bf Conditional lower bound on $\at$]\label{lem:rooney_rule_increases_a_fast:main_body}
  For all $s\in \N$, if $\ell>0$, then $(\cE^\es\ \land \ \cF^\es)$ implies
  \begin{align*}
    \inparen{a^{\sexp{s+1}}>a^\es+ \frac{a^{\sexp{1}}(b-1)}{(k-\ell+1)(a^{\sexp{1}}+b)} }.
  \end{align*}
\end{lemma}
\noindent {\em Proof outline:}
First, we show that $\nfrac{U_{Y}^\es}{U_{X}^\es}\leq k-\ell$ conditioned on $\cE^\es$.
Then, we show an upper bound on the median of the beta distribution, such that, given $\cF^\es\negsp \negsp \coloneqq \negsp \sinparen{\beta^\es \negsp \negsp \leq \negsp {\rm median}(\beta^\es)}$, it holds
\begin{align*}
  \frac{(1-\betat)}{\betat+(k-\ell) }\geq \frac{(b-1)}{(k-\ell+1)(a+b)}.\yesnum\label{eq:bound_on_med:main_body}
\end{align*}
Now, conditioning on $\cE^\es$ and $\cF^\es$ and using the above, we get
\begin{align*}
  \delta^\es=\frac{(1-\betat)}{\betat+ \nfrac{U_{Y}^\es}{U_{X}^\es}} \geq \frac{(1-\betat)}{\betat+(k-\ell) } \stackrel{\eqref{eq:bound_on_med:main_body}}{\geq} \frac{(b-1)}{(k-\ell+1)(\at+b)}.
\end{align*}
The lemma follows by using the update rule and simplifying.\\

Next, we extend \cref{lem:rooney_rule_increases_a_fast:main_body} along with other facts to show a lower bound on $\at$ (alluded to earlier).
Towards this, note that $Pr[\cF^\et]=\frac12$.
We show that $\Pr[\cE^\et]=\rho$ by analyzing an equivalent urn model.
Let $Z^\et \negsp \in \negsp \zo$ be the indicator random variable that $(\cE^\et \land \cF^\et)$ occurs.
Since ${\cE}^\et$ and $\cF^\et$ are independent,
\begin{align*}
  \Pr[Z^\et=1] \coloneqq \Pr[\cE^\et \land \cF^\et]
  =\Pr[\cF^\et]\cdot \Pr[{\cE}^\et]
  = \frac{\rho}{2}.\yesnum\label{eq:prob_Z_proof1:main_body}
\end{align*}
\noindent Note that, the event $\cF^\et$ only depends on the draw of $\betat$\footnote{Specifically, the inverse CDF of the draw of $\betat$.} and that
${\cE}^\et$ only depends on $X^\et$ and $Y^\et$.
Thus, it follows that $Z^\et$ only depends on random variables from the $t$-th iteration,
and that, for all $t_1\neq t_2$, $Z^{\sexp{t_1}}$ and $Z^{\sexp{t_2}}$ are independent random variables.
This allows us to use the Chernoff bound on the sum $\sum_{s=1}^t Z^\es$:
\begin{align*}
  \Pr\insquare{\tmpsum \leq \frac{t\rho}{4}}
  &\stackrel{\eqref{eq:prob_Z_proof1:main_body}}{=}\Pr\insquare{\tmpsum \leq \inparen{1-\frac12} \cdot \expect\insquare{\tmpsum}}\\
  &\hspace{0mm}\leq \exp\inparen{-\frac{1}{2^2\cdot 2}\cdot \frac{t\rho}{2}}.\yesnum\label{eq:chernoff_bound_on_z:main_body}
\end{align*}
\noindent Using \cref{lem:rooney_rule_increases_a_fast:main_body} and that $(Z^\es=1)\negsp \coloneqq\negsp  (\cE^\es\land\cF^\es)$ for each $s\negsp\in\negsp [t]$, we get
\begin{align*}
  a^{t} \geq a^{\sexp{1}}+\inparen{\sum\nolimits_{s=1}^t Z^\es}\cdot \frac{a^{\sexp{1}}(b-1)}{(k-\ell+1)(a^{\sexp{1}}+b)}.
\end{align*}
Substituting this in Equation~\eqref{eq:chernoff_bound_on_z:main_body} we get
\begin{align*}
  &\hspace{-3mm}\Pr\insquare{\at \geq a^{\sexp{1}}+\frac{t\rho}{4}\frac{a^{\sexp{1}}(b-1)}{(k-\ell+1)(a^{\sexp{1}}+b)}}\geq 1-\exp\inparen{-\frac{t\rho}{16}}.\yesnum\label{eq:lowerbound_on_prob:main_body}
\end{align*}
Let $C$ be the constant $C\coloneqq \frac{a^{\sexp{1}}(b-1)}{4(a^{\sexp{1}}+b)}>0,$
and
$\underline{a} \coloneqq a^{\sexp{1}}+C\frac{t\rho}{(k-\ell+1)}.$
(Recall that $a^{\sexp{1}}$ and $b$ are fixed constants strictly larger than 1.)
Now we can prove the result as follows.
\begin{align*}
  \ex{\betat} &= \expect_{\at}\insquare{\Phi(\at)}\\
  &\hspace{-0mm} = \int_{a^{\sexp{1}}}^{\lb}\Phi(a)\ d\Pr[\at=a] + \int_{\lb}^{\infty}\Phi(a)\ d\Pr[\at=a]\\
  &\hspace{-0mm}\geq \Phi(\lb)\cdot \int_{\lb}^{\infty} d\Pr[\at=a]\tag{$\Phi$ is increasing}\\
  &\hspace{-0mm}\hspace{-1mm} \stackrel{\eqref{eq:lowerbound_on_prob:main_body}}{\geq} \Phi(\lb)\cdot \inparen{1-\exp\inparen{-\frac{t\rho}{16}}}\\
  &\hspace{-0mm}\stackrel{}{=} \frac{\lb}{\lb+b}\cdot \inparen{1-\exp\inparen{-\frac{t\rho}{16}}}\\
  &\hspace{-0mm}\stackrel{}{=} \inparen{1-\frac{b}{a^{\sexp{1}} + b + C\cdot \frac{t\rho}{(k-\ell+1)}}}   \cdot \inparen{1-\exp\inparen{-\frac{t\rho}{16}}}.\yesnum\label{eq:bound_on_at_proof}
\end{align*}
Simplifying the above equation and choosing $C_1\coloneqq \frac{C}{(a^{\sexp{1}}+b)}$ completes the proof.
\begin{remark}
  Notice that the only properties of the beta distribution that proof used, are that $\phi(\cdot)$ is an increasing function and that the median of $\betat$ is upper bounded by a suitable value.
  Abstracting these properties would give an analogous proof of Equation~\eqref{eq:gen_lowerbound}.
\end{remark}

\begin{remark}
  It might be possible to tighten the dependence of Equation~\eqref{eq:upperbound_on_bias} on $\ell$ using techniques from \cite{celis2020interventions}.
  In particular, the factor $\frac{t\rho}{k-\ell+1}$ may improve to, roughly, $\ell\cdot\frac{t\rho}{k-\ell+1}$.
\end{remark}

\subsection{Proof sketch of Theorem~\ref{thm:no_rooney_rule}}
Fix an iteration $t\in \N$.
In the proof, we show an upper bound on $\expect[\at]$.
Then, using that $\Phi(\cdot)$ is concave and increasing, we can prove the desired upper bound on $\expect[\betat]=\expect_{\at}[\Phi(\at)]$.

Let $D^\et$ denote the values of all random variables for the first $[t-1]$ iterations:
$$D^\et\coloneqq \inbrace{\beta^\es}_{1\leq s<t}\cup \inbrace{X^\es}_{1\leq s<t}\cup \inbrace{Y^\es}_{1\leq s<t}.$$
Then, we prove the following lemmas.
\begin{lemma}[\bf Upper bound on $\deltat$ when $\betat$ is close to 1]\label{lem:delta_small_when_beta_large:main_body}
  For all iterations $\ s\in \N$, constant $\eps\in \inparen{0,\nfrac{1}{2}}$, and values of $D^\es$, if $\beta^\es\geq 1-\eps$, then $\delta^\es\leq 2\eps$.
\end{lemma}
\noindent {\em Proof:}
$\delta^\es$ is an increasing function of $\nfrac{\uys}{\uxs}$.
In the worst case we have $\nfrac{\uys}{\uxs}=0$.
Substituting this in Equation~\eqref{eq:def_of_delta:main_body}:
$$\delta^\es\ \stackrel{\eqref{eq:def_of_delta:main_body}}{=}\ \ \frac{(1-\beta^\es)\cdot\uxs}{\beta^\es\uxs+\uys}\leq\frac{(1-\beta^\es)}{\beta^\es}\leq\frac{\eps}{1-\eps}\ \stackrel{(\eps<\frac12)}{\leq}\ \
2\eps.$$
\begin{lemma}[\bf Upper bound on $\ex{\deltat}$ when $\betat$ is not close to 1] \label{lem:delta_bounded_in_expect:main_body}
  For all iterations $s\in\N$, constant $\eps\in \inparen{0,\nfrac{1}{2}}$, values of $D^\es$, and parameters $a^{\sexp{1}},b>1$, it holds
  $$\Pr[\beta^\es< 1-\eps]\cdot \exgiv{\delta^\es}{(\beta^\es< 1-\eps)\land (\ux\neq 0)}\leq \frac{b}{a^{\sexp{1}}-1}.$$
\end{lemma}
\noindent Due to space constraints we omit the proof of \cref{lem:delta_bounded_in_expect:main_body}.
It roughly follows by upper bounding $\delta^\es$ by $(\beta^\es)^{-1}$, rearranging, and using the value of the inverse moment of the beta distribution.
\begin{lemma}[\bf Upper bound on $\Pr\insquare{\ux\neq 0}$ when $\betat$ is not close to 1]\label{lem:delta_0_whp:main_body}
  For all iterations $t\in \N$ and values of $D^\es$, if $\ell=0$, then there exists an $n_0\negsp \in \negsp \N$, such that,
  \mbox{for all $n\negsp \geq \negsp n_0$ and $\eps\negsp \leq\negsp  \frac{\ln{n}}{n}$, it holds}
  \begin{align*}
    \probgiv{\ux\neq 0 }{ \betat < 1-\eps }\leq \exp\inparen{-\frac{1}{8}\eps n(1-\rho)}.
  \end{align*}
\end{lemma}
\noindent{\em Proof outline:} In the proof, we show that the panel selects an $X$-candidate with probability at most $\exp\inparen{-\frac{1}{8}\eps n(1-\rho)}$.
Notice that this implies the lemma, since if no $X$-candidates are selected then $\uxs = 0$.
Towards this, we use the Chernoff bound to show that, conditioned on $\beta^\es < 1-\eps$, with high probability at least $k$ $Y$-candidates have higher observed utilities than $X_{\max}$.\\

\noindent (Note that \cref{lem:delta_0_whp:main_body} uses $\ell=0$.)
We would like to upper bound $\exgiv{\delta^\es }{ D^\es}$ for all values of $D^\es$.
To this end, consider three cases:
\begin{enumerate}
  \item $\beta^\es$ is close to 1 ($\beta^\es\geq 1-\eps$),
  \item $\beta^\es$ is far from 1 ($\beta^\es< 1-\eps$) and $\ux\neq 0$, and
  \item $\beta^\es$ is far from 1 ($\beta^\es< 1-\eps$) and $\ux=0$.
\end{enumerate}
Using \cref{lem:delta_small_when_beta_large:main_body}, in Case (1) we have $\delta^\es\leq 2\eps$.
From \cref{lem:delta_bounded_in_expect:main_body} and \cref{lem:delta_0_whp:main_body}, we can bound $\exgiv{\delta^\es}{ D^\es}$
by $\factor \cdot \exp\inparen{-\frac{1}{8}\eps n(1-\rho)}$.
Finally, in Case (3) it holds that $\delta^\es=0$ (see Equation~\eqref{eq:def_of_delta:main_body}).
Note that, each of these results holds for all values of $D^\es$.
Combining these and setting $\eps\coloneqq \frac{8\ln{n}}{n(1-\rho)}$ we can show that
\begin{align*}
  \exgiv{\delta^\es}{D^\es}\leq \uniffactor{}.\yesnum\label{eq:upperbound_on_expec_deltat:main_body}
\end{align*}
Consider $t_1 < t_2$. Using Equation~\eqref{eq:upperbound_on_expec_deltat:main_body} we have
\begin{align*}
  \ex{\delta^{t_1}\cdot \delta^{t_2}} &= \int_{0}^{\infty}x\cdot \ex{\delta^{t_2}\mid \delta^{t_1}=x}\cdot \Pr[\delta^{t_1}=x]dx\\
  &\stackrel{\eqref{eq:upperbound_on_expec_deltat:main_body}}{\leq} \int_{0}^{\infty}x\cdot  \inparen{\uniffactor{}} \cdot \Pr[\delta^{t_1}=x]dx\\
  &= \inparen{\uniffactor{}}\cdot \int_{0}^{\infty}x \cdot \Pr[\delta^{t_1}=x]dx\\
  &\leq \inparen{\uniffactor{}}\cdot\ex{\delta^{t_1}}\\
  &\stackrel{\eqref{eq:upperbound_on_expec_deltat:main_body}}{\leq} \inparen{\uniffactor{}}^2.
\end{align*}
Similarly for $t_1<t_2<\dots<t_s$ we can show that
\begin{align*}
  \expect\insquare{\prod\nolimits_{i\in [s]}\inparen{1+\delta^{t_i}}}\leq \inparen{1+\uniffactor{}}^s.\yesnum\label{eq:expect_product_delta}
\end{align*}
Pick an $n_0$, such that, $\uniffactor{}\leq \frac{1}{t}$.
Then, we can show:
\begin{align*}
  \expect\insquare{\at}
  \negsp = \negsp a^{\sexp{1}}\cdot \expect\insquare{\prod\nolimits_{i=1}^{t-1}\inparen{1+\delta^{t_i}}}
  \negsp\leq\negsp a^{\sexp{1}} + 2t\uniffactorb{}.
  \yesnum\label{eq:bound_on_a:main_body}
\end{align*}
Further, we have
\begin{align*}
  \expect\sinsquare{\betat}
  &= \expect_{\at}\sinsquare{\Phi(\at)}\\
  & \stackrel{}{\leq} \ \Phi\inparen{\expect_{\at}\sinsquare{\at}} \tag{$\Phi$ is concave}\\
  & = \frac{\expect_{\at}\sinsquare{\at}}{\expect_{\at}\sinsquare{\at}+b}. \tag{$\Phi$ is increasing}
\end{align*}
Now substituting Equation~\eqref{eq:bound_on_a:main_body} in the above, rearranging, and choosing $C_2\coloneqq \frac{32a^{\sexp{1}}}{a^{\sexp{1}}-1}$, we can complete the proof.

\begin{remark}
  Notice that the only properties of the beta distribution used in the proof, are that $\Phi(\cdot)$ is concave and increasing and an upper bond on the inverse moment of the inverse moment of the beta distribution.
  Abstracting these properties gives an analogous proof for Equation~\eqref{eq:gen_upperbound_main}.
\end{remark}

\section{Empirical observations} \label{sec:empirical_results}

We enlisted participants on Amazon Mechanical Turk for an iterative candidate selection experiment in order to observe the effect of the Rooney Rule on their long-term behavior.
Participants were shown colored tiles representing candidates, along with the tiles' \emph{observed} utilities, and were instructed to select a small fixed number of tiles with the goal of maximizing the \emph{latent} utility of the selection; this was incentivized by tying their bonus payment to the latent utility attained.
They were informed that the observed utilities were noisy estimates of the latent utilities.
We required some participants to follow the Rooney Rule and did not constrain others.
In this section, we give an overview of the experimental design and our analysis of the results.

\begin{subsection}{Experimental design} \label{sec:experimental_design}
  A total of 76 participants located in the U.S. were recruited using Amazon Mechanical Turk to complete a repeated ``selection'' experiment.\footnote{Participants were paid at a rate of about \$15/hour. The experiment received institutional review board (IRB) exemption, Category 2 and 3, from [Anonymized Institution]. To complete a version of the experiment similar to the one presented to participants, see \cite{celis2020website}. The code for the experiment is available at \cite{celis2020code}.}
  In each of $T=25$ iterations, participants were presented with $n=100$ tiles, representing candidates, in order of highest to lowest observed utility.
  Half of the tiles were red (representing $Y$-candidates) and half were blue (representing $X$-candidates); i.e., $\rho=\frac12$.
  The latent utilities, $X_i^{\sexp{t}}, Y_j^{\sexp{t}}$ for each tile was drawn from ${\rm Unif}(0, 100)$, the uniform distribution on the interval [0, 100].
  The observed utilities for the red tiles $\wt{Y}_j^{\sexp{t}}$ were sampled from $\distN(Y_j^{\sexp{t}}, 3)$.\footnote{$\distN(\mu, \sigma)$ is the normal distribution with mean $\mu$ and standard deviation $\sigma$.}
  The observed utilities for the blue tiles were $\wt{X}_i^{\sexp{t}}$ were sampled from $ \distN(\beta \cdot X_i^{\sexp{t}}, 3)$ where $\beta=\frac{2}{3}$ is the bias coefficient.\footnote{We chose $\beta=\frac{2}{3}$ because previous empirical work has found evidence for bias at least as extreme as $\frac{2}{3}$ \cite{wenneras2001nepotism}.}
  All observed and latent utilities were then rounded to the nearest integer.

  The participants were instructed to select $k=10$ tiles to try to maximize the total latent utility $U^{\sexp{t}}$ in each iteration $t\in [T]$, given the observed utility of each tile.
  They were told that the observed utilities were noisy estimates of the latent utilities.
  In each round, after they submitted their choices, participants were shown the latent utility of each tile they selected; hence, they could learn about the latent utilities of each group over time.\footnote{For images of the user interface, the instructions, the demonstration, and feedback, we refer the reader to Section~\ref{sec:demo}.}

  About half of the participants (39 out of 76), chosen at random, were constrained to follow the Rooney Rule with $\ell = 1$, i.e., they were required to select at least one blue tile.
  The other half of the participants (37 out of 76) were unconstrained (i.e., $\ell=0$).
  We refer to the constrained group of participants as \rrgrp{} and the unconstrained group of participants as \unconsgrp{}.
  Our goal was to compare the performance across the two groups  \rrgrp{} and \unconsgrp{} in the long term, i.e., after having been granted some opportunity to learn about the latent utilities of both groups.
  Hence, our analysis focuses on the choices made in the last fifteen iterations.\footnote{This number is to some extent arbitrary, but we found similar results for other divisions.}

\end{subsection}

\begin{subsection}{The Rooney Rule results in less inequality in the long term} \label{subsec:experiment_contrasting_outcomes}

  Recall that participants selected $k=10$ tiles in each iteration and that, according to the distribution of latent utilities, participants should select five blue and five red tiles on average, the observed utilities exhibited bias against the blue tiles, rendering them lower in the ranking and hence less likely to be selected.
  In the last fifteen observations, we observe that the \rrgrp{} selects 4.0 blue tiles on average, while \unconsgrp{} selects 2.5 on average.

  As participants subject to the Rooney Rule are required to select at least one blue tile, one might wonder if this difference between the \rrgrp{} and \unconsgrp{} is solely due to the Rooney Rule requirement.
  Surprisingly, this is not the case, suggesting that additional learning has occurred.
  In particular, we consider the number of blue tiles selected \textit{in addition to} the ones required by the Rooney Rule;
  i.e., we subtract 1 from the number of blue tiles selected for \rrgrp{} participants and 0 from \unconsgrp{}.
  We find a statistically significant ($p<0.001$) difference between the number of blue tiles selected in addition to the requirement between \rrgrp{} and \unconsgrp{} (see Figure \ref{tab:num_selected_minority_over_required}).

  \begin{figure}[ht]
    \centering
    \setlength{\abovecaptionskip}{0pt}
    \caption[The number of blue tiles selected in addition to the required number is greater for \rrgrp{} than \unconsgrp{} with significance $p<0.001$ using Welch's $t$-test.]{The number of blue tiles selected in addition to the required number is greater for \rrgrp{} than \unconsgrp{} with significance $p<0.001$ using Welch's $t$-test.\footnotemark}\vspace{2mm}

    \begin{tabular}{ r c c }
      & {Mean} & {Standard deviation} \\ \hline
      \rrgrp & 3.0 & 2.4 \\
      \unconsgrp & 2.5 & 3.1 \\
    \end{tabular}\vspace{2mm}

    \begin{tabular}{ r c c }
      & {Welch's $t$-test} \\ \hline
      Null hypothesis & population means same \\
      Alternative hypothesis & population means different \\
      $t$-statistic & 8.2 \\
      Degrees of freedom & 1800 \\
      $p$-value & $<$0.001 \\
    \end{tabular}

    \label{tab:num_selected_minority_over_required}
  \end{figure}

  \footnotetext{Welch's $t$-test assumes independence between samples. We also analyzed mixed effects model to account for the multiple samples from each participant and found similar results (see Section~\ref{sec:num_minority_comparison_app}).}

  \noindent This suggests that participants ``learn'' to select more equal numbers of blue and red tiles as a result of implementing the Rooney Rule.
  This aligns with our theoretical result that participants subject to the Rooney Rule reduce their biases at a faster rate than an unconstrained panel.

\end{subsection}

\begin{subsection}{The Rooney Rule does not negatively affect the utility in the long term} \label{subsec:experiment_optimality}

  We observe that that the total utility $U^{\sexp{t}}$ modestly increases over iterations $t\in [T]$ for both \rrgrp{} and \unconsgrp{}.
  This suggests that participants in both \rrgrp{} and \unconsgrp{} are learning to select more optimally over time.

  However, we are particularly interested in understanding how well a hiring panel subject to the Rooney Rule can fulfill its (latent) utility-maximizing objective in the long term.
  To this end, we compare the mean latent utility attained by \rrgrp{} and \unconsgrp{} during the last fifteen iterations.
  We measure the latent utility achieved by \rrgrp{} and \unconsgrp{} as a proportion of an ``optimal'' latent utility that could be achieved if one knew the distributions of latent and observed utilities a priori.

  In a given iteration $t\in [T]$, define the \textit{optimal strategy set} as the tiles with the top $k$ values from $\{ \nfrac{\wt{X}_i^{\sexp{t}}}{\beta} \}_{i \in G_X} \cup \{ \wt{Y}_j^{\sexp{t}} \}_{j\in G_Y}$. We call the latent utility of the optimal strategy set the \textit{optimal strategy utility}.\footnote{Because of the noise added to each observed utility, the optimal strategy may not be the same as the set of tiles with the greatest latent utilities.
  }
  From this point on, we report the latent utility for a given selection as a fraction of the optimal strategy utility.
  We also define the \textit{latent utility derived from blue (resp. red) tiles} for a given selection to mean the latent utility of selected blue (resp. red) tiles as a fraction of the optimal strategy utility.

  \begin{figure}[ht]
    \setlength{\abovecaptionskip}{-2pt}
    \caption{The latent utility of \rrgrp{} is 2\% less than that of \unconsgrp{}, but the latent utility derived from the blue tiles was higher for \rrgrp{} than \unconsgrp{}.}\vspace{2mm}

    \centering

    \begin{tabular}{ r c c c }
      \multicolumn{4}{c}{\begin{tabular}{@{}c@{}}Latent utility: mean \textit{(standard deviation)}
    \end{tabular}
    } \\ \hline
    & overall & \begin{tabular}{@{}c@{}}derived from blue\end{tabular} & \begin{tabular}{@{}c@{}}derived from red\end{tabular} \\ \hline
    \rrgrp & 0.87 \textit{(0.11)} & 0.37 \textit{(0.22)} & 0.50 \textit{(0.24)} \\
    \unconsgrp & 0.89 \textit{(0.12)} & 0.22 \textit{(0.28)} & 0.67 \textit{(0.31)} \\
  \end{tabular}

  \label{tab:latent_utility}
\end{figure}

\noindent
We find that, in the long term, the latent utility obtained by \rrgrp{} is less than that of the \unconsgrp{} by only 2\%.
Moreover, the latent utility derived from the blue tiles was higher for the \rrgrp{} than \unconsgrp{} (see Figure \ref{tab:latent_utility}).
Hence, even though the latent utility attained by \rrgrp{} and \unconsgrp{} were not substantially different, a much greater proportion of \rrgrp{} latent utility came from blue tiles.
(The difference between the means of latent utility derived from blue tiles for the two groups is statistically significant at $p<0.001$. See Section~\ref{sec:latent_utility_details} for details.)
This suggests that implementing the Rooney Rule can increase the contributions from underrepresented candidates to the overall latent utility without substantially decreasing the total latent utility attained.

As another metric for how well a hiring panel subject to the Rooney Rule can fulfill its utility-maximizing objective, we also analyzed \textit{optimal strategy overlap}, the number of selected tiles in the optimal strategy set. We find similar results: the total optimal strategy overlap is not very different for \rrgrp{} versus \unconsgrp{}, but the optimal strategy overlap derived from to blue tiles is much higher for \rrgrp{} (see Section~\ref{sec:optimal_strategy_overlap_details} for details).

\end{subsection}

\subsection{Summary of empirical results}

We observe two dominant trends in participant behavior in the empirical iterative selection experiment:
\begin{itemize}[leftmargin=10pt]
  \item Participants subject to the Rooney Rule select more blue tiles \textit{in addition to} the required number in the long term than participants when there is no rule in effect. We interpret this to mean that they learn to select more equal numbers of blue and red tiles faster than the unconstrained participants. This aligns with our main theoretical result that participants subject to the Rooney Rule will reduce their biases at a faster rate than an unconstrained panel.
  \item
  There was no substantial difference in the mean latent utility of \rrgrp{} versus \unconsgrp{}. However, while the mean total latent utility between \rrgrp{} and \unconsgrp{} was not very different, a much larger proportion of the latent utility came from blue tiles for participants subject to the Rooney Rule. We interpret this to mean that participants subject to the Rooney Rule learn to increase the latent utility derived from blue tiles while nearly matching the total latent utility of selections made by unconstrained participants.
\end{itemize}
\noindent These trends suggest that, over the long term, implementing the Rooney Rule can yield significant benefits for increasing representation of underrepresented candidates without substantially decreasing total latent utility.

\begin{section}{Limitations}\label{sec:limitations}

  \paragraph{Model and theoretical results.} %
  We consider a model with two groups, where the panel is implicitly biased against one, underrepresented group.
  However, in real-world applications, sometimes there may be multiple and intersecting groups (for example, those defined by race and gender) and where the candidates at the intersection of two or more underrepresented groups may face a larger implicit bias~\cite{williams2014double, bertrand2004emily}.

  Further, we consider that the panel evaluates the candidates selected (say, for an interview) in aggregate (it compares $U^{\sexp{t}}$ and $\wt{U}^{\sexp{t}}$) and sees the latent utilities of all selected candidates without noise.
  But, the panel may assign more or less weight to the latent utility of the $X$-candidates and learn only a noisy version of the latent utility of the selected candidates.
  We make some progress toward this by allowing for noise in the panel's implicit bias (see \cref{sec:generalization_for_other_distributions}).

  \paragraph{Empirical results.} %
  The iterative selection experiment in our empirical results is different from real-world candidate selection tasks in several ways:
  For one, the bias $\beta$ is imposed exogenously on the observed utilities that the selection panel sees, rather than coming from their own implicit cognitive processes.
  If selection panels evaluate the merit of candidates based on their own experience or criteria, they might trust these assessments more than the observed utilities presented to them in the experiment.
  Similarly, the stakes of the experiment for participants were lower than in real-world tasks; %
  selection panels in the real world may pay more attention to slight variations in the performance of the candidates selected.

  In future work, it would be interesting to study how participant biases change over time if they assess candidates based on subjective criteria (like a recommendation letter) that is biased against one group, rather than being provided biased observed utilities at the beginning of each iteration.
  It would also be valuable to incorporate the real-world social biases of participants into an iterative candidate selection task---this could shed insight on how actual social biases shape people's decisions.

  \paragraph{Specificity.}
  Our theoretical analysis and empirical observations are only applicable when the panel is not explicitly biased and the selection process is not {tokenistic}.
  An explicitly biased panel would remain biased no matter which candidates are presented, and a tokenistic process may not require the panel to consider underrepresented candidates seriously. %

  Further, other societal and structural biases and discrimination can also influence the panel, as when standardized test scores disadvantage minority and low-income candidates \cite{elsesser2019lawsuit}.
  The Rooney Rule is only one of the many policies  to mitigate bias and discrimination more broadly, and it is important to consider implementing it as a part of a larger toolkit for positive social change. %

\end{section}

\section{Proofs of main theoretical results}\label{sec:proofs}
In this section, we present proofs of our main theoretical results (Theorem~\ref{thm:rooney_rule} and Theorem~\ref{thm:no_rooney_rule}).
We gave an overview of these proofs in Section~\ref{sec:proof_overviews}.
In order to be self-contained this section repeats some lemmas from Section~\ref{sec:proof_overviews}.

\subsection{Additional notation}\label{sec:preliminaries}
We first need some additional notation.
Let $Z_{(i:n)}$ be the $i$-th {\em order statistic} from $n$ i.i.d. draws of a random variable $Z$.
E.g., $X_{(i:|G_X|)}$ is the $i$-th largest utility among the all $X$-candidates.
Let $n_X$ be the number of $X$-candidates and $n_Y$ be the number of $Y$-candidates: $n_X\coloneqq |G_X|$ and $n_Y\coloneqq |G_Y|$.
Let ${\rm median}(Z)$ denote the median of a random variable $Z$.
Define $\ux$ and $\uy$ be the total latent utility of all $X$-candidates selected and all $Y$-candidates selected in iteration $t$ respectively:
$$\ux\coloneqq \sum_{i\in G_X \cap S_{\ell}^\et}\xit{}\ \text{\ and\ }\ \uy\coloneqq \sum_{j\in G_Y\cap S_{\ell}^\et}\yjt{} .$$
It will be convenient to consider $\deltat\coloneqq \frac{U^\et}{\wt{U}^\et}-1$.
Expressing the update rule~\eqref{eq:update_rule} with $\deltat$:
\begin{align*}
  a^{\sexp{t+1}} = \at \cdot (1 + \deltat).\yesnum\label{eq:update_app}
\end{align*}
Substituting value of $U^\et$ and $\wt{U}^\et$ we get:
\begin{align*}
  \deltat = \frac{(1-\betat)\cdot \ux }{\betat\ux+\uy}.\yesnum\label{eq:def_of_delta}
\end{align*}
Observe that $\deltat$ is a decreasing function of $\uy$ and $\betat$, and an increasing function of $\ux$.
Further, since $0\leq \betat\leq 1$, $\ux\geq 0$, and $\uy\geq 0$, it follows that $\deltat\geq 0$.
Then using Equation~\eqref{eq:update_app} we have that $a^{\sexp{t+1}}\geq \at$.
This suffices to prove the following fact.
\begin{fact}\label{fact:a_is_monotonic_in_t}
  For all $t_1\geq t_2$, we have
  $a^{\sexp{t_1}}\geq a^{\sexp{t_2}}$.
\end{fact}

\subsection{Proof of Theorem~\ref{thm:rooney_rule}}\label{sec:proof:thm:rooney_rule}
We first state a fact used in the proof of \cref{thm:rooney_rule}.
\begin{fact}\label{fact:bounding_median_pdf}
  For all $a,b>1$, then for $X\sim \Beta(a,b)$ it holds that
  \begin{enumerate}[itemsep=0pt]
    \item if $1<b<a$, then $ {\rm median}(X)\leq \frac{a-1}{a+b-2}$,
    \item if $1<a<b$, then $ {\rm median}(X)\leq \frac{a}{a+b}$, and
    \item if $a=b$, then $ {\rm median}(X) = \frac{1}{2}$.
  \end{enumerate}
\end{fact}
\begin{proof}
  The first two results are proved in \cite[Theorem 1]{payton1989bounds}.
  The last result (with $a=b$) follows by symmetry.
\end{proof}
\noindent Next, we define two events which will come up in the proof.
\begin{definition}
  For all iterations $t\in \N$, define the events
  \begin{align*}
    \text{$\cE^\et \coloneqq \inbrace{ X_{(n_X : n_X)}^\et> Y_{(n_Y : n_Y)}^\et} \text{ and }\
    \cF^\et \coloneqq \inparen{\betat\leq {\rm median}(\betat)}.$}
  \end{align*}
\end{definition}
\noindent It holds that $\Pr[\cF^\et]=\frac12$. Next, we calculate $\Pr[\cE^\et]$.
\begin{proposition}\label{prop:prob_of_eventE}
  For all iterations $t\in \N$, if $\cP$ is a continuous distribution, then
  \begin{align*}
    \Pr[\cE^\et]=\frac{n_X}{n_X+n_Y}=\rho.\label{eq:prob_r}\yesnum
  \end{align*}
\end{proposition}
\begin{proof}
  Recall that all utilities $X$ and $Y$, are drawn according to $\cP$.
  Since $\cP$ is a continuous distribution, two utilities clash with 0 probability.
  Thus, an equivalent method to order candidates is: %
  Consider $n_X$ red balls and $n_Y$ blue balls.
  Let each red ball correspond to a particular $X$-candidate and each blue ball correspond to a particular $Y$-candidate.
  Number them from $\inbrace{1,\dots,(n_X+n_Y)}$ uniformly at random.
  Output this order.
  Here, $\cE^\et$ corresponds to the event that a red ball is numbered 1.
  Using symmetry we have $\Pr[\cE^\et]=\frac{n_X}{(n_X+n_Y)}=\rho$.
\end{proof}
\noindent Finally, we state a lemma which shows that $\at$ increases by a constant about conditioned on $\cE^\et$ and $\cF^\et$. Note that this uses $\ell>0$.
\begin{lemma}[\bf Conditional lower bound on $\at$]\label{lem:rooney_rule_increases_a_fast}
  For all iterations $t\in \N$, if $\ell>0$, then $(\cE^\et\ \land\  \cF^\et)$ implies
  \begin{align*}
    \inparen{a^{\sexp{t+1}}>\at+ \frac{a^{\sexp{1}}(b-1)}{(k-\ell+1)(a^{\sexp{1}}+b)} }.\yesnum\label{eq:a_increases_by_const}
  \end{align*}
\end{lemma}
\begin{proof}[Proof of \cref{thm:rooney_rule} (Assuming \cref{lem:rooney_rule_increases_a_fast})]
  Let $Z^\et\in \zo$ be the indicator random variable that $\cE^\et$ and $\cF^\et$ occur.
  Since ${\cE}^\et$ and $\cF^\et$ are independent, using \cref{prop:prob_of_eventE} and $\Pr[\cF^\et]=\frac12$: %
  \begin{align*}
    \Pr[Z^\et=1] \coloneqq \Pr[\cF^\et\ \land \ {\cE}^\et]
    =\Pr[\cF^\et]\cdot \Pr[{\cE}^\et]
    = \frac{\rho}{2}.\yesnum\label{eq:prob_Z_proof1} %
  \end{align*}
  \noindent Note that the event $\cF^\et$ only depends on the draw of $\betat$ (specifically, the inverse CDF of this the draw of $\betat$) and independent of $\at$ and all other random variables in the dynamics.
  Further, ${\cE}^\et$ only depends on the draw of latent utilities in the $t$-th iteration, so, it is  independent of all other random variables in the dynamics.
  It follows that $Z^\et$ only depends on random variables from the $t$-th iteration.
  This shows that for all $t_1\neq t_2$, $Z^{\sexp{t_1}}$ and $Z^{\sexp{t_2}}$ are independent random variables.
  Thus, we can use the Chernoff bound~\cite{motwani1995randomized} on the sum $\sum_{s=1}^t Z^\es$.
  \begin{align*}
    \Pr\insquare{\sum_{s=1}^t Z^\es \leq \inparen{1-\frac12}\cdot\expect\insquare{\sum_{s=1}^t Z^\es}}&\stackrel{\eqref{eq:prob_Z_proof1}}{=}\Pr\insquare{\sum_{s=1}^t Z^\es \leq \frac{t\rho}{2\cdot 2}}\\
    &\leq \exp\inparen{-\frac{1}{2^2\cdot 2}\cdot \frac{t\rho}{2}}. \yesnum\label{eq:chernoff_bound_on_z}
  \end{align*}
  \noindent \mbox{Using \cref{lem:rooney_rule_increases_a_fast} and that $Z^\es\negsp \implies\negsp  (\cE^\es\land\cF^\es)$ for each $s\negsp\in\negsp [t]$, we get}
  \begin{align*}
    a^{t} \geq a^{\sexp{1}}+\inparen{\sum\nolimits_{s=1}^t Z^\es}\cdot \frac{a^{\sexp{1}}(b-1)}{(k-\ell+1)(a^{\sexp{1}}+b)}.
  \end{align*}
  Using Equation~\eqref{eq:chernoff_bound_on_z} we get
  \begin{align*}
    \Pr\insquare{\at \geq a^{\sexp{1}}+\frac{t\rho}{4}\frac{a^{\sexp{1}}(b-1)}{(k-\ell+1)(a^{\sexp{1}}+b)}}\geq 1-\exp\inparen{-\frac{t\rho}{16}}.\yesnum\label{eq:lowerbound_on_prob}
  \end{align*}
  We need some additional notation before proceeding.
  Define $\Phi(a)$ as the expectation of $\betat$ when $\at=a$: %
  $$\Phi(a)\coloneqq \exgiv{\betat}{\at=a}= \frac{\at}{\at+b}.$$
  Note that $\Phi$ is an increasing function.
  Let $\underline{a}$ be the lower bound on $\at$ in Equation~\eqref{eq:lowerbound_on_prob}, i.e.,
  $$\underline{a} \coloneqq a^{\sexp{1}}+\frac{t\rho}{(k-\ell+1)}\cdot \frac{a^{\sexp{1}}(b-1)}{4(a^{\sexp{1}}+b)}.$$
  Recall that $a^{\sexp{1}}$ and $b$ are fixed constants strictly larger than 1.
  To simplify notation, define a constant $C$ as $$C\coloneqq \frac{a^{\sexp{1}}(b-1)}{4(a^{\sexp{1}}+b)}>0.$$
  Now, we have
  \begin{align*}
    \ex{\betat} &= \expect_{\at}\insquare{\exgiv{\betat}{ \at=a}}\\
    &= \expect_{\at}\insquare{\Phi(\at)}\\
    & = \int_{\lb}^{\infty}\Phi(a)\ d\Pr[\at=a]\\ %
    &\geq  \Phi(\lb)\cdot \Pr[\at\negsp\geq\negsp \lb] \tag{$\Phi$ is increasing}\\
    &\hspace{-1mm} \stackrel{\eqref{eq:lowerbound_on_prob}}{\geq} \Phi(\lb)\cdot \inparen{1-\exp\inparen{-\frac{t\rho}{16}}}\yesnum\label{eq:used_later_in_generalization}\\
    &\stackrel{}{=} \frac{\lb}{\lb+b}\cdot \inparen{1-\exp\inparen{-\frac{t\rho}{16}}}\\
    &\stackrel{}{=} \frac{a^{\sexp{1}} + C\cdot \frac{t\rho}{(k-\ell+1)}}{a^{\sexp{1}} + b + C\cdot \frac{t\rho}{(k-\ell+1)}}   \cdot \inparen{1-\exp\inparen{-\frac{t\rho}{16}}}.\yesnum\label{eq:bound_on_at_proof_sm}
  \end{align*}
  Simplifying the above equation, and choosing $C_1\coloneqq \frac{C}{(a^{\sexp{1}}+b)}$ completes the proof.
\end{proof}

\subsubsection{Proof of Lemma~\ref{lem:rooney_rule_increases_a_fast}}\label{sec:proof:lem:rooney_rule_increases_a_fast}
The following proposition shows that conditioned on $\cE^\et$ and $\betat$ being ``small,'' $\delta$ is ``large.''
We use this proposition in the proof of \cref{lem:rooney_rule_increases_a_fast}.
\begin{proposition}[\bf Conditional lower bound on $\deltat$]\label{prop:delta_big_then_beta_small}
  For all iterations $t\in \N$ and values $\omt{} >0$, if $\ell>0$,
  then %
  \begin{align*}
    \cE^\et \ \land \ \inparen{\betat\leq 1-\frac{\omt{}\cdot (k-\ell+1)}{1+\omt{} }}\implies \inparen{\deltat \geq \omt{}}.\label{eq:eventb}\yesnum
  \end{align*}
\end{proposition}
\begin{proof}
  The basic idea is use $\cE^\et$ to upper bound $\nfrac{\uy}{\ux}$ by $(k-\ell)$; then one can substitute this bound, along with the bound on $\betat$, in the definition of $\deltat$ to get the the required result.
  Assume that the two events occur.
  Notice that the event in Equation~\eqref{eq:eventb} is equivalent to
  \begin{align*}
    \frac{(1-\betat)}{\betat+(k-\ell) }\geq \omt{}.\yesnum\label{eq:eventbb}
  \end{align*}
  Further, recall that event $\cE^\et$ is \text{\small $\inparen{Y_{(n_Y : n_Y)} ^\et\leq X_{(n_X : n_X)}^\et}$}.
  Combining this with
  $$\ux\geq X_{(n_X : n_X)}^\et\quad \text{and}\quad \uy\leq (k-\ell)\cdot Y_{(n_Y : n_Y)}^\et,$$ we get %
  \begin{align}
    \sfrac{\uy}{\ux}\leq (k-\ell).\label{eq:bounds_on_uxanduy}
  \end{align}
  Combining Equations~\eqref{eq:eventbb} and \eqref{eq:bounds_on_uxanduy}, we get
  \begin{align*}
    \delta^\et &\stackrel{\eqref{eq:def_of_delta}}{=} \frac{(1-\betat)}{\betat +\sfrac{\uy}{\ux}}
    \stackrel{\eqref{eq:bounds_on_uxanduy}}{\geq} \frac{(1-\betat)}{\betat+(k-\ell) }
    \stackrel{\eqref{eq:eventbb}}{\geq} \omt{}.
  \end{align*}
\end{proof}
\noindent Now, we are ready to prove \cref{lem:rooney_rule_increases_a_fast}.
\begin{proof}[Proof of Lemma~\ref{lem:rooney_rule_increases_a_fast}]
  First, we prove an upper bound on ${\rm median}(\betat)$, then the proof follows from \cref{prop:delta_big_then_beta_small}.
  To simplify notation let $$k^\prime\coloneqq (k-\ell+1).$$
  We divide this proof into 3 cases (as in Fact~\ref{fact:bounding_median_pdf}):\\

  \noindent {\bf Case A ($1<b<\at$):}
  Choose $w=\frac{b-1}{k^\prime (\at+b-2)-(b-1)}$.
  By a direct calculation
  \begin{align*}
    1-\frac{wk^\prime }{1+w} = \frac{\at-1}{\at+b-2} \stackrel{{\rm\cref{fact:bounding_median_pdf}}}{\geq}\ {\rm median}(\betat).\yesnum\label{eq:case_a_bound}
    \end{align*}\\

    \noindent {\bf Case B ($1\leq \at<b$):}
    Choose $w=\frac{1}{2k^\prime}$. %
    We have
    \begin{align*}
      1-\frac{wk^\prime}{1+w} = \frac{1+1/k^\prime}{2+1/k^\prime}>\frac{1}{2} \ \stackrel{(b>\at)}{>}\ \frac{a}{a+b}
      \stackrel{{\rm\cref{fact:bounding_median_pdf}}}{\geq}\ {\rm median}(\betat).\yesnum\label{eq:case_b_bound}
    \end{align*}
    \noindent {\bf Case C ($\at=b$):}
    Choose $w=\frac{1}{2k^\prime-1}$. Notice %
    \begin{align*}
      1-\frac{wk^\prime}{1+w} = 1-\frac{k^\prime}{2k^\prime} = \frac{1}{2}
      \stackrel{{\rm\cref{fact:bounding_median_pdf}}}{=}\ {\rm median}(\betat).\yesnum\label{eq:case_d_bound}
    \end{align*}
    \white{.}\\

    \noindent Thus, we can choose a $w\geq \frac{b-1}{k^\prime \cdot (\at+b)},$
    such that, in each case\footnote{To see this just use the fact that $\at,b>1$.}
    \begin{align}
      1-\frac{w(k-\ell+1)}{1+w}\geq {\rm median}(\betat).\label{eq:all_cases_bound}
    \end{align}
    \noindent Using this, we have
    \begin{align*}
      \hspace{-3mm}\cF^\et\iff \inparen{\betat\leq{\rm median}(\betat)} \stackrel{\eqref{eq:all_cases_bound}}{\implies} \inparen{\betat \leq 1-\frac{w(k-\ell+1)}{1+w}}. \yesnum\label{eq:implication_1}
    \end{align*}
    Now, combining Equation~\eqref{eq:implication_1} with \cref{prop:delta_big_then_beta_small}, we have
    \begin{align*}
      \hspace{-4mm}(\cE^\et\ \land \ \cF^\et) %
      & \quad \stackrel{}{\implies} \quad \inparen{\deltat \geq w}\tagnum{Using \eqref{eq:implication_1} and {\rm \cref{prop:delta_big_then_beta_small} }}\customlabel{eq:intermediate_result}{\theequation}\\
      &\quad \iff \inparen{a^{\sexp{t+1}}\geq \at+\at w}\\ %
      &\quad \stackrel{}{\implies} \inparen{a^{\sexp{t+1}}>\at+ \frac{\at}{\at+b}\cdot\frac{b-1}{(k-\ell+1)} }. \tagnum{$w\geq \frac{b-1}{k^\prime \cdot (\at+b)}$}
      \customlabel{eq:chain_statement_1}{\theequation}
    \end{align*}
    From \cref{fact:a_is_monotonic_in_t} we know that $\at\geq a^{\sexp{1}}$.
    Therefore
    \begin{align*}
      \frac{\at}{\at+b}\geq \frac{a^{\sexp{1}} }{a^{\sexp{1}}+b}.  \yesnum\label{eq:chain_statement_2}
    \end{align*}
    Combining Equation~\eqref{eq:chain_statement_1} and Equation~\eqref{eq:chain_statement_2}, we get
    \begin{align*}
      (\cE^\et\land \cF^\et)\implies \inparen{a^{\sexp{t+1}}>\at+ \frac{a^{\sexp{1}}(b-1)}{(k-\ell+1)(a^{\sexp{1}}+b)} }.
    \end{align*}
  \end{proof}

  \vspace{-2mm}
  \subsection{Proof of Theorem~\ref{thm:no_rooney_rule}}\label{sec:proof:thm:no_rooney_rule}
  In this section, we prove a generalization Theorem~\ref{thm:no_rooney_rule} which we presented in \cref{sec:thm:no_rooney_rule_gen}.
  The generalization holds for a large class of distributions $\cP$ including ${\rm Unif}(0,1)$.
  Specializing the proof in this section to the case when $\cP\coloneqq {\rm Unif}(0,1)$ proves Theorem~\ref{thm:no_rooney_rule}.

  To state the assumption on $\cP$ (which was also mentioned in Section~\ref{sec:thm:no_rooney_rule_gen}), we need some additional notation.
  Let $M$ be the supremum of the support of $\cP$: $M\coloneqq \sup(\supp(\cP))$.
  Define $T_\cP(\eps)$ as the probability that an $X$ drawn from $\cP$ is at least $(1-\eps) M$.
  Formally,
  $T_{\text{\footnotesize $\cP$}}(\eps)\coloneqq \Pr_{\text{\footnotesize $X\negsp\sim\negsp\cP$}}[M\geq (1-\eps)M].$
  We make the following assumption of $\cP$.
  \begin{mdframed}[style=MyFrame]
    \noindent {\bf Assumption.} Assume that for each $n\in\N$ there is an $\eps>0$ satisfying
    \begin{align}
      \eps+e^{-n\cdot T_\cP(\eps)} \leq \frac{1}{\poly(n)}.\tagnum{Assumption}\customlabel{eq:assumption_1_app}{\theequation}
    \end{align}
    \vspace{-4mm}

    \noindent For instance, when $\cP={\rm Unif}(0,1)$, we have $T_\cP(\eps)=\eps$ and we can choose $\eps=\frac{\ln n}{n}$.
  \end{mdframed}
  Note that the above assumption holds for several common distributions, including, any continuous distribution with a compact interval support.
  This also includes many common truncated distributions (e.g., truncated powerlaw and truncated normal distributions); see Section~\ref{sec:facts_related_to_assump}.

  Formally, we would prove the following generalization.
  \begin{theorem}[\textbf{Slow learning without the Rooney Rule}]\label{thm:no_rooney_rule_generalized}
    Under the dynamics of implicit bias~\eqref{alg:iterative_selection_and_panel_learning},
    given parameters $a^1,b > 1$,
    for all iterations $t\in \N$, number of candidates selected $k\in \N$, ratio of the underrepresented candidates $\rho\in (0,1)$,
    and the distribution of latent utilities $\cP$ satisfying Assumption~\eqref{eq:assumption_1_app},
    when the panel is not constrained by the Rooney Rule (i.e., $\ell=0$) %
    there exists a $n_0\in \N$ and $C>0$, such that for all $n\geq n_0$
    \begin{align}
      \expect\big[\beta^{\sexp{t}}\big] &\leq \expect\big[\beta^{\sexp{1}}\big] + C\cdot\frac{t}{\poly{(n)}},\label{eq:lowerbound_thm1}
    \end{align}
    where the expectations are over draws of implicit bias $\beta^{\sexp{s}}$ and latent utilities of candidates in all previous iterations $s\in [t-1]$.
    Here, $\poly(\cdot)$ hides the dependence on $(1-\rho)$.
  \end{theorem}
  \noindent We do not state the dependence on $(1-\rho)$ as it changes with the choice of $\cP$.
  But in the proof, we detail how to make the dependence explicit when $\cP\coloneqq {\rm Unif}(0,1)$, and so proving \cref{thm:no_rooney_rule}.

  Given a value of $n\in \N$, let $\eps(n)>0$ be number promised in Assumption~\eqref{eq:assumption_1_app}.
  To simplify notation, we write $\eps$ instead of $\eps(n)$ when the dependence on $n$ is clear from the context.
  Also, for all $t\in \N$, let $D^\et$ denote the values of all random variables for the first $[t-1]$ iterations:
  $$D^\et\coloneqq \inbrace{\beta^\es}_{1\leq s<t}\cup \inbrace{X^\es}_{1\leq s<t}\cup \inbrace{Y^\es}_{1\leq s<t}.$$
  We first present three lemmas which are used in the Proof of Theorem~\ref{thm:no_rooney_rule_generalized}.
  \begin{lemma}[\bf Upper bound on $\deltat$ when $\betat$ is close to 1]\label{lem:delta_small_when_beta_large}
    For all iterations $t\in \N$, values of $D^\et$, and constant $\eps\in (0,\nfrac{1}{2})$, if $\betat\geq 1-\eps$, then $\deltat\leq 2\eps$.
  \end{lemma}
  \begin{proof}
    In Section~\ref{sec:preliminaries} we saw that $\deltat$ is decreasing in $\uy$ and increasing in $\ux$.
    In the worst case we have $\uy=0$.
    Substituting this and $\betat\geq 1-\eps$ in  Equation~\eqref{eq:def_of_delta} we have
    $$\deltat\ \stackrel{\eqref{eq:def_of_delta}}{=}\ \ \frac{(1-\betat)\cdot\ux}{\betat\ux+\uy}\leq\frac{(1-\betat)}{\betat}\leq\frac{\eps}{1-\eps}\ \stackrel{(\eps<\frac12)}{\leq}\ \
    2\eps.$$
  \end{proof}
  \noindent The proofs of Lemmas~\ref{lem:delta_bounded_in_expect} and Lemma~\ref{lem:delta_0_whp} are longer and
  we defer them to Section~\ref{sec:proof:lem:delta_bounded_in_expect} and Section~\ref{sec:proof:lem:delta_0_whp}.
  \begin{lemma}[\bf Upper bound on $\ex{\deltat}$ when $\betat$ is not close to 1]\label{lem:delta_bounded_in_expect}
    For all iterations $t\in\N$, values of $D^\et$, constant $\eps\in (0,\nfrac{1}{2})$, and parameters $a^{\sexp{1}},b>1$, it holds that
    $$\prob{\betat< 1-\eps}\cdot \exgiv{\deltat}{ (\betat< 1-\eps)\land (\ux\neq 0)}\leq \frac{b}{a^{\sexp{1}}-1}.$$
  \end{lemma}
  \noindent Note that the next lemma requires $\ell=0$.
  \begin{lemma}[\bf Lower bound on $\Pr\insquare{\ux=0}$ when $\betat$ is not close to 1]\label{lem:delta_0_whp}
    For all iteration $t\in \N$ and values of $D^\et$, if $\ell=0$, then there exist a constant $n_0\in \N$, such that,
    for all $n\geq n_0$ and corresponding $\eps(n)$ from Assumption~\eqref{eq:assumption_1_app}, it holds that
    \begin{align*}
      \probgiv{\ux=0}{ \betat < 1-\eps(n) }\geq 1-\exp\inparen{-\frac{1}{8}n(1-\rho)\cdot T_\cP(\eps(n)))}.
    \end{align*}
  \end{lemma}
  \noindent Note that the bounds in \cref{lem:delta_small_when_beta_large}, \cref{lem:delta_bounded_in_expect}, and \cref{lem:delta_0_whp} hold for all values of $D^\et$.

  Now, we are ready to prove the Theorem~\ref{thm:no_rooney_rule_generalized}.
  \begin{proof}[Proof of Theorem~\ref{thm:no_rooney_rule_generalized}]
    Fix a large enough $n_0$ so that for all $n\geq n_0$, $\eps(n)<\frac{1}{2}$.
    This is possible, since $\eps(n)\leq c_5 \cdot n^{-c_4}$ from Assumption~\eqref{eq:assumption_1_app}.

    We first upper bound $\exgiv{\delta^\es}{D^\es}$ for all values of $D^\es$ and $s\in [t-1]$.
    Then, combining this with Equation~\eqref{eq:update_app} would give us an upperbound on $\at$.
    To simplify notation, we do not explicitly write the conditioning on $D^\es$ in the RHS below.
    \noindent We have
    \begin{align*}
      \exgiv{\delta^\es}{D^\es} &= \exgiv{\delta^\es}{(\beta^\es\geq 1-\eps)}\cdot \Pr[\beta^\es\geq 1-\eps]\\
      &\quad +\left(\begin{array}{l}\text{$\displaystyle\ex{\delta^\es \mid (\beta^\es< 1-\eps)\land (\uxs\neq 0)}$} \\[1mm] \cdot \prob{\uxs\neq 0 \mid \beta^\es< 1-\eps}\cdot \Pr[\beta^\es< 1-\eps]
      \end{array}\right)\\
      &\quad +\left(\begin{array}{l}\ex{\delta^\es\mid (\beta^\es< 1-\eps)\land (\uxs=0)} \\[1mm] \cdot \Pr[\uxs= 0\mid \beta^\es< 1-\eps]\cdot \Pr[\beta^\es< 1-\eps]
      \end{array}\right)\\[2mm]
      &\leq 2\eps\tag{Lemma~\ref{lem:delta_small_when_beta_large}}\\
      &\quad + \inparen{\factor \cdot \prob{\uxs\neq 0 \mid \beta^\es< 1-\eps}}\tag{Lemma~\ref{lem:delta_bounded_in_expect}}\\
      &\quad + 0 \tag{Using $(\uxs=0)\implies (\delta^\es=0)$}\\[2mm]%
      &\hspace{-6mm}\stackrel{{\rm Lemma}~\ref{lem:delta_0_whp}}{\leq}\ 2\eps + \factor \cdot e^{-\frac18\cdot n(1-\rho) T_\cP(\eps)}\\
      &\leq\quad\quad 2\eps + \factor \cdot e^{-\frac18\cdot n(1-\rho) T_\cP(\eps)}.\yesnum\label{eq:upperbound_on_expec_deltat}
    \end{align*}
    \noindent (Note that we can use \cref{lem:delta_small_when_beta_large}, \cref{lem:delta_bounded_in_expect}, and \cref{lem:delta_0_whp} as they hold for all values of $D^\es$.)
    Equation~\eqref{eq:upperbound_on_expec_deltat} implies an upper bound $\ex{\delta^{t_1}\cdot \delta^{t_2}}$ for some $t_1,t_2\in \N$.
    To see this, let $t_1 < t_2$, and note that
    \begin{align*}
      \ex{\delta^{t_1}\cdot \delta^{t_2}} &= \int_{0}^{\infty}x\cdot \exgiv{\delta^{t_2}}{ \delta^{t_1}=x} \cdot \Pr[\delta^{t_1}=x]dx\\
      &\stackrel{\eqref{eq:upperbound_on_expec_deltat}}{\leq} \int_{0}^{\infty}x\cdot  \inparen{2\eps + \factor\cdot e^{-\frac18\cdot n(1-\rho) T_\cP(\eps)}}  \cdot \Pr[\delta^{t_1}=x]dx\\
      &= \inparen{2\eps + \factor\cdot e^{-\frac18\cdot n(1-\rho) T_\cP(\eps)}} \cdot \int_{0}^{\infty}x \cdot \Pr[\delta^{t_1}=x]dx\\
      &\leq \inparen{2\eps + \factor\cdot e^{-\frac18\cdot n(1-\rho) T_\cP(\eps)}}\cdot\ex{\delta^{t_1}}\\
      &\stackrel{\eqref{eq:upperbound_on_expec_deltat}}{\leq} \inparen{2\eps + \factor\cdot e^{-\frac18\cdot n(1-\rho) T_\cP(\eps)}}^2.
    \end{align*}
    \noindent Similarly extending this to $t_1<t_2<\dots<t_s$, we have
    \begin{align*}
      \ex{\delta^{t_1}\cdot \delta^{t_2}\cdots \delta^{t_s}}\leq \inparen{2\eps + \factor\cdot e^{-\frac18\cdot n(1-\rho) T_\cP(\eps)}}^s.\yesnum\label{eq:expect_product_delta_sm}
    \end{align*}
    Next, we bound the expectation of $\at$.
    \begin{align*}
      \ex{\at} &\stackrel{\eqref{eq:update_app}}{=} a^{\sexp{1}}\cdot \ex{\prod_{s=1}^{t-1}\inparen{1+\delta^\es}}\\
      &\stackrel{\eqref{eq:expect_product_delta_sm}}{\leq} a^{\sexp{1}}\cdot\inparen{1+2\eps + \factor\cdot e^{-\frac18\cdot n(1-\rho) T_\cP(\eps)}}^t.\yesnum\label{eq:bound_on_at_in_middle}
    \end{align*}
    Where the last inequality follows by expanding $\prod_{i=0}^\et\inparen{1+\deltat}$, using linearity of expectation and upper bounding each term with Equation~\eqref{eq:expect_product_delta_sm}, and finally re-factorizing the resulting expression.\footnote{Alternatively, we can repeat the above derivation for $\ex{(1+\delta^{t_1})\cdot (1+\delta^{t_2})\cdots (1+\delta^{t_k})}$.}

    From Assumption~\eqref{eq:assumption_1_app}, we have
    $\eps+e^{-nT_\cP(\eps)} \leq \frac{1}{\poly(n)}.$
    This implies that $\max\sinparen{\eps, e^{-nT_\cP(\eps)}}\leq \frac{1}{\poly(n)}$.
    Thus, we can pick a large enough $n_0\in \N$, such that, for all $n\geq n_0$
    $$\inparen{2\eps + \factor\cdot e^{-\frac18\cdot n(1-\rho) T_\cP(\eps)}}\leq \frac1t.$$
    \noindent Substituting this in Equation~\eqref{eq:bound_on_at_in_middle} we have
    \begin{align*}
      \ex{\at}&\leq a^{\sexp{1}}\cdot\inparen{1+ 2\eps + \factor\cdot e^{-\frac18\cdot n(1-\rho) T_\cP(\eps)}}^t\\
      &\leq a^{\sexp{1}} + 2ta^{\sexp{1}}\inparen{2\eps + \factor\cdot e^{-\frac18\cdot n(1-\rho) T_\cP(\eps)}}\tagnum{For all $t>0$ and $x\in (0,\nfrac1t)$, $(1+x)^t\leq 1+2tx$}\customlabel{eq:bound_intermediate}{\theequation}\\
      &\stackrel{}{\leq} a^{\sexp{1}} + \frac{t}{\poly(n)} .\tagnum{Assumption~\eqref{eq:assumption_1_app}}\customlabel{eq:upperbounud_on_at}{\theequation}
    \end{align*}
    The last inequality follows since $\eps$ and $e^{-n T_\cP(\eps)}$ are bonded by $\frac{1}{\poly(n)}$. (Note that the $\poly(n)$ factor in $\eqref{eq:upperbounud_on_at}$ maybe different from that in Assumption~\eqref{eq:assumption_1_app}.)
    As we will show, Equation~\eqref{eq:upperbounud_on_at} suffices to prove the theorem.
    Here, we note that when $\cP$ is ${\rm Unif}(0,1)$ we prove the following version of Equation~\eqref{eq:upperbounud_on_at}.
    \begin{mdframed}[style=FrameBox2]
      \begin{remark}\label{rem:bound_when_unif}
        When $\cP$ is ${\rm Unif}(0,1)$, substituting $\eps\coloneqq \frac{8\ln{n}}{n(1-\rho)}$ in Equation~\eqref{eq:bound_intermediate} gives us
        $$\ex{\at}\leq a^{\sexp{1}} + 2t\uniffactorb{}.$$
      \end{remark}
    \end{mdframed}
    Before proceeding, we need some additional notation.
    Define $\Phi(a)$ as the expectation of $\betat$ when $\at=a$: %
    $$\Phi(a)\coloneqq \exgiv{\betat}{ \at=a}= \frac{\at}{\at+b}.$$
    Note that $\Phi$ is an increasing and concave function.

    Now, using that $\Phi$ is a concave function, we have
    \begin{align*}
      \expect\sinsquare{\betat} %
      &= \expect_{\at}\sinsquare{\Phi(\at)}
      \ {\leq} \ \Phi\inparen{\expect_{\at}\sinsquare{\at}}. \tag{\text{$\Phi$ is concave}}
    \end{align*}
    Further, using that $\Phi$ is an increasing function and Equation~\eqref{eq:upperbounud_on_at} we have
    \begin{align*}
      \ex{\betat}&\leq  \Phi\inparen{\expect_{\at}\sinsquare{\at}}\\
      &\leq \Phi\inparen{a^{\sexp{1}} + \frac{t}{\poly(n)}}\tag{\text{$\Phi$ is increasing and Equation~\eqref{eq:upperbounud_on_at}}}\\
      &\leq \frac{a^{\sexp{1}} + \frac{t}{\poly(n)}}{a^{\sexp{1}} + b + \frac{t}{\poly(n)}}\\
      &\leq \frac{a^{\sexp{1}} + \frac{t}{\poly(n)}}{a^{\sexp{1}} + b}\cdot
      \frac{1}{1+\frac{t}{a^{\sexp{1}}+b}\frac{1}{\poly(n)}}\\
      &\leq \frac{a^{\sexp{1}} + \frac{t}{\poly(n)}}{a^{\sexp{1}} + b}\\
      &\leq \frac{a^{\sexp{1}}}{a^{\sexp{1}}+b} + \frac{t}{a^{\sexp{1}} + b}\frac{1}{\poly(n)}\\
      &\leq \frac{a^{\sexp{1}}}{a^{\sexp{1}}+b} + \frac{t}{\poly(n)}.
    \end{align*}
    This proves \cref{thm:no_rooney_rule_generalized}.

    We claim that \cref{thm:no_rooney_rule} follows by substituting the bound from \cref{rem:bound_when_unif} in the above equation.
    To see this, note that
    \begin{align*}
      \ex{\betat}&\leq \Phi\inparen{\expect_{\at}\sinsquare{\at}}\\
      &\leq \frac{a^{\sexp{1}} + 2t\uniffactorb{}}{a^{\sexp{1}} + b + 2t\uniffactorb{}}\tag{Using \cref{rem:bound_when_unif}}\\
      &\leq \inparen{a^{\sexp{1}} + 2t\uniffactorb{}}\cdot \frac{1}{a^{\sexp{1}} + b} \cdot
      \inparen{{1+\frac{2t}{a^{\sexp{1}}+b}\uniffactorb{}}}^{-1}\\
      &\leq \frac{a^{\sexp{1}} + 2t\uniffactorb{}}{a^{\sexp{1}} + b}\\
      &\leq \frac{a^{\sexp{1}}}{a^{\sexp{1}}+b} + \frac{2t}{a^{\sexp{1}} + b}\uniffactorb{}\\
      &\leq \frac{a^{\sexp{1}}}{a^{\sexp{1}}+b} + \frac{32t a^{\sexp{1}}\ln{n}}{n(1-\rho)(a^{\sexp{1}}-1)}.\yesnum\label{eq:upperbonud_used_later}
    \end{align*}
    Now \cref{thm:no_rooney_rule} follows by choosing $C_2\coloneqq \frac{32a^{\sexp{1}}}{a^{\sexp{1}}-1}$.
  \end{proof}

  \subsubsection{Proof of Lemma~\ref{lem:delta_bounded_in_expect}}\label{sec:proof:lem:delta_bounded_in_expect}
  \begin{proof}
    In Section~\ref{sec:preliminaries} we saw that $\delta^\et$ is decreasing in $\uy$ and increasing in $\ux$.
    In the worst case we have $\uy=0$.
    Assume that $\ux\neq 0$.
    Substituting this in Equation~\eqref{eq:update_app}, we get
    \begin{align*}
      \delta^\et\ &\stackrel{}{=}\ \  \frac{(1-\beta^\et)\cdot \ux}{\beta^\et\ux+\uy} \leq \frac{(1-\beta^\et)}{\beta^\et}.\yesnum\label{eq:delta_bounded_in_expect_bound_on_delta}
    \end{align*}
    Let ${\rm pdf}\colon [0,1] \to [0,1]$ be the probability density function of $\Beta(\at,b)$.
    Fix any value of $D$ (the proof is independent of this value).
    Taking the expectation we have
    \begin{align*}
      \exgiv{\delta^\et}{(\beta^\et< 1-\eps)\land (\ux\neq 0)}
      & \stackrel{\eqref{eq:delta_bounded_in_expect_bound_on_delta}}{\leq} \exgiv{\nfrac{1}{\beta^\et}-1}{ \beta^\et< 1-\eps\land (\ux\neq 0)}\\
      & = \ \ \exgiv{\nfrac{1}{\beta^\et}-1}{ \beta^\et< 1-\eps} \\
      & =\int_{0}^{1} \inparen{\frac{1}{x}-1}\cdot \probgiv{\beta^\et=x}{ \beta^\et<1-\eps}\ dx\\
      & =\int_{0}^{1} \inparen{\frac{1}{x}-1}\cdot \frac{\Pr[\beta^\et=x\ \land \ \beta^\et<1-\eps]}{\Pr[\beta^\et<1-\eps]}\ dx\\
      & =\int_{0}^{1-\eps} \inparen{\frac{1}{x}-1}\cdot \frac{\Pr[\beta^\et=x\ \land\ \beta^\et<1-\eps]}{\Pr[\beta^\et<1-\eps]}\ dx\\
      & =\frac{1}{\Pr[\beta^\et< 1-\eps]}\cdot \int_{0}^{1-\eps} \inparen{\frac{1}{x}-1}\cdot {\rm pdf}(x)\ dx \\
      & \leq \frac{1}{\Pr[\beta^\et< 1-\eps]}\cdot \int_{0}^{1} \inparen{\frac{1}{x}-1}\cdot {\rm pdf}(x)\ dx \\
      & \leq \frac{\ex{\nfrac{1}{\beta^\et}-1 }}{\Pr[\beta^\et< 1-\eps]} \yesnum\label{eq:used_later_in_generalization_2}\\
      &=\frac{1}{\Pr[\beta^\et< 1-\eps]}\cdot \inparen{\frac{a^\et+b-1}{a^\et-1}-1}\tag{$\av_{\beta\sim \Beta(a,b)}\insquare{\frac{1}{\beta}}=\frac{a+b-1}{(a-1)}$~\cite[Eq~25.15]{balakrishnan2016continuous}}\\
      & =\frac{1}{\Pr[\beta^\et< 1-\eps]}\cdot \frac{b}{a^\et-1}\\
      & \hspace{-4mm}\stackrel{{\rm\cref{fact:a_is_monotonic_in_t}}}{\leq}\ \frac{1}{\Pr[\beta^\et< 1-\eps]}\cdot\frac{b}{a^{\sexp{1}}-1}.
    \end{align*}
  \end{proof}
  \subsubsection{Proof of Lemma~\ref{lem:delta_0_whp}}\label{sec:proof:lem:delta_0_whp}
  \begin{proof}[Proof of Lemma~\ref{lem:delta_0_whp}]
    Fix any value of $D^\et$ (the proof is independent of this value).
    Recall that $M$ is the supremum of the support of $\cP$, i.e., $M\coloneqq \sup(\supp(\cP))$ and $T_\cP(\eps)\coloneqq \Pr_{X\sim \cP}\insquare{X \geq (1-\eps)\cdot M}.$
    Further, from Assumption~\eqref{eq:assumption_1_app} we know that for all $n\geq n_0$
    $$e^{-n\cdot T_\cP(\eps)}\leq \frac{1}{\poly(n)}.$$
    Choose $n_0$ large enough such that
    \begin{align*}
      nT_\cP(\eps)\geq \frac{2k}{1-\rho}.\label{eq:expectation_of_zsum_is_large}\yesnum
    \end{align*}

    \noindent Assume that $\betat < (1-\eps)$.
    Then in the $t$-th iteration, all $X$-candidates have a latent utility of at most $(1-\eps)M$.
    For each $Y$-candidate $j\in [n_Y]$, let $Z_i$ be the indicator random that $(\yjt \geq (1-\eps) M)$.
    Since for each $j$, $\yjt$ are drawn independently, $\{Z_j\}_j$ are independent random variables.
    One can verify that
    $\ex{Z_i}=T_\cP(\eps).$
    Further, by linearity of expectation we have
    \begin{align*}
      \ex{\sum\nolimits_{j=1}^{n(1-\rho)}Z_j} = n(1-\rho)\cdot T_\cP(\eps).\yesnum\label{eq:expectation_of_zsum}
    \end{align*}
    Notice that if $Z_j=1$, then the $j$-th $Y$-candidate has a higher utility than all $X$-candidates.
    Thus, if $\sum_{j\in [n_Y]}Z_j \geq k$, then no $X$-candidate would be selected in the $t$-th iteration, implying that $\ux=0$.
    Now, we lower bound the probability that this happens as follows
    \begin{align*}
      \probgiv{\ux=0}{ \betat < 1-\eps } &\geq \Pr\insquare{\sum\nolimits_{j=1}^{n(1-\rho)}Z_j \geq k\mid \betat < 1-\eps }\\
      &= \Pr\insquare{\sum\nolimits_{j=1}^{n(1-\rho)}Z_j \geq k}.
    \end{align*}
    Where the last equality follows since $Z_i$ are independent of $\betat$.
    Using Equation~\eqref{eq:expectation_of_zsum_is_large}, we get
    $\frac12 n(1-\rho)T_\cP(\eps) \geq k.$
    Substituting this, and using Chernoff bound~\cite{motwani1995randomized}, we get
    \begin{align*}
      \Pr\insquare{\ux=0\mid \betat < 1-\eps }& \geq   \Pr\insquare{\sum\nolimits_{j=1}^{n(1-\rho)}Z_j \geq \inparen{1-\frac{1}{2}}\cdot n(1-\rho) T_\cP(\eps)}\\
      & \stackrel{\eqref{eq:expectation_of_zsum}}{=}  \Pr\insquare{\sum\nolimits_{j=1}^{n(1-\rho)}Z_j
      \geq \ex{\sum\nolimits_{j=1}^{n(1-\rho)}Z_j}}\\
      &\geq 1-\exp\inparen{-\frac{1}{2\cdot 2^2}\cdot n(1-\rho) T_\cP(\eps)}. \tag{Using Chernoff bound~\cite{motwani1995randomized}}
    \end{align*}
  \end{proof}

  \subsection{Fact related to Assumption~\eqref{eq:assumption_1}}\label{sec:facts_related_to_assump}
  \begin{fact}\label{fact:asmp_holds_1}
    If $\cP$ has a continuous probability density function and ${\rm sup}(\supp(\cP))\in \supp(\cP)$ then $\cP$ satisfies
    Assumption~\eqref{eq:assumption_1}. %
  \end{fact}
  \begin{proof}
    Let $M\coloneqq {\rm sup}(\supp(\cP))$.
    Let ${\rm pdf}_{\cP}$ be the probability density function of $\cP$ and $c={\rm pdf}_{\cP}(M)>0$.
    Since the pdf of $\cP$ is continuous, there exists a $\eta>0$ such that $x\in [(1-\eta)M,M]$, ${\rm pdf}_{\cP}[x]>\frac{c}{2}$ (say).
    Then for any $\eps\leq \eta$, $T_{\cP}(\eps)\geq \frac{cM\eps}{2}$.
    Now, setting $\eps(n) \coloneqq \frac{\eta \ln{(n)}}{n}$, we have
    \begin{align*}
      \eps(n)+e^{-n\cdot T_\cP(\eps(n))} \leq \frac{2\eta \ln{(n)}}{n} + e^{-\frac{c\eta M}{2}\ln{(n)}} = O\inparen{n^{-\frac{c\eta M}{2}}+\frac{\ln{n}}{n}}.
    \end{align*}
    Note that fixing $\cP$ (and so, $M,\eta,$ and $c$), $O\inparen{n^{-\frac{c\eta M}{2}}+\frac{\ln{n}}{n}}$ is $\frac{1}{\poly(n)}$.
  \end{proof}
  \cref{fact:asmp_holds_1} implies that Assumption~\eqref{eq:assumption_1} holds for several common distributions, including, any continuous distributions with a compact interval support, and for truncated normal and exponential distributions.

  \section{Proofs of generalizations of theoretical results}\label{sec:formal_statements_of_remarks}
  In this section, we present the formal statements and proofs of the generalizations of our theoretical results (which were discussed in Section~\ref{sec:technical_remarks}).
  Section~\ref{sec:proof:generalization_for_other_distributions} presents the generalization to other distributions of implicit bias.
  Section~\ref{sec:proof:generalization_for_other_update_rules} presents the generalization to other updates rules.
  Finally, Section~\ref{sec:proof:convergestoone} proves a result mentioned in Section~\ref{sec:theoretical_results}.
  In particular, that without the Rooney Rule, for a fixed $n$ as $t\to\infty$, the panel becomes unbiased.

  The proofs in Sections \ref{sec:proof:generalization_for_other_distributions} and \ref{sec:proof:generalization_for_other_update_rules} are similar to the proofs in Section~\ref{sec:proofs}.
  Instead of repeating the entire proofs, we highlight how the proofs in Sections \ref{sec:proof:generalization_for_other_distributions} and \ref{sec:proof:generalization_for_other_update_rules}  differ from those in Section~\ref{sec:proofs}.

  \subsection{Generalization to other distributions of implicit bias}\label{sec:proof:generalization_for_other_distributions}
  In this section, we present the formal statements of our results for other distributions of implicit bias.
  Recall the following setting from \cref{sec:generalization_for_other_distributions}.

  \paragraph{Notation.} Let $\cD(a)$ be any continuous distribution supported on $[0,1]$ and parameterized by $a>1$. %
  Formally, we consider the following dynamics of implicit bias:
  \setlength{\algomargin}{0.5em}
  \begin{algorithm}[h!]
    \SetAlgoNoEnd
    \caption{A dynamics for implicit bias}
    \label{alg:iterative_selection_and_panel_learning_2}
    \kwInit{The Rooney Rule parameter $\ell\in \Z_{\geq 0}$, a parameter $a^{\sexp{1}}>1$, and a constant $b>1$.}\vspace{2mm} %
    \For{$t=1,2,\dots$}{
    \vspace{1mm}
    {\bf Sample} $\sinbrace{Y_j^{\sexp{t}}}_{j\in G_Y}$ and $\sinbrace{X_i^{\sexp{t}}}_{i\in G_X}$ i.i.d. from $\cP$.\vspace{1mm}

    {\bf Sample} $\betat$ from $\D(a^{\sexp{t}})$.\hspace{18.75em} \commentalg{Implicit bias} \vspace{1mm}

    {\bf Define} $\wt{X}_i^{\sexp{t}}  \coloneqq \betat\cdot X_i^{\sexp{t}}$ for all $i\in G_X$.\hspace{13em}\commentalg{Observed utilities}\vspace{1mm}

    {\bf Select} $S^{\sexp{t}}_\ell\coloneqq \argmax_{S\in \mathcal{R}(\ell)}\  \util{}(S, \wt{X}^{\sexp{t}}, {Y}^{\sexp{t}})$.\hspace{9.4em}\commentalg{Shortlist}\vspace{1mm}\\
    {\bf Let} $U^{\sexp{t}} \coloneqq \util{}(S^{\sexp{t}}_\ell,X,Y)$ and $\wt{U}^{\sexp{t}}\coloneqq \util{}(S^{\sexp{t}}_\ell,\wt{X},{Y})$.\vspace{1mm}

    {\bf Update}  $a^{\sexp{t+1}} \coloneqq \inparen{\nfrac{U^{\sexp{t}}}{\wt{U}^{\sexp{t}}}} \cdot  a^{\sexp{t}}$.\hspace{15.9em}\commentalg{Update implicit bias}\vspace{1mm}
    }
  \end{algorithm}
  Note that the only difference from dynamics~\ref{alg:iterative_selection_and_panel_learning} is that $\cD(\at)$ is not a beta distribution.
  Define $\Phi(x)$ as the expected value of $\beta$ drawn from $\cD(x)$:
  $$\Phi(x)\coloneqq \expect_{\beta\sim \cD(x)}\sinsquare{\beta},$$
  and ${\rm median}(a)$ as the median of $\beta\sim\cD(a)$
  For example, we can consider $\cD(a)$ to be the truncated normal distribution with mean $a$ and fixed variance, in which case, ${\rm median}(a)=a$.\\[-1mm]

  Assume that
  \begin{enumerate}[leftmargin=*]
    \item $\Phi(\cdot)$ is an increasing function,
    \item $\Phi(\cdot)$ is concave,
    \item $\expect_{\beta\sim\cD(a)}\insquare{\frac{1}{\beta}}$ is decreasing in $a$ and finite for every $a>1$, and
    \item there is a constant $C_3>0$, such that, for all $a\negsp >\negsp 1,\ k\negsp \in \negsp [n], \ell\negsp \in\negsp  [k]$
    \begin{align*}
      \frac{1-{\rm median}(a)}{{\rm median}(a)+(k-\ell)}>\frac{C_3}{a \cdot (k-\ell+1)}.%
    \end{align*}
  \end{enumerate}
  Then, we prove the following versions of \cref{thm:rooney_rule} and \cref{thm:no_rooney_rule}:

  \begin{lemma}[\textbf{Fast learning with the Rooney Rule for other distributions of implicit bias}]
    Under the dynamics of implicit bias~\eqref{alg:iterative_selection_and_panel_learning_2},
    for all $a^{\sexp{1}}> 1$ there is a constant $C_4>0$, such that
    for all $t\in \N,\ n\in \N,\ k\in [n],\ \ell\in [k]$, and continuous and bounded distribution $\cP$,
    when the panel applies the Rooney Rule, the following holds
    \begin{align}
      \expect\sinsquare{\betat}\geq \Phi\inparen{a^{\sexp{1}}+\frac{C_3}{C_4}\frac{t\rho}{(k-\ell+1)}}\cdot \inparen{1-e^{\frac{t\rho}{16}}},\label{eq:gen_lowerbound_sm}
    \end{align}
    where the expectation is over the draws of implicit bias $\beta^{\sexp{s}}$ and latent utilities of candidates in all previous iterations $s\in [t-1]$.
  \end{lemma}
  \begin{proof}
    First, we prove a lower bound on $\expect[\at]$ (similar to \cref{lem:rooney_rule_increases_a_fast}).
    To do so, note that the proof of \cref{lem:rooney_rule_increases_a_fast} only uses the following property of the beta distribution:
    For $w = \frac{b-1}{(k-\ell+1) \cdot (\at+b)}$, it holds
    \begin{align}
      1-\frac{w(k-\ell+1)}{1+w}\geq {\rm median}(\betat).
    \end{align}
    Notice that this is equivalent to
    \begin{align}
      \frac{1-{\rm median}(\betat)}{{\rm median}(\betat)+(k-\ell)}\geq \omt.%
    \end{align}
    Assumption (4) gives an analogous statement for $\omt{}=\frac{C_3}{\at\cdot (k-\ell+1)}$.
    Then, following the same arguments as in proof of \cref{lem:rooney_rule_increases_a_fast} and using the events $\cE^\et$ and $\cF^\et$ defined there,
    we can show that\\
    \begin{mdframed}[style=FrameBox]
      For all iterations $t\in \N$, if $\ell>0$, then $(\cE^\et\land \cF^\et)$ implies
      \begin{align*}
        \inparen{a^{\sexp{t+1}}>\at+ \frac{C_3}{(k-\ell+1)} }.\yesnum\label{eq:a_increases_by_const_sm}
      \end{align*}
    \end{mdframed}
    Next, using the Chernoff bound~\cite{motwani1995randomized} (and following the same steps), we can establish the following analogue version of Equation~\eqref{eq:lowerbound_on_prob}
    \begin{align*}
      \Pr\insquare{\at \geq a^{\sexp{1}}+\frac{t\rho}{4}\frac{C_3}{(k-\ell+1)}}\geq 1-\exp\inparen{-\frac{t\rho}{16}}.
    \end{align*}
    Finally, using the fact that $\Phi$ is increasing, we can prove the analogue of \eqref{eq:used_later_in_generalization} with the same steps.
    This gives us the required result.
  \end{proof}

  \begin{lemma}[\textbf{Slow learning without the Rooney Rule for other distributions of implicit bias}]
    Under the dynamics of implicit bias~\eqref{alg:iterative_selection_and_panel_learning_2},
    for all $a^{\sexp{1}}> 1$ there is a constant $C_5>0$, such that
    for all $t\in \N,\ k\in \N,\ \ell\in [k]$,
    when the panel applies the Rooney Rule and $\cP\coloneqq {\rm Unif}(0,1)$,
    there is a constant $n_0\in \N$, such that for all $n\geq n_0$
    the following holds
    \begin{align}
      \expect\sinsquare{\betat}\leq \Phi\inparen{a^{\sexp{1}}+C_5{\frac{t\ln{n}}{n(1-\rho)}}}+    {\frac{t\ln{n}}{n(1-\rho)}},\label{eq:gen_upperbound}
    \end{align}
    where the expectation is over draws of implicit bias $\beta^{\sexp{s}}$ and latent utilities of candidates in all previous iterations $s\in [t-1]$.
  \end{lemma}
  \begin{proof}
    This proof uses arguments similar to the proof of \cref{thm:no_rooney_rule}.
    First, we prove an upper bound on $\expect[\delta \mid D]$ (similar to Equation~\eqref{eq:upperbound_on_expec_deltat}).
    To so so, note that the proof of Equation~\eqref{eq:upperbound_on_expec_deltat} uses three lemmas: \cref{lem:delta_small_when_beta_large}, \cref{lem:delta_bounded_in_expect}, and \cref{lem:delta_0_whp}.
    \cref{lem:delta_small_when_beta_large} and \cref{lem:delta_0_whp} do not use properties specific to the beta distribution, so, hold in this setting.
    \cref{lem:delta_bounded_in_expect} uses some properties of the beta distribution,
    but, not until after Equation~\eqref{eq:used_later_in_generalization_2}.
    Continuing from Equation~\eqref{eq:used_later_in_generalization_2} and using Assumption (3) gives us the following analogue of \cref{lem:delta_bounded_in_expect}.\\
    \begin{mdframed}[style=FrameBox]
      For all iterations $t\in\N$, values $D^\et$, constant $\eps\in (0,\nfrac{1}{2})$, and parameters $a^{\sexp{1}}>1$, it holds that
      $$\Pr[\betat\negsp<\negsp 1-\eps]\cdot \ex{\deltat\mid (\betat\negsp < \negsp  1-\eps)\land (\ux \negsp \neq \negsp  0)} \leq \mathop{\expect}\limits_{\beta\sim \cD(a^{\sexp{1}})}\insquare{\frac1\beta-1}.$$
    \end{mdframed}
    This suffices to prove an upper bound on $\expect[\delta \mid D]$ following the same arguments.
    Then, following the steps similar to the rest of the proof, we can prove the following upper bound on $\expect[\at]$ (see Remark~\ref{rem:bound_when_unif}):
    \begin{align*}
      \expect[\at]\leq a^{\sexp{1}} + 2t \cdot\frac{16\ln{n}}{n(1-\rho)}\cdot {\mathop{\expect}\nolimits_{\beta\sim \cD(a^{\sexp{1}})}\insquare{\frac1\beta}}.
    \end{align*}
    Now, using the fact that $\Phi$ is increasing and concave, the result follows by the arguments used to prove Equation~\eqref{eq:upperbonud_used_later}.
  \end{proof}

  \begin{remark}
    Even if $\Phi$ is not concave, we can use Markov's inequality~\cite{motwani1995randomized} to lower bound the probability that $\at$ is large.
    This will give us:
    \begin{align*}
      \Pr\insquare{\at\geq a^{\sexp{1}}+ O\inparen{  \inparen{   \frac{t\ln{n}}{n(1-\rho)}   }^{\frac{1}{2}} }  } \leq O\inparen{\inparen{\frac{t\ln{n}} {n(1-\rho)}}^{\frac{1}{2}}}.
    \end{align*}
    This suffices to prove the following upper bound on $\expect\sinsquare{\betat}$:
    $$\expect\sinsquare{\betat}\leq \Phi\inparen{a^{\sexp{1}}+C\inparen{\frac{t\ln{n}}{n(1-\rho)}}^{\frac{1}{2}} }+  \inparen{\frac{t\ln{n}}{n(1-\rho)}}^{\frac{1}{2}},$$
    for some constant $C>0$.
  \end{remark}

  \subsection{Generalization to other update rules}\label{sec:proof:generalization_for_other_update_rules}
  In this section, we present the formal statements of our results for other update rules (also see, \cref{sec:generalizing_for_other_update_rules}).
  Recall that in \cref{sec:generalizing_for_other_update_rules} we consider general updated rules of the form
  \begin{align*}
    a^{\sexp{t+1}}\coloneqq \at \cdot F\inparen{\frac{U^\et}{\wt{U}^\et}},\yesnum\label{eq:update_rule_gen}
  \end{align*}
  where $F\colon [1,\infty)\to[1,\infty)$ is a continuous function, satisfying:
  \begin{enumerate}
    \item $F(1)=1$,
    \item $F$ is strictly increasing, and
    \item $F$ is concave.
  \end{enumerate}
  We prove the following versions of \cref{thm:rooney_rule} and \cref{thm:no_rooney_rule}:
  \begin{lemma}[\textbf{Fast learning with the Rooney Rule for other update rules}]\label{lem:gen_rule_1}
    Under the update rule~\eqref{eq:update_rule_gen}, if $F$ satisfies the above assumptions,
    for all $t\in \N$, $\eps\in (0,1)$, $n\in \N$, $k\in [n]$, $\ell\in [k]$, $\rho\in (0,1)$,
    $a^{\sexp{1}},b > 1$,
    and continuous and bounded distribution $\cP$,
    when the panel is constrained by the $\ell$-th order Rooney Rule,
    then there exists an iteration $t\in \N$, such that, $$\expect[\betat]\geq 1-\eps,$$
    where the expectation is over the draws of implicit bias $\beta^{\sexp{s}}$ and latent utilities of candidates in all previous iterations $s\in [t-1]$.
  \end{lemma}
  \begin{proof}
    At a high-level this proof follows the same structure as the proof of \cref{thm:rooney_rule}.
    Define the events $\cE^\et$ and $\cF^\et$ as:
    \begin{align*}
      \text{$\cE^\et \coloneqq \inbrace{ X_{(n_X : n_X)}^\et> Y_{(n_Y : n_Y)}^\et} \text{ and }\
      \cF^\et \coloneqq \inparen{\betat\leq {\rm median}(\betat)}.$}
    \end{align*}
    The first step in the proof of \cref{thm:rooney_rule} is to establish a lower bound on $a^{\sexp{t+1}}$ conditioned on $\cE^\et$ and $\cF^\et$.
    This lower bound does not generalize to the general update rule~\eqref{eq:update_rule_gen}.
    But, we can use the intermediate result, from Equation~\eqref{eq:intermediate_result} (in the proof of \cref{lem:rooney_rule_increases_a_fast}).
    Equation~\eqref{eq:intermediate_result} shows that conditioned on $(\cE^\et\land \cF^\et)$, it holds that $\deltat\geq \frac{(b-1)}{(k-\ell+1)(\at+b)}$.

    Applying this result for update rule~\eqref{eq:update_rule_gen}, we have
    \begin{align*}
      \Pr\insquare{a^{\sexp{t+1}}\geq \at\cdot F\inparen{1+\frac{(b-1)}{(k-\ell+1)(\at+b)}}}\geq \Pr\insquare{\cE^\et \land \cF^\et} \stackrel{}{=} \frac{\rho}{2}.
    \end{align*}
    If we can show that $\at\cdot F\inparen{1+\frac{(b-1)}{(k-\ell+1)(\at+b)}}$ is at least $\at+c$ for some {\em constant} $c>0$,
    then, the result would follow by deriving a concentration result like Equation~\eqref{eq:bound_on_at_proof}, and then, following the argument in Equation~\eqref{eq:lowerbound_on_prob}.

    Towards this observe that by the concavity of $F$, for all $y\geq x\geq 1$ and $\eta> 0$, it holds that
    \begin{align*}
      \frac{F\inparen{1+\frac{\eta}{y}} - F(1)}{\nfrac{\eta}{y}} \geq \frac{F\inparen{1+\frac{\eta}{x} }-F(1)}{\nfrac{\eta}{x}}.
    \end{align*}
    Let $\eta\coloneqq \frac{b-1}{k-\ell+1}$, $y\coloneqq \at+b$, and $x\coloneqq a^{\sexp{1}}+b$.
    Rearranging
    \begin{align*}
      \at\cdot F\inparen{1+\frac{(b-1)}{(k-\ell+1)(\at+b)}} &= \at\cdot F\inparen{1+\frac{\eta}{\at+b}}\\
      & \geq \at \cdot \inparen{\ 1+ \inparen{F\inparen{1+\frac{\eta}{a^{\sexp{1}}+b} }-1} \cdot \frac{\nfrac{\eta}{(\at+b)}}{\nfrac{\eta}{(a^{\sexp{1}}+b)}}\ }\\
      &= \at + \at \cdot \inparen{F\inparen{1+\frac{\eta}{a^{\sexp{1}}+b}}-1}  \cdot \frac{a^{\sexp{1}}+b}{\at+b}\\
      &\geq \at + a^{\sexp{1}} \cdot \inparen{F\inparen{1+\frac{\eta}{a^{\sexp{1}}+b}}-1} \tag{Using $\at\geq a^{\sexp{1}}$}\\
      &= \at + c.
    \end{align*}
    Where $c\coloneqq a^{\sexp{1}} \cdot \inparen{F\inparen{1+\frac{\eta}{a^{\sexp{1}}+b}}-1} >0$.
    By our earlier discussion, this proves \cref{lem:gen_rule_1}.
  \end{proof}

  \begin{lemma}[\textbf{Slow learning without the Rooney Rule for other update rules}]
    Under the update rule~\eqref{eq:update_rule_gen}, if $F$ satisfies the above assumptions,
    for all $t\in \N$, $\eps\in (0,1)$, $k\in[n]$, and $\rho\in (0,1)$,
    and $a^{\sexp{1}},b > 1$,
    when the panel is not constrained by the Rooney Rule (i.e., $\ell=0$) and $\cP$ is ${\rm Unif}(0,1)$,
    then there exists an $n_0\in \N$, such that, for all $n\geq n_0$ $$\expect[\betat]\leq \expect\sinsquare{\beta^{\sexp{1}}}+\eps,$$
    where the expectation is over the draws of implicit bias $\beta^{\sexp{s}}$ and latent utilities of candidates in all previous iterations $s\in [t-1]$.
  \end{lemma}
  \begin{proof}
    At a high-level this proof follows the same structure as the proof of \cref{thm:no_rooney_rule}.
    Let $D^\et$ denote the values of $\beta^\es$ for all $s\in [t-1]$:
    $$D^\et\coloneqq \inbrace{\delta^s}_{0\leq s<t}.$$
    First, for all values of $D^\et$, we upper bound $\exgiv{\deltat}{D^\et}$ by $O\inparen{\frac{\ln{n}}{n}}$.
    We can do so by following the steps used to prove Equation~\eqref{eq:upperbound_on_expec_deltat} (Notice, that arguments there do not depend on the update rule.)

    Now, if we can prove an upper bound on $\expect\sinsquare{\at}$ which reduces as $n$ increases, then the result follows using similar arguments to the ones used to derive Equation~\eqref{eq:upperbonud_used_later}.
    Towards this consider
    \begin{align*}
      \expect\sinsquare{\at} &= a^{\sexp{1}}\cdot \expect\insquare{F(1+\delta^{\sexp{1}})\cdot F(1+\delta^{\sexp{2}}) \cdots \cdot F(1+\delta^{\sexp{t-1}})}. \yesnum\label{eq:expect_exp_of_a}
    \end{align*}
    To show an upper bound on $\ex{\at}$, we need to upper bound $\expect\insquare{F(1+\delta^\es)\mid D^\et}$ for all $s\in [t-1]$.
    We can do so as follows
    \begin{align*}
      \exgiv{F(1+\deltat)}{ D^\et} \ &\stackrel{(F \text{ is concave})}{\leq} \  F\inparen{\expect\insquare{1+\deltat\mid D^\et}}\\
      &\hspace{-1mm} \stackrel{(F \text{ is increasing})}{\leq} \  F\inparen{1+O\inparen{\frac{\ln{n}}{n}}}.
    \end{align*}
    Combining this with Equation~\eqref{eq:expect_exp_of_a} (and following similar steps to prove Equation~\eqref{eq:expect_product_delta_sm}), we get
    \begin{align*}
      \expect\sinsquare{\at} &\leq a^{\sexp{1}}\cdot \inparen{ F\inparen{1+O\inparen{\frac{\ln{n}}{n}}}}^t.
    \end{align*}
    Since $F(1)=1$ and $F$ is continuous, the theorem follows by choosing a large enough $n_0$.
  \end{proof}

  \subsection{Asymptotic learning without the Rooney Rule }\label{sec:proof:convergestoone}
  In this section, we prove the \cref{thm:convergestoone}. We first give a sketch of the proof and then present the complete proof.
  \begin{lemma}[\textbf{Asymptotic learning without the Rooney Rule}]\label{thm:convergestoone}
    Under the dynamics of implicit bias~\eqref{alg:iterative_selection_and_panel_learning}, for all population size $n\in \N$, number of candidates selected $k\in \N$, Rooney Rule parameter $1\leq \ell\leq k$, ratio of the underrepresented candidates $\rho\in (0,1)$,
    parameters $a^{\sexp{1}},b > 1$,
    and bounded distribution of latent utility $\cP$,
    when the panel is not constrained by the order Rooney Rule (i.e., $\ell=0$), then as $t\to\infty$
    \begin{align} \label{eq:convergestoone}
      \expect\big[\beta^\et\big]\to 1.
    \end{align}
  \end{lemma}
  \begin{proof}[Proof sketch]
    First, we reduce the theorem to proving the following statement:
    for all $M\geq 0$, $\Pr[\at\leq M]\to 0$ as $t\to \infty$.
    This reduction uses the fact that when $\at$ is large, then $\betat$ is concentrated around 1.
    More formally, we use the fact that:
    \begin{align}
      \lim_{a\to \infty}\expect_{\beta\sim\Beta(a,b)}\insquare{\beta}=1.\label{eq:prop_ofprop_of_beta_dist_for_convergencetoone}
    \end{align}
    (The concentration then follows, say from the Markov's inequality~\cite{motwani1995randomized}.)

    Next, we observe that we always have $\nfrac{U^\et}{\wt{U}^\et}\geq 1$.
    To see this, note that $(U^\et-\wt{U}^\et)$ is $(1-\betat)\cdot\sum_{i\in S_\ell^\et\cap G_X}\xit{}$, and since $1-\betat\geq 0$ and $\xit{}\geq 0$, $(U^\et-\wt{U}^\et)\geq 0$.
    However, as $\at\to \infty$, the random variable $\nfrac{U^\et}{\wt{U}^\et}$ concentrates around 0.
    Conditioning on the fact that $\at\leq M$, we can show that there are positive constants $c_0$ and $c_1$ which are functions of $M$ such that
    $$\probgiv{\frac{U^\et}{\wt{U}^\et}>1+c_0}{ \at \leq M}\geq c_1.$$
    This uses the fact that $\cP$  is a non-point distribution and that the closure of the support of the beta distribution contains 1.

    Consider a particle performing a random walk starting at 0 and with an absorbing barrier at $O\inparen{\log_{(1+c_0)}\inparen{\frac{M}{a^{\sexp{1}}}}}$.
    The particle moves one unit to the right with probability $c_1$, and stays at the its position with probability $(1-c_1)$.
    Using our previous result, we can prove that
    $$\Pr[a^\et\leq M]\leq \Pr\big[\ \text{particle is absorbed at after time $t$}\ \big].$$
    Using standard analysis of random walks we get that $\Pr[\at\leq M]\to 0$ as $t\to \infty$.
  \end{proof}

  Now we give the full proof of \cref{thm:convergestoone}.
  \subsubsection{Proof of Lemma~\ref{thm:convergestoone}}\label{sec:proof:thm:convergestoone}
  Fix any $\eps>0$.
  We will prove that there exists a $\tau\in \N$, such that,
  $$\expect[\beta^{\sexp{\tau}}]\geq 1-\eps.$$
  Letting $\eps\to 0$ proves the theorem.\\

  \noindent {\em Notation.}
  Let ${\rm supp}(\cdot)$ denote the support of a distribution and ${\rm int}(\cdot)$ denote the interior of an interval.
  Assume that $\cP$ have support ${\rm supp}(\cP)=[L,U]$ where $L,U\geq 0$.
  Choose
  \begin{align}
    m\coloneqq\frac{U+L}{2}\quad \text{and}\quad \epsilon\coloneqq \frac{U-L}{2(U+L)}.
  \end{align}
  It holds that $m(1-\epsilon)=L+\frac{U-L}{4}$.

  We first state some propositions which we use in the proof of \cref{thm:convergestoone}. Their proofs appear ini Section~\ref{sec:proof:supporting_lemmas}
  \begin{proposition}[\textbf{Lower bound on \ensuremath{\Pr[\cE^\et]}}]\label{lem:lower_bound_on_E}
    For all iterations $t\in \N$,
    define $\cE^\et$ as the following event
    \begin{align}
      \cE^\et\coloneqq \inparen{X_{(n_X:n_X)}^\et > m}\ \ \land \ \
      \inparen{Y^\et_{(n_Y-k:n_Y)} < m\cdot (1-\epsilon)}.\label{eq:lower_bound_on_prob_1}
    \end{align}
    There is a constant $c>0$, such that $\Pr\insquare{\cE^\et}>c$.
  \end{proposition}
  \begin{proposition}[\textbf{Lower bound on \ensuremath{\probgiv{\mathcal{H}}{ \at\leq M}}}]\label{lem:lower_bound_on_H}
    Given $0<\eta<\epsilon$ and $M>a^{\sexp{1}}$, for all iterations $t\in \N$,
    define $\cH^\et$ as the following event
    $$\cH^\et\coloneqq \inparen{1-\epsilon\leq \betat\leq 1-\eta}.$$
    There exists a constant $c>0$ such that $\Pr[\cH^\et\mid \at\leq M]=c$.
  \end{proposition}
  \begin{proposition}[\textbf{Conditional lower bound on $\Pr[\deltat>c_0]$}]\label{lem:delta_increases_with_const_prob}
    Given an $M>0$, there exist constants $c_0,c_1>0$ such that at any iteration $t\in \N$
    $$\Pr\insquare{\deltat> c_0\mid \at\leq M}=c_1.$$
  \end{proposition}
  Now we are ready to present the proof
  \begin{proof}[Proof of \cref{thm:convergestoone}]
    Fix any $\eta>0$.
    We will prove that there exists a $\tau\in \N$, such that,
    $$\expect[\beta^{\sexp{\tau}}]\geq 1-\eta.$$
    Letting $\eta\to 0$ proves the theorem.
    Towards this, we first prove a lower upper about on the probability that $a^{\sexp{\tau}}$ is small.
    Let $M>0$ be a large constant and $\eps>0$ be a small constant, we will specify these (in terms of $\eta$) later.
    For all $i\in [t]$, let $Z_i$ be the indicator random variable that $\delta^{i}> c_0$.
    From Fact~\ref{fact:a_is_monotonic_in_t} we know that if $\at \leq M$, then for all $i\in [t]$, $a^{(i)} \leq M$.
    Using \cref{lem:lower_bound_on_E} for each $i\in [t]$ we get
    \begin{align}
      \text{For all\ } i\in [t], \quad \Pr[Z_i]\geq c_1\quad \text{and}\quad \expect[Z_i]\geq c_1.\label{eq:stats_of_zi}
    \end{align}
    We also have
    \begin{align*}
      \at&= a^{\sexp{1}} \cdot \prod_{i=1}^{t-1}(1+\delta^{(i)})
      \geq a^{\sexp{1}} \cdot \prod_{i=1}^{t-1}(1+c_0)^{Z_i}
      = a^{\sexp{1}} \cdot (1+c_0)^{\sum_{i=0}^{t-1}Z_i}.
    \end{align*}
    Thus, if $\at\leq M$, then
    \begin{align*}
      \sum_{i=0}^{t-1}Z_i&\leq \frac{\log\inparen{  \frac{M}{a^{\sexp{1}}}  }}{\log(1+c_0)}.
    \end{align*}
    Choose
    \begin{align}
      \tau= \frac{\log\inparen{  \frac{M}{a^{\sexp{1}}}  }}{\epsilon c_1\log(1+c_0)}.\label{eq:choice_of_tau}
    \end{align}
    Recall we are interested in $t\geq \tau$.
    Now, we can prove the required inequality
    \begin{align*}
      \Pr\insquare{\at\leq M} &\leq \Pr\insquare{\sum_{i=0}^{t-1}Z_i \leq \frac{\log\inparen{  \frac{M}{a^{\sexp{1}}}  }}{\log(1+c_0)}}\\
      &\leq \frac{1}{\expect\insquare{\sum_{i=0}^{t-1}Z_i}}\frac{\log\inparen{  \frac{M}{a^{\sexp{1}}}  }}{\log(1+c_0)}\\ %
      &\stackrel{\eqref{eq:stats_of_zi}}{\leq} \frac{1}{tc_1}\frac{\log\inparen{  \frac{M}{a^{\sexp{1}}}  }}{\log(1+c_0)}\\
      &\stackrel{\eqref{eq:choice_of_tau}}{\leq} \epsilon.\tagnum{$t\geq \tau$} \customlabel{lower_bound_on_delta}{\theequation}
    \end{align*}
    Set $M=\frac{2b}{\eta}$ and $\eps=\frac{\eta}{2}$, then
    \begin{align*}
      \expect[\beta^{\sexp{\tau}}] = \frac{a^{\sexp{\tau}}}{a^{\sexp{\tau}}+b} \geq (1-\eps)\frac{M}{M+b}\geq 1-\eta.
    \end{align*}
    Where in the first in equality we use that $\frac{a^{\sexp{\tau}}}{a^{\sexp{\tau}}+b}$ is an increasing function of $a^{\sexp{\tau}}$.
  \end{proof}

  \subsubsection{Proofs of supporting propositions}\label{sec:proof:supporting_lemmas}
  The proofs in this section use the following fact.
  \begin{fact}\label{fact:lowerbound_on_pdf}
    Given a distribution $\cD$ with a closed interval support and continuous probability density function, for any random variable $X\sim \cD$ and $m\in {\rm int}({\rm supp}(\cD))$, we have
    $$\Pr\insquare{X< m}\in (0,1).$$
  \end{fact}
  \begin{proof}
    From our assumptions we know that the pdf is strictly positive on ${\rm int}({\rm supp}(\cD))$, thus cdf is strictly increasing on ${\rm int}({\rm supp}(\cD))$.
    Also, as ${\rm int}({\rm supp}(\cD))$ is an open interval, we have a $\epsilon>0$, such that, $(m\pm\epsilon)\in {\rm int}({\rm supp}(\cD))$.
    This gives us the required bounds
    $$0\leq \Pr\insquare{X< m+\epsilon}<\Pr\insquare{X< m}<\Pr\insquare{X< m+\epsilon}\leq 1.$$
  \end{proof}
  \noindent We now present the proof of \cref{lem:lower_bound_on_E} which is used in the proof of \cref{thm:convergestoone}.
  \begin{proof}[Proof of \cref{lem:lower_bound_on_E}]
    We note that $X_{(n_X:n_X)}^\et$ and $Y^\et_{(n_Y-k:n_Y)}$ are independent.
    Thus, $\Pr\sinsquare{\cE^\et}$ can be written as a product of two probabilities.
    The lemma follows by lower bounding each probability in this product by positive constants from Fact~\ref{fact:lowerbound_on_pdf}.
    We note that we know $m, m\cdot (1-\epsilon)\in (L,U)$.
    Therefore, we can apply Fact~\ref{fact:lowerbound_on_pdf}.
    This gives us $c_1,c_2\in (0,1)$ such that
    \begin{align}
      F(m)&=c_1\quad \text{and} \quad F(m\cdot (1-\epsilon))=c_2\label{eq:using_fact_on_cs}.
    \end{align}
    The lemma follows since
    \begin{align*}
      \Pr\insquare{\cE^\et} &\negsp =
      \Pr\insquare{X_{(n_X:n_X)}^\et > m}\cdot
      \Pr\insquare{Y^\et_{(n_Y-k:n_Y)} < m\cdot (1-\epsilon)}\tag{$X_{(n_X:n_X)}^\et$ and $Y^\et_{(n_Y-k:n_Y)}$ are independent}\\
      &\geq \prod_{i\in [n_X]}\Pr\insquare{X_{i}^\et > m} \cdot \prod_{i\in [n_Y]}\Pr\insquare{Y_{i}^\et <m\cdot (1-\epsilon)}\\
      &=(1-F(m))^{n_X}\cdot (F(m\cdot (1-\epsilon)))^{n_Y}\\
      &\stackrel{\eqref{eq:using_fact_on_cs}}{\geq} (1-c_1)^{n_X}\cdot c_2^{n_Y}.
    \end{align*}
    The lemma follows by choosing $c=(1-c_1)^{n_X}\cdot c_2^{n_Y}>0$.
  \end{proof}

  \begin{proof}[Proof of \cref{lem:lower_bound_on_H}]
    \begin{align*}
      \probgiv{\cH^\et}{ \at\leq M}&=\frac{1}{B(\at, b)}\int_{1-\epsilon}^{1-\eta} x^{\at-1}(1-x)^{b-1}dx.\\
      &\geq \frac{(1-\epsilon)^{\at-1}}{B(\at, b)}\int_{1-\epsilon}^{1-\eta}(1-x)^{b-1}dx\\
      &\geq \frac{(1-\epsilon)^{\at-1}}{B(\at, b)} \frac{\epsilon^b-\eta^b}{b}\\
      &\geq \frac{c_1\cdot\Gamma(\at+b)}{\Gamma(\at)\Gamma(b)}\tag{Let $c_1\coloneqq  (\epsilon^b-\eta^b)\cdot\frac{(1-\epsilon)^{\at-1}}{b}>0$}\\
      &\geq \frac{c_1}{2\cdot\Gamma(\at)\Gamma(b)}\tag{Using {$\inf_{x>0}\Gamma(x)\geq \frac12.$}}\\
      &\geq \frac{c_1}{\big(\Gamma(M)+\Gamma(a^{\sexp{1}})\big)\cdot\Gamma(b)}\tag{Using {$\Gamma(x)$ is convex on $(0,\infty)$} and $a^{\sexp{1}}\leq \at\leq M$}\\
      &\geq c.\yesnum\label{eq:lower_bound_on_H_without_E}
    \end{align*}
    The lemma follows by noting that $c>0$.
  \end{proof}

  \begin{proof}[Proof of \cref{lem:delta_increases_with_const_prob}]
    We begin by giving lower bound on $\deltat$ conditioned on the event $\cE^\et\land \cH^\et\land (\at\leq M)$
    Notice that conditioned on $\cE^\et\land \cH^\et\land (\at\leq M)$, $X_{(n_X:n_X)}^\et\in S^\et$ and $X_{(n_X:n_X)}^\et>m$.
    This gives us
    \begin{align*}
      \deltat&=\frac{(1-\betat) \cdot \ux}{\betat\cdot\ux +
      \uy}\\
      &\geq \frac{(1-\betat) \cdot \ux}{kU}\\
      &\geq \frac{(1-\betat)\cdot X_{(n_X:n_X)}^\et}{kU}\tag{$X_{(n_X:n_X)}^\et\in S^\et$ conditioned on $\cE^\et$}\\
      &> \frac{(1-\betat)\cdot m }{kU}\\
      &\geq \frac{\eta m }{kU}.\yesnum\label{eq:lower_bound_on_delta_cond_on_E}
    \end{align*}
    In other words, we have
    \begin{align}
      \probgiv{\deltat> \frac{\eta m }{kU}}{ (\at\leq M), \cH^\et, \cE^\et}=1.\label{eq:delta_is_pos_const_cond_on_H_E_M}
    \end{align}
    Before proceeding we note that $\betat$, $X_{(n_X:n_X)}^\et$, and $Y^\et_{(n_Y-k:n_Y)}$ are independent, and remain independent with conditioning on $\at$.
    Thus, $\cE^\et$ and $\cH^\et$ are independent conditioned on $\at=x$.
    Further,
    \begin{align*}
      \probgiv{\cH^\et \land  \cE^\et}{ \at\leq M} &= \int_{a^{\sexp{1}}}^{M}\probgiv{\cH^\et \land  \cE^\et}{ \at=x}\cdot \probgiv{\at=x }{ \at\leq M}dx\\
      &=\int_{a^{\sexp{1}}}^{M}\probgiv{\cH^\et}{ \at=x}
      \cdot \probgiv{\cE^\et}{ \at=x}
      \cdot
      \probgiv{\at=x}{ \at\leq M}dx\tag{$\cE^\et$ and $\cH^\et$ are independent conditioned on $\at=x$.}\\
      &=\Pr\insquare{\cE^\et}\int_{a^{\sexp{1}}}^{M}\probgiv{\cH^\et}{ \at=x}\cdot \probgiv{\at=x }{ \at\leq M}dx\tag{$\cE^\et$ is independent of $\at$}\\
      &=\Pr\insquare{\cE^\et}\cdot\probgiv{\cH^\et}{\at\leq M}.\yesnum\label{eq:E_and_H_are_independent}
    \end{align*}
    Thus, $\cE^\et$ is independent of $\cH^\et$ conditioned on $(\at\leq M)$.
    Now we have all the tools to finish the proof
    \begin{align*}
      \probgiv{\deltat> \frac{\eta m }{kU}}{ \at\leq M}& \geq \probgiv{\deltat\geq \frac{\eta m }{kU}}{ \at\leq M, \cH^\et, \cE^\et} \cdot\probgiv{\cH^\et \land \cE^\et}{\at\leq M}\\
      &\stackrel{\eqref{eq:delta_is_pos_const_cond_on_H_E_M}}{\geq} \probgiv{\cH^\et\land \cE^\et}{\at\leq M}\\
      &\stackrel{\eqref{eq:E_and_H_are_independent}}{\geq} \probgiv{\cH^\et}{\at\leq M}\cdot\Pr\insquare{\cE^\et}\\
      &\stackrel{}{\geq} c_2\cdot c_3.\tag{Using \cref{lem:lower_bound_on_E,lem:lower_bound_on_H}}
    \end{align*}
    The lemma follows by choosing $c_0=\frac{\eta m}{kU}>0$ and $c_1=c_2\cdot c_3>0$.
  \end{proof}

  \section{Conclusion}\label{sec:discussion_conclusion}

  We consider a theoretical model of how a hiring panel's implicit bias changes over repeated candidate selection processes and study the effect of the Rooney Rule on the panel's implicit bias.
  We show that, under our model, if the panel uses the Rooney rule, then the rate at which its implicit bias reduces decreases with the size of the shortlist but is independent of the  number of candidates.
  However, when the panel does not use the Rooney Rule, the same rate decreases with the  number of candidates (Section~\ref{sec:theoretical_results}).
  Thus, in the regime where the total number of candidates is much larger than the size of the shortlist, our results predict a significantly faster reduction in the panel's implicit bias by using the Rooney Rule---giving another reason to use it to mitigate implicit bias in the long term.

  Toward understanding the robustness of this result, we consider some extensions of the model and show how our results generalize to them (Section~\ref{sec:technical_remarks}).
  Here, in particular, we consider noise in the implicit bias of the panel, which can change the distribution of the implicit bias, and
  identify properties of the resulting distribution where our results hold (see Section~\ref{sec:generalization_for_other_distributions}).
  An interesting direction for future work could be to examine the effects of noise in the observed utilities---as in our empirical study (Section~\ref{sec:empirical_results}).

  Our empirical findings based on the experiment on Amazon Mechanical Turk
  show that, over multiple rounds, the Rooney Rule helped the participants learn to select a shortlist which is more representative of the qualified candidates in the applicant pool,
  without substantially decreasing the latent utility of selected candidates.

  Our findings provide evidence that the Rooney Rule can reduce the biases of the panel in repeated candidate selection processes. The policy can be a simple and effective tool as part of a larger effort to mitigate implicit bias.

  \section*{Acknowledgements}

  This research was supported in part by a J.P. Morgan Faculty Award and an AWS MLRA grant.

  \newpage
  \bibliographystyle{plain}
  \bibliography{bib}

  \newpage
  \appendix

  \section{Extended empirical observations} \label{sec:exted-empirical-observations}

  In Section \ref{sec:demo}, we discuss the instructions, demonstration and user interface shown to participants in the experiment in detail. (Screen captures of the experiment are shown in Section \ref{sec:experiment_figures}.)
  In Section \ref{sec:num_minority_comparison_app}, we provide more detailed analysis of the mean number of blue tiles selected for \rrgrp{} versus \unconsgrp{}, including a Welch's $t$-test and a linear mixed-effects model for the mean number of blue tiles selected for \rrgrp{} and \unconsgrp{}.
  In Section \ref{sec:latent_utility_details}, we provide more detailed analysis of the latent utility attained by participants from \rrgrp{} versus \unconsgrp{}, including a linear regression using iteration $t$ as a regressor and a $t$-test showing that the latent utility derived from blue tiles is higher for \rrgrp{} than for \unconsgrp{}.
  Finally, in Section \ref{sec:optimal_strategy_overlap_details}, we show that the trends observed when we analyze optimal strategy overlap are similar to those when we analyze latent utility.

  \subsection{Instructions, demonstration and user interface} \label{sec:demo}

  Before participants started the experiment, they were asked to sign a consent form and then provided with instructions and a brief demonstration.
  Thirty-nine participants were in \unconsgrp{} and 37 participants were in \rrgrp{}. Participant groups were uneven because some participants exited the experiment before completing all 25 iterations, and were therefore eliminated from the analysis.

  In the consent form, participants were told that the purpose of the experiment was to study ``how people learn to play a game over many repeated iterations,'' without specifying more detail about the candidate selection scenario that we study or that there was bias against one group.
  They were also informed of their baseline compensation (\$1.00), their estimated average bonus (\$1.20) which was proportional to their score, and the approximate time to complete the task (8 minutes).
  The mean bonus paid was \$1.37, and the standard deviation was \$0.14.

  If they provided consent, participants were shown an instructions page (see Figure \ref{fig:complete_instructions}) describing that they were to select 10 elements from the list of tiles, that each tile had an \textit{estimated} and \textit{true} value (corresponding to observed and latent utilities), and they they should try to maximize the true values of the tiles they selected.
  They were informed that the sum of the true utilities of their selections corresponded to a score, for which 15,000 points corresponds to \$1.00.
  Then they were directed to a demonstration of the experiment.

  In the demonstration, participants were shown a short list of red tiles (see Figure \ref{fig:complete_instructions_demo1}). Unlike in the real experiment where participants could only see the observed utilities of tiles before submitting their selections, in the demonstration, participants could select a tile to see both the observed and latent utility of that tile. At the bottom of the screen, they could see the sum of the observed and latent utilities of the set of currently selected tiles. To complete the demonstration, participants needed to select the three tiles with the highest observed utilities, by exploring the latent utilities of the tiles until they found the top three.

  The observed utilities for the tiles in the demonstration were small distortions of the latent utilities, such that the ordering of observed utilities was nearly (but not exactly) the same as the ordering of the latent utilities.
  The purpose of the demonstration was to show participants how the task worked and to lead them to believe that the observed utility of a tile is a good estimate of its latent utility; this aligns with real-world candidate selection contexts where the selection panel assumes that its assessments of candidates are unbiased estimates of their latent utility.
  We used red tiles (representing the overrepresented group) because we assume that the selection panel already has built experience closely evaluating overrepresented group candidates, while they may have little to no experience closely evaluating underrepresented group candidates.

  The observed and latent utilities in the demonstration were held constant across all participants. See Figure \ref{fig:complete_instructions_demo2} for the observed and latent utilities of the three tiles with the highest latent utilities. Once the tiles with the top three latent utilities were selected, participants could click ``Done'' and proceed to the experiment.

  Figures \ref{fig:mturk_demo1}, \ref{fig:mturk_demo2}, and \ref{fig:mturk_demo3} show the user interface for an example selection process.
  At the start of an iteration (see Figure \ref{fig:mturk_demo1}), participants were presented with a scrollable box containing the
  list of 100 tiles, sorted from high to low in order of observed utilities. Below the list of tiles, there was a bar chart showing the number of points scored in each previous iteration so far.

  When $k=10$ tiles were selected (see Figure \ref{fig:mturk_demo2}), a forward arrow button at the bottom of the list of tiles turned green, and participants could submit their choices. Participants were free to exit the study at any time by selecting the ``exit study'' button and still receive compensation, but any participants who exited the study early were removed from the analysis.

  After participants submitted their selections, they were shown both the observed utilities (in strikethrough text) and the latent utilities of selected tiles (see Figure \ref{fig:mturk_demo3}). In gray boxes at the bottom of the page, participants can see the total estimated and actual points earned during the round, as well as their overall cumulative points and bonus payment. Then participants can select the ``Next round'' to proceed to the next iteration.

  \subsection{Number of blue tiles selected: comparison of participant-group means and linear mixed-effects model} \label{sec:num_minority_comparison_app}

  In Section \ref{sm:num_blue_selected}, we find that there is a significant difference in the mean number of blue tiles selected for \rrgrp{} versus \unconsgrp{}, using a $t$-test analogous to Figure \ref{tab:num_selected_minority_over_required}.
  In Section \ref{sm:mixed_effects_model}, we analyze a linear mixed-effects model as a way to account for correlations between selection decisions completed by a participant. We find similar results to Welch's $t$-test.

  \subsubsection{Late-stage mean for \rrgrp{} and \unconsgrp{}} \label{sm:num_blue_selected}

  We use Welch's independent samples $t$-test to compare the mean number of blue tiles selected between \rrgrp{} and \unconsgrp{} during the last fifteen iterations.
  Our null hypothesis is that there is no difference in the means.
  We reject the null hypothesis at significance level $p < 0.001$.
  We find that the \rrgrp\ selected significantly more blue tiles on average during the last fifteen iterations, by an estimated margin of 1.5.
  See Figure \ref{tab:num_selected_minority} for the means and standard deviations for \rrgrp{} and \unconsgrp{}, as well as the results from Welch's $t$-test.

  \begin{figure}[h]
    \caption{The number of blue tiles selected is different for \rrgrp{} and \unconsgrp{} with significance $p < 0.001$.}\vspace{2mm}
    \centering
    \begin{tabular}{ r c c }
      & {Mean} & {Standard deviation} \\ \hline
      \rrgrp{} & 4.0 & 2.4 \\
      \unconsgrp{} & 2.5 & 3.1 \\
    \end{tabular}\vspace{2mm}

    \begin{tabular}{ r c c }
      & {Welch's $t$-test} \\ \hline
      Null hypothesis & population means equal \\
      Alternative hypothesis & population means different \\
      $t$-statistic & 9.5 \\
      Degrees of freedom & 1151 \\
      $p$-value & $<$0.001 \\
    \end{tabular}

    \label{tab:num_selected_minority}
  \end{figure}

  \subsubsection{Linear mixed-effects model} \label{sm:mixed_effects_model}

  Welch's $t$-test assumes independence between samples, but each participant completed multiple selection decisions. This could violate the independence assumption, if the selection decisions of a given participant are correlated.
  To account for within- and between-participant variation, we can analyze a linear mixed effects model.
  Linear mixed effects models capture hierarchy in data when, for example, multiple data points are collected from individuals in a sample.
  We treat $\ell$ as a fixed effect and participant IDs as a random effect.
  We take data from selection decisions during the last fifteen iterations.
  We reject the hypothesis that $\ell$ has no effect on the number of blue tiles selected with $p < 0.01$.
  See Figure \ref{tab:mixed_effects} for more detail.

  \begin{figure}[h]
    \caption{Using a mixed effects model, we find that $\ell$ has a significant effect on the number of blue tiles selected during the last fifteen iteration with $p < 0.01$.}\vspace{2mm}
    \centering
    \begin{tabular}{ r c }
      & {Mixed-effects model} \\ \hline
      \rrgrp{} estimate \textit{(standard error)} & 4.0 \textit{(0.55)} \\
      \unconsgrp{} estimate \textit{(standard error)} & 2.5 \textit{(0.38)} \\
      Null hypothesis & population means equal \\
      Alternative hypothesis & population means different \\
      $t$-statistic & 2.8 \\
      Degrees of freedom & 74  \\
      $p$-value & $<$0.01\\
    \end{tabular}

    \label{tab:mixed_effects}
  \end{figure}

  In a traditional comparison of means, the number of blue tiles selected is modeled as a function of membership in \rrgrp{} or \unconsgrp{} plus noise. In a linear mixed-effects model, the number of blue tiles selected is modeled as a linear function of membership in \rrgrp{} or \unconsgrp{} plus which participant made the selection plus noise. We do not see a significant effect if we use the linear mixed-effects model to compare the mean number of minority candidates selected \textit{in addition to} the required number.

  \subsection{Total latent utility: rate of increase over iterations and comparison of participant-group means} \label{sec:latent_utility_details}

  In Section \ref{sec:latent_util_regression}, we analyze a linear regression model of the mean latent utility attained by \rrgrp{} and \unconsgrp{}. We find that latent utility modestly increases over time for both groups.
  In Section \ref{sec:latent_util_from_blue_details}, we examine the mean latent utility derived from blue tiles for \rrgrp{} and \unconsgrp{}. We find that \rrgrp{} attains higher latent utility derived from blue tiles at significance level $p<0.001$.

  \subsubsection{Linear regression for latent utility over iterations} \label{sec:latent_util_regression}

  We fit a linear regression with outcome variable latent utility and predictors iteration $t$ for \rrgrp{} and \unconsgrp{} separately. The null hypothesis that the regression coefficients are zero, indicating no relationship between latent utility and intercept or iteration $t$. The regression coefficients for the iteration $t$ for \rrgrp{} and \unconsgrp{} are each significant at $p<0.05$. The regressor for iteration $t$ in \rrgrp{} is significant at $p<0.001$. See Figure \ref{tab:latent_utility_regression} for details.

  \begin{figure}[h!]
    \caption{Participants learn to attain greater latent utility over iterations, with a statistically significant effect for \rrgrp{} and \unconsgrp{}.}\vspace{2mm}
    \centering

    \begin{tabular}{ r c c }
      \hline
      \textbf{\rrgrp{}} & Intercept & Iteration $t$ \\ \hline
      Estimated coefficient & 0.77 & 0.0031 \\
      Standard error & 0.008 & 0.001  \\
      $t$-statistic & 91 & 5.2 \\
      $p$-value & $<$0.001 & $<$0.001  \\
      \hline
      \textbf{\unconsgrp{}{}} & Intercept & Iteration $t$ \\ \hline
      Estimated coefficient & 0.82 & 0.0010 \\
      Standard error & 0.007 & 0.0  \\
      $t$-statistic & 120 & 2.0 \\
      $p$-value & $<$0.001 & $<$0.05  \\
    \end{tabular}

    \label{tab:latent_utility_regression}
  \end{figure}

  \newpage

  \subsubsection{Latent utility derived from blue tiles is significantly higher for \rrgrp{} than for \unconsgrp{}} \label{sec:latent_util_from_blue_details}

  To further understand the latent utility of selection decisions from participant groups, we can show the proportion of total utility derived from the blue and red tiles separately.

  \begin{figure}[h]
    \caption{The latent utility derived from blue tiles is higher for \rrgrp{} than \unconsgrp{} with significance $p < 0.001$.}\vspace{2mm}

    \centering
    \begin{tabular}{ r c c }
      & {Mean} & {Standard deviation} \\ \hline
      \rrgrp{} & 0.38 & 0.24 \\
      \unconsgrp{} & 0.23 & 0.30 \\
    \end{tabular}\vspace{2mm}

    \begin{tabular}{ r c c }
      & \begin{tabular}{@{}c@{}}\textbf{Welch's $t$-test:} comparison of \\ latent utility means \\ derived from blue\end{tabular}\\ \hline
      Null hypothesis & population means equal \\
      Alternative hypothesis & population means different \\
      $t$-statistic & 9.5 \\
      Degrees of freedom & 1100 \\
      $p$-value & $<$0.001 \\
    \end{tabular}

    \label{tab:latent_util_ttest}
  \end{figure}

  \noindent In Figure \ref{fig:latent_util_by_group_proportion2}, we plot the and latent utility derived from blue and red tiles as a proportion of the optimal strategy utility.
  Notice that the proportions between blue and red tiles are closer to equal for \rrgrp{} than for \unconsgrp{}.
  \begin{figure}[h]
    \setlength{\belowcaptionskip}{-2pt}
    \centering
    \includegraphics[width=\columnwidth/2]{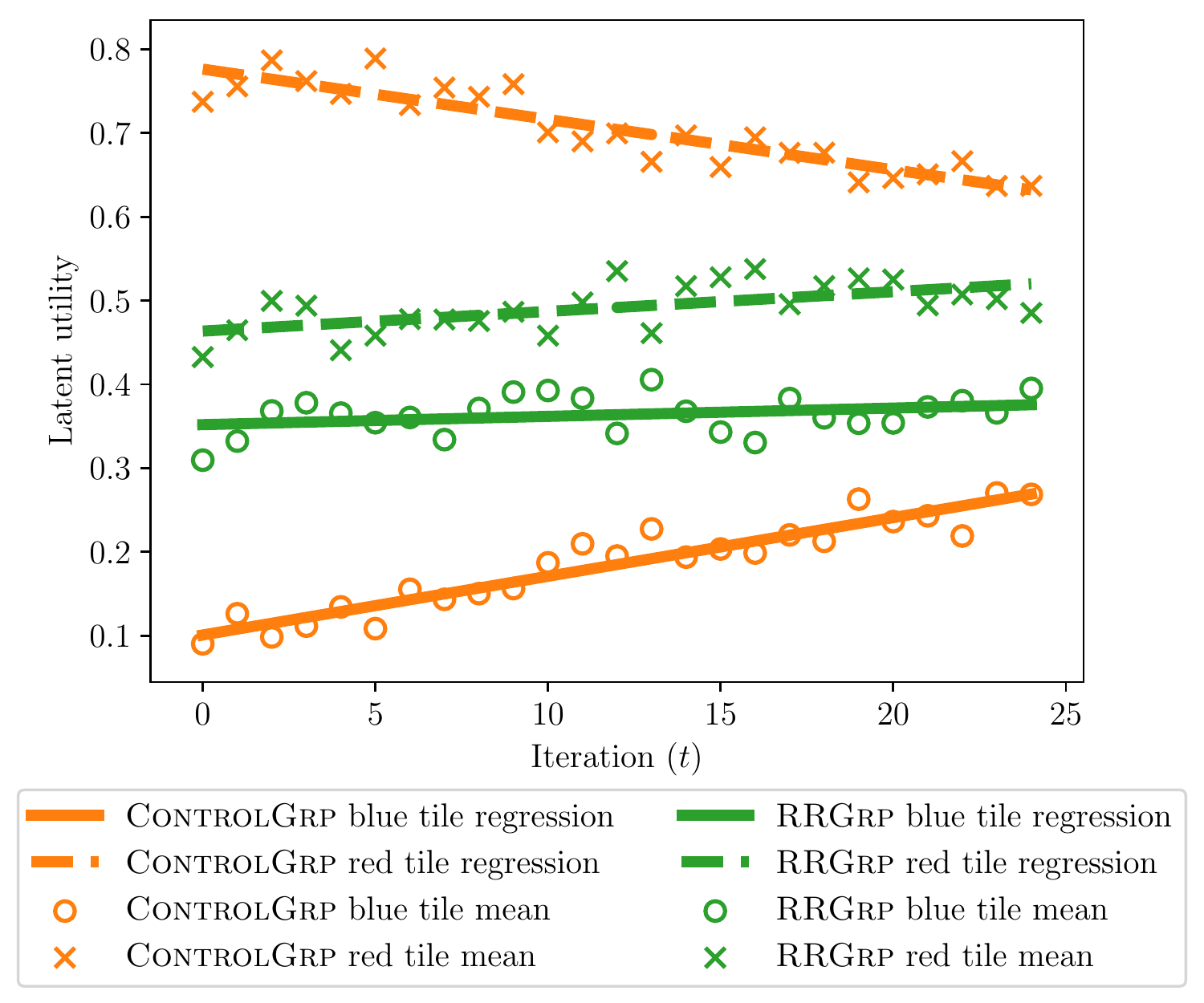}
    \caption{The latent utility derived from blue tiles was significantly greater for \rrgrp{} than for \unconsgrp{}.}
    \label{fig:latent_util_by_group_proportion2}
  \end{figure}

  \subsection{Optimal strategy overlap: alternate metric for optimality of selection decisions} \label{sec:optimal_strategy_overlap_details}

  Let \textit{optimal strategy overlap} be the number of selected tiles in the optimal strategy set as a proportion of $k$. Let \textit{optimal strategy overlap derived from blue (resp. red) tiles} mean the number of selected blue (resp. red) tiles in the optimal strategy set as a proportion of the number of blue (resp. red) tiles in the optimal strategy set.

  In this section, we conduct analysis similar to Section \ref{sec:latent_utility_details} on latent utilities.
  We find similar trends for optimal strategy overlap as for latent utilities.
  In Section \ref{sec:optimal_overlap_regression}, we analyze a linear regression model of the mean optimal strategy overlap attained by \rrgrp{} and \unconsgrp{}.
  We find that optimal strategy overlap modestly increases for both groups.
  In Section \ref{sec:optimal_strategy_overlap_from_blue}, we examine the mean optimal strategy overlap derived from blue tiles for \rrgrp{} and \unconsgrp{}.
  We find that \rrgrp{} attains higher optimal strategy overlap derived from blue tiles at significance level $p<0.001$.

  \subsubsection{Linear regression for optimal strategy overlap over iterations} \label{sec:optimal_overlap_regression}

  We fit a linear regression with outcome variable optimal strategy overlap and predictors iteration $t$ and Rooney Rule parameter $\ell$. We find that the optimal strategy overlap increases over iterations $t$. We observed a mean 20\% in the optimal strategy overlap over the twenty-five iterations for \rrgrp{} and 14\% for \unconsgrp{}.

  \begin{figure}[h!]
    \caption{Participants learn to attain greater optimal overlap over iterations.}\vspace{2mm}
    \centering

    \begin{tabular}{ r c c }
      \hline
      \textbf{\rrgrp{}} & Intercept & Iteration $t$  \\ \hline
      Estimated coefficient & 0.40 & 0.0032 \\
      Standard error & 0.018 & 0.001  \\
      $t$-statistic & 22 & 2.5  \\
      $p$-value & $<$0.001 & $<$0.05  \\ \hline
      \textbf{\unconsgrp{}} & Intercept & Iteration $t$  \\ \hline
      Estimated coefficient & 0.42 & 0.0023 \\
      Standard error & 0.015 & 0.001  \\
      $t$-statistic & 28 & 2.2  \\
      $p$-value & $<$0.001 & $<$0.05  \\
    \end{tabular}

    \label{tab:optimal_strategy_overlap_regression}
  \end{figure}

  \subsubsection{Optimal strategy overlap derived from blue tiles is significantly higher for \rrgrp{} than for \unconsgrp{}} \label{sec:optimal_strategy_overlap_from_blue}
  \begin{figure}[b!]
    \caption{\rrgrp{} and \unconsgrp{} attained the same mean optimal strategy overlap, but the proportion due to blue tiles is higher for \rrgrp{}.}\vspace{2mm}
    \centering
    \begin{tabular}{ r c c c }
      \multicolumn{4}{c}{Optimal strategy overlap mean \textit{(standard deviation)}} \\
      & overall & \begin{tabular}{@{}c@{}}derived from blue\end{tabular} & \begin{tabular}{@{}c@{}}derived from red\end{tabular} \\ \hline
      \rrgrp{} & 0.46 \textit{(0.29)} & 0.50 \textit{(0.33)} & 0.43 \textit{(0.45)} \\
      \unconsgrp{} & 0.46 \textit{(0.25)} & 0.30 \textit{(0.39)} & 0.63 \textit{(0.43)} \\
    \end{tabular}\vspace{2mm}
    \begin{tabular}{ r c c }
      & \begin{tabular}{@{}c@{}}\textbf{Welch's $t$-test:} comparison of \\ optimal strategy overlap means \\ derived from blue\end{tabular} \\ \hline
      Null hypothesis & population means equal \\
      Alternative hypothesis & population means different \\
      $t$-statistic & 9.5 \\
      Degrees of freedom & 1100 \\
      $p$-value & $<$0.001 \\
    \end{tabular}
    \label{tab:optimal_strategy_overlap}
  \end{figure}
  Although it is interesting to know that both \rrgrp{} and \unconsgrp{} learn to attain more optimal strategy overlap over time, we are more interested in seeing what participants from the two groups learned after a number of iterations. To this end, we compare the mean optimal strategy overlap of \rrgrp{} versus \unconsgrp{} during the last fifteen iterations.
  We find no statistically significant difference between the mean of \rrgrp{} and \unconsgrp{}. However, the proportion of selected blue tiles in the optimal strategy set is higher for \rrgrp{} than \unconsgrp{}, and we can reject the null hypothesis that means are equal with statistical significance $p < 0.001$. See Figure \ref{tab:optimal_strategy_overlap} for the means and standard deviations. We find similar results in the case of the proportion of latent utilities coming from $X$-tiles. We refer the reader to Section~\ref{sec:latent_utility_details} for details.
  \noindent

  In Figure \ref{fig:optimal_overlap_by_group}, we plot the mean proportions of blue and red tiles in the optimal strategy set that were selected in each iteration $t$. We also fit linear regressions to show the general trends.

  \begin{figure}[t!]
    \centering
    \includegraphics[width=\columnwidth/2]{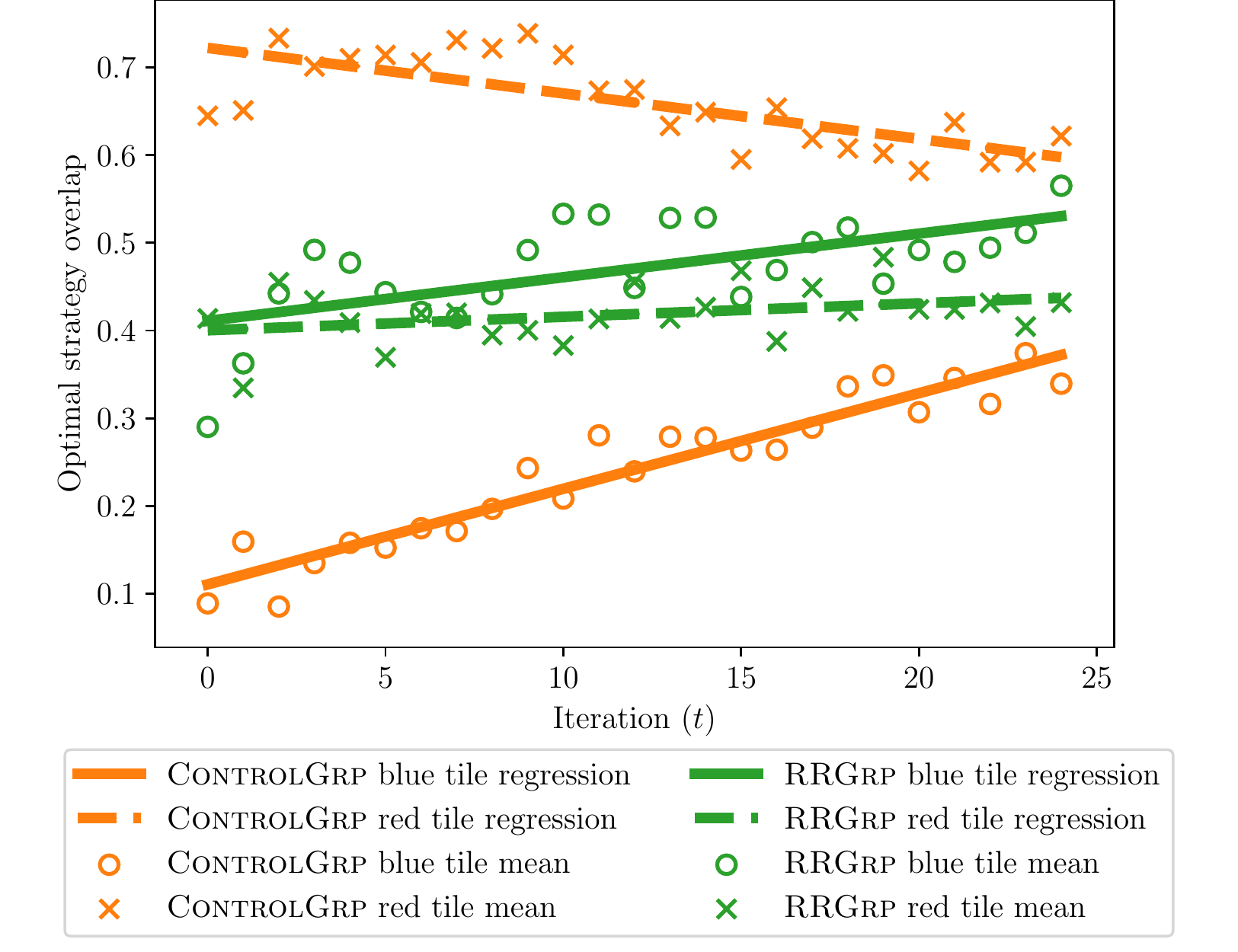}
    \caption{The number of selected blue tiles in the optimal strategy set is greater for \rrgrp{}.}

    \label{fig:optimal_overlap_by_group}
  \end{figure}

  \subsection{Figures for the experiment user interface} \label{sec:experiment_figures}
  In this section, we present the instructions, demonstration, and user interface shown to participants during the experiment. We refer the reader to Section \ref{sec:demo} for a detailed description of each component of the experiment. Figure \ref{fig:complete_instructions} shows the instructions provided to participants. Figures \ref{fig:complete_instructions_demo1} and \ref{fig:complete_instructions_demo2} show the demonstration participants completed before proceeding to the experiment. Figure \ref{fig:mturk_demo1} shows the experiment user interface before selections are made. Figure \ref{fig:mturk_demo2} shows the experiment user interface after $k=10$ tiles are selected but before the selections are submitted. Figure \ref{fig:mturk_demo3} show the experiment user interface after selections are submitted.

  \begin{figure*}[h!]
    \centering
    \includegraphics[width=\textwidth]{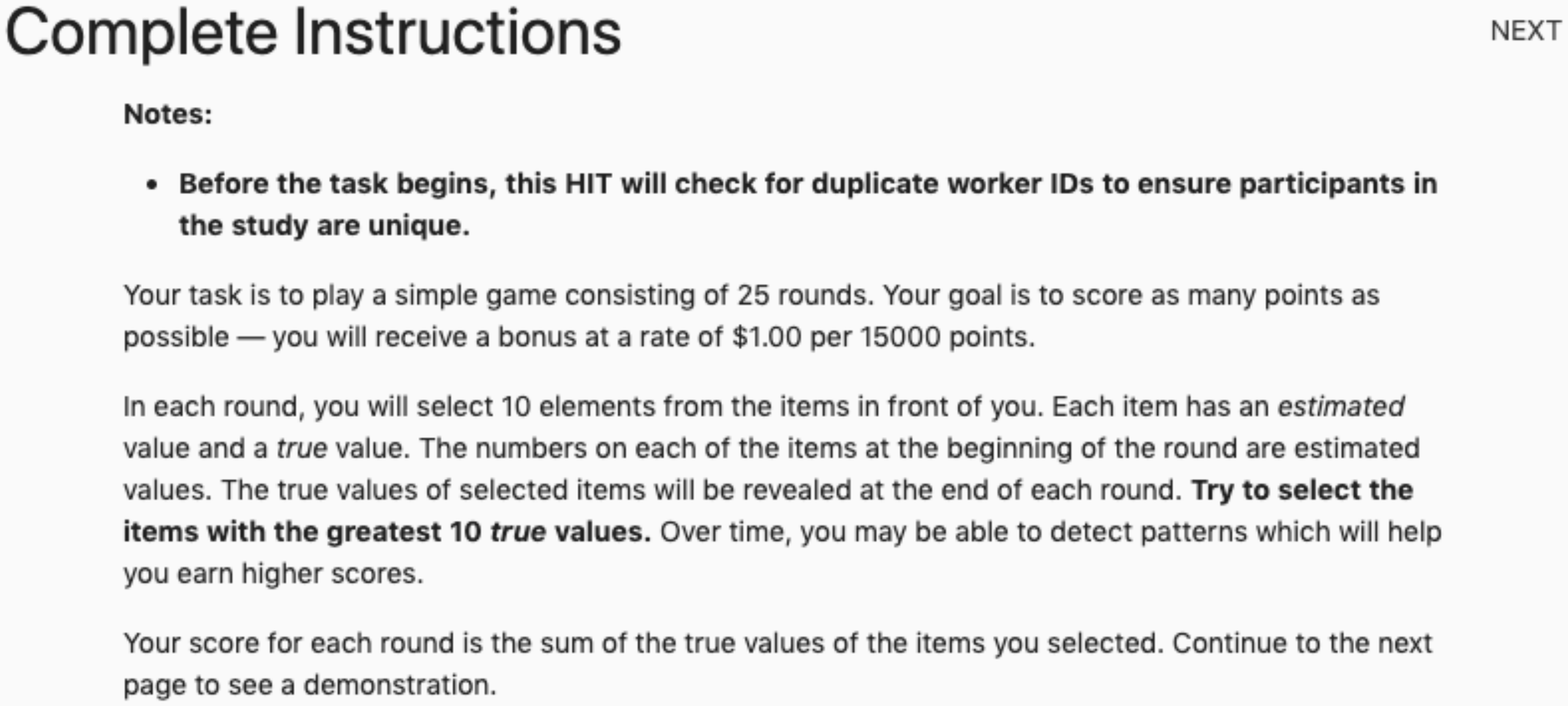}
    \caption{After participants provided informed consent, they were presented with complete instructions describing that each tile had an estimated and true value, and that they should try to maximize the true value of their selections.}
    \label{fig:complete_instructions}
  \end{figure*}

  \begin{figure*}
    \centering
    \includegraphics[width=\textwidth]{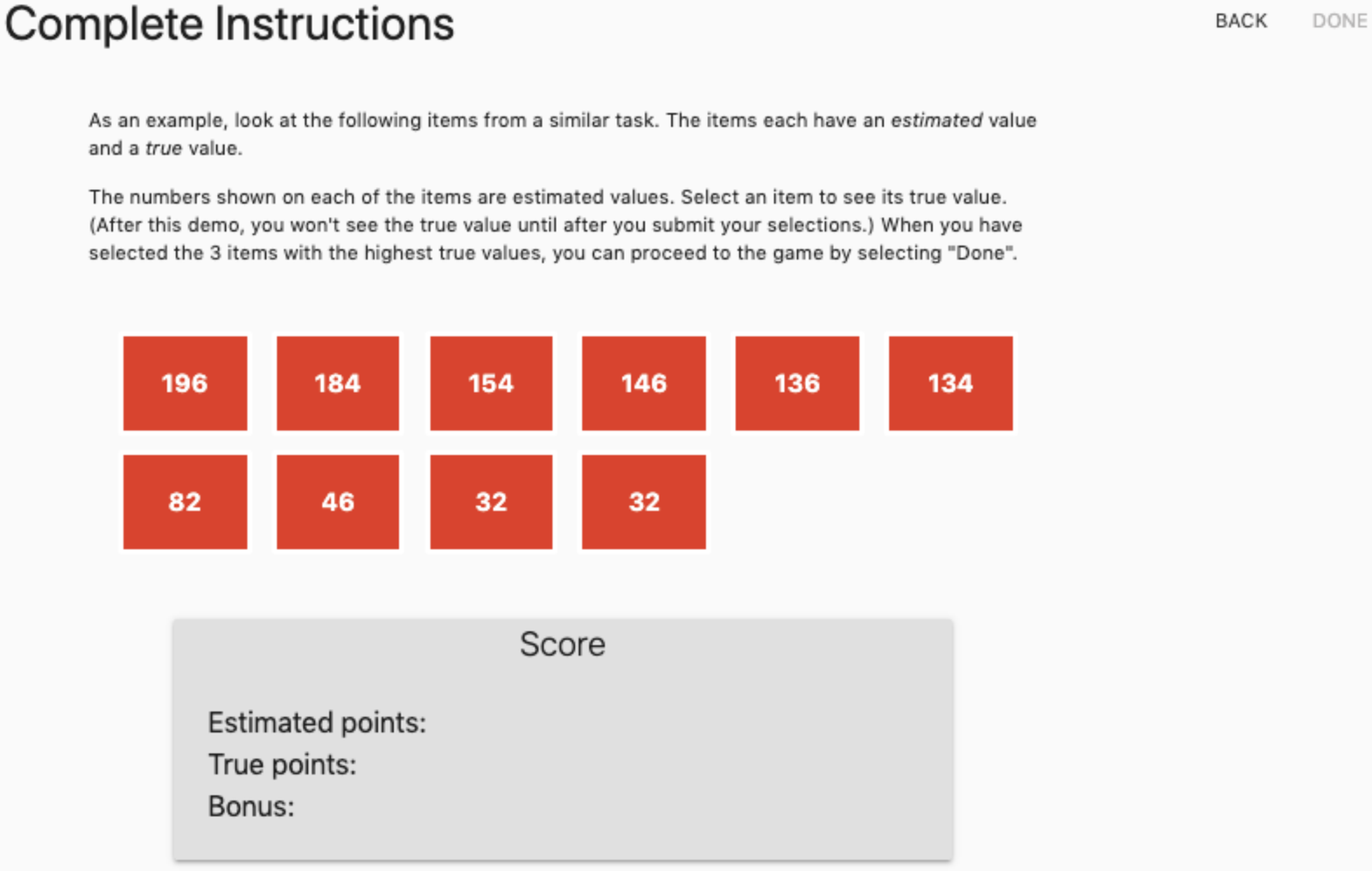}
    \caption{After the instructions but before starting the experiment, participants were shown a brief demonstration to provide intuition that the observed values were noisy estimates of the true values.}
    \label{fig:complete_instructions_demo1}
  \end{figure*}

  \begin{figure*}
    \centering
    \includegraphics[width=\textwidth]{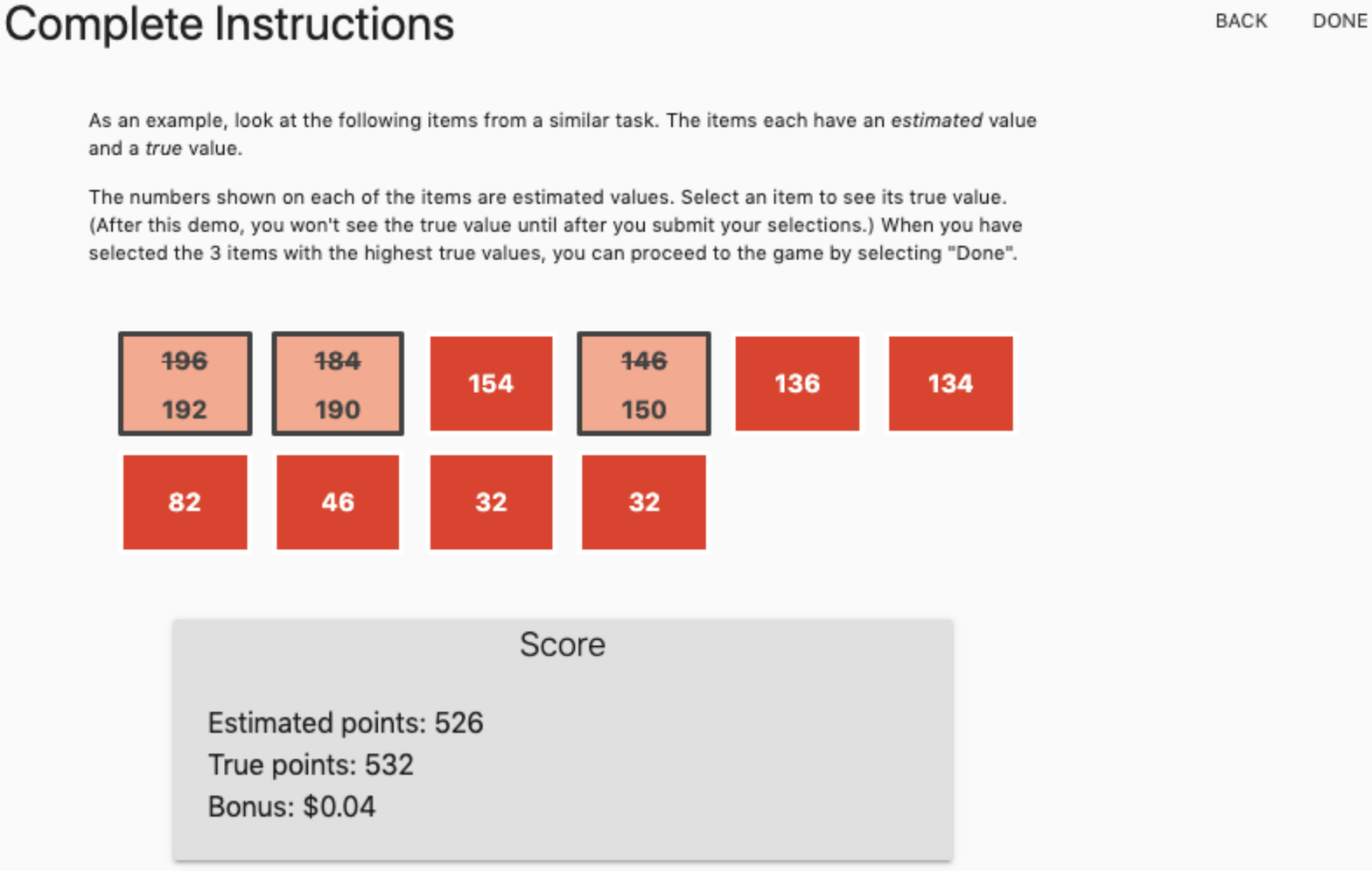}
    \caption{The first, second and fourth tiles in the ranking of estimated values were the tiles with the top three true values. Once participants selected these tiles, they were able to move on to the experiment.}
    \label{fig:complete_instructions_demo2}
  \end{figure*}

  \begin{figure*}
    \centering
    \includegraphics[width=\textwidth]{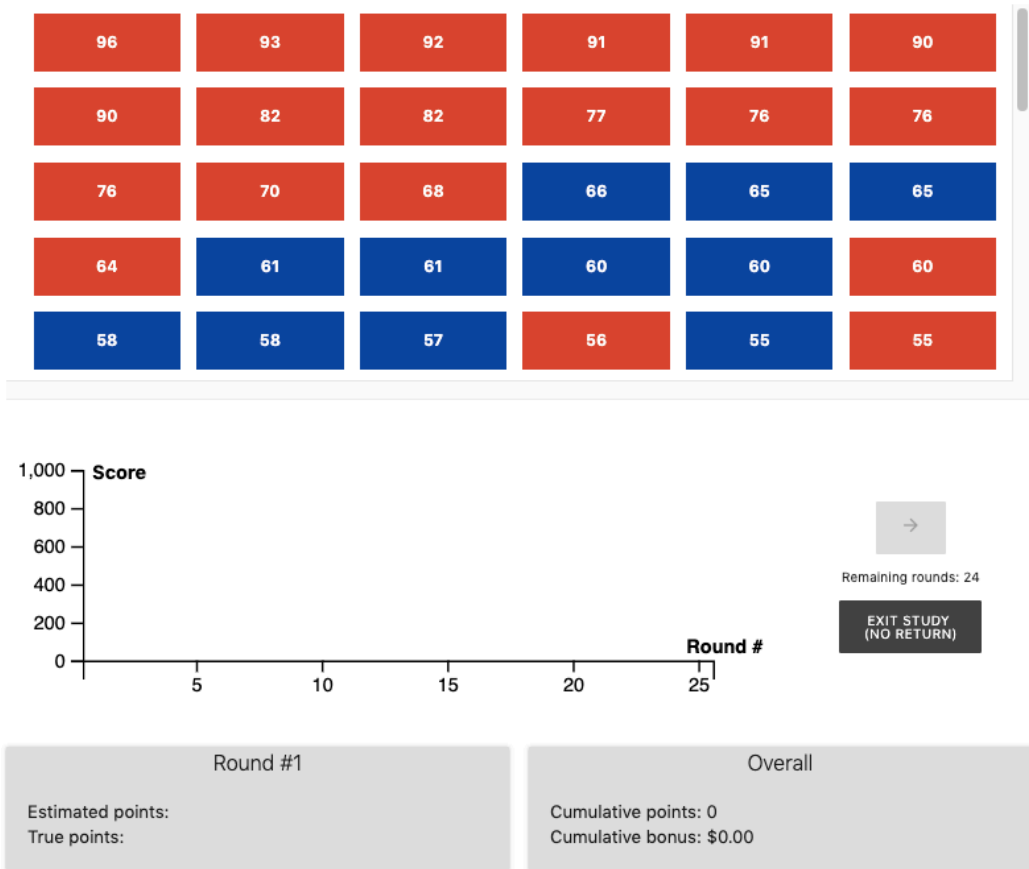}
    \caption{Before making their selections, participants were shown a list of tiles with their estimated values. Since the estimated values were biased against the blue group, blue tiles appeared lower down in the ranking.}
    \label{fig:mturk_demo1}
  \end{figure*}

  \begin{figure*}
    \centering
    \includegraphics[width=\textwidth]{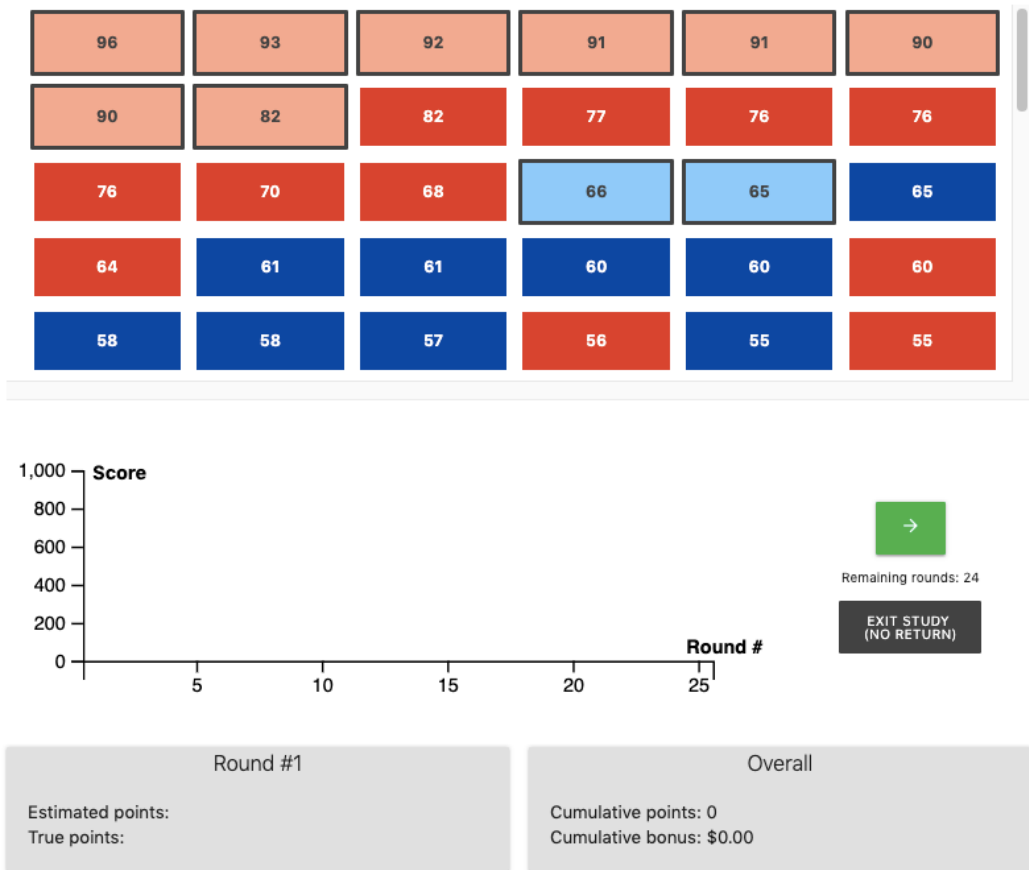}
    \caption{Participants selected tiles by clicking on them. Selected tiles appeared in a lighter shade with a grey outline.}
    \label{fig:mturk_demo2}
  \end{figure*}

  \begin{figure*}
    \centering
    \includegraphics[width=\textwidth]{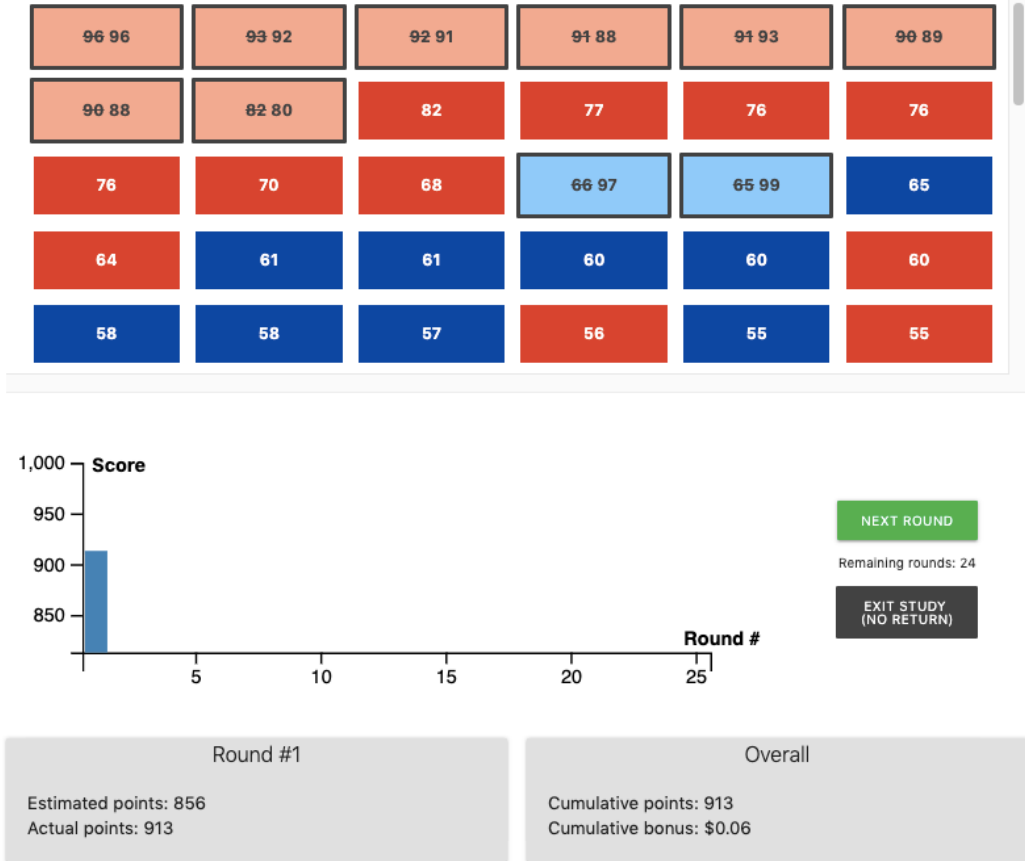}
    \caption{After participants made their selections, they were shown the estimated and true values of selected tiles. They were also shown the total estimated and actual points earned and their cumulative points and bonus payment.}
    \label{fig:mturk_demo3}
  \end{figure*}

  \end{document}